\documentclass[12pt]{article}

\usepackage[dvipdfmx]{graphicx}
\usepackage{amsmath}
\usepackage{ascmac}
\usepackage{graphicx}
\usepackage{mathrsfs}
\usepackage{stmaryrd}
\usepackage{amsmath,amssymb}
\usepackage{comment}
\usepackage[breaklinks=true]{hyperref}
\usepackage{breakcites}
\usepackage{latexsym}
\usepackage{cases}
\usepackage{url}
\usepackage{pifont}
\usepackage{multibib}
\usepackage{multirow}
\usepackage{setspace}
\usepackage{arydshln}
\usepackage{bigints}
\usepackage{here}
\usepackage{titling}
\usepackage{natbib}
\usepackage[english]{babel}
\usepackage{threeparttable}
\usepackage[a4paper]{geometry}
\usepackage{hyperref}
\usepackage{amsthm}
\hypersetup{
bookmarkstype=none,
colorlinks,
linkcolor=blue,
citecolor=black}
\usepackage{prettyref}
\usepackage{natbib}
\usepackage{threeparttable}
\allowdisplaybreaks[1]

\newcites{App}{References}

\usepackage{amsthm}
\newtheorem{theorem}{Theorem}[section]

\newtheorem{assumption}{Assumption}[section]

\newtheorem{remark}[theorem]{Remark}
\newtheorem{definition}[theorem]{Definition}
\newtheorem{lemma}[theorem]{Lemma}
\newtheorem{example}[theorem]{Example}

\newcommand{\MF}{\mathcal{F}}
\newcommand{\MD}{\mathcal{D}}
\newcommand{\MG}{\mathcal{G}}
\newcommand{\MZ}{\mathcal{Z}}
\newcommand{\MX}{\mathcal{X}}

\newcommand{\MP}{\mathcal{P}}

\newcommand{\MH}{\mathcal{H}}
\newcommand{\MB}{\mathcal{B}}

\newcommand{\Bz}{\mathbf{z}}
\newcommand{\Bb}{\mathbf{b}}
\newcommand{\Real}{\mathbb{R}}
\newcommand{\Natural}{\mathbb{N}}

\newcommand{\argmax}{\mathop{\rm arg~max}\limits}

\newcommand{\indep}{\perp \!\!\! \perp}

\makeatother

\begin{document}
\setstretch{1.2} 
\title{Estimation of Optimal Dynamic Treatment Assignment Rules under Policy Constraints\thanks{I would like to thank the editor, Stephane Bonhomme, and anonymous board member and referees for their constructive comments and suggestions. I am grateful to Toru Kitagawa, Aleksey Tetenov, Ryo Okui, Jeff Rowley, and participants in seminars at UCL and University of Tokyo, as well as those at the Cemmap/WISE Workshop on Advances in Econometrics in Xiamen, the 2019 Asian Meeting of the Econometric Society in Xiamen, and the 2020 World Congress of the Econometric Society for their comments and suggestions. This work was supported by JSPS KAKENHI Grant (number 22K20155) and ERC Grant (number 715940).}}
\author{Shosei Sakaguchi\thanks{Faculty of Economics, The University of Tokyo, 7-3-1 Hongo, Bunkyo-ku, Tokyo 113-0033, Japan. Email: sakaguchi@e.u-tokyo.ac.jp.} }
\date{\today}

\maketitle
\vspace{-1cm}

\begin{abstract} 
\begin{spacing}{1.2}
Many policies involve dynamics in their treatment assignments, where individuals receive sequential interventions over multiple stages.
We study estimation of an optimal dynamic treatment regime that guides the optimal treatment assignment for each individual at each stage based on their history. We propose an empirical welfare maximization approach in this dynamic framework, which estimates the optimal dynamic treatment regime using data from an experimental or quasi-experimental study while satisfying exogenous constraints on policies. The paper proposes two estimation methods: one solves the treatment assignment problem sequentially through backward induction, and the other solves the entire problem simultaneously across all stages. We establish finite-sample upper bounds on worst-case average welfare regrets for these methods and show their optimal $n^{-1/2}$ convergence rates. We also modify the simultaneous estimation method to accommodate intertemporal budget/capacity constraints. 
\end{spacing}
\medskip
\begin{spacing}{1.5}
\end{spacing}
\end{abstract}
\newpage{}

\setstretch{1.5} 

\section{Introduction \label{sec:introduction}}

Many policies involve dynamics in their treatment assignments.
Some policies assign a series of treatments to each individual across multiple stages, such as job training programs consisting of multiple stages \citep[e.g.,][]{Lechner_2009,Rodriguez_et_al_2018}.
Some policies are characterized by when to start/stop consecutive treatment assignment, such as unemployment insurance programs that reduce benefit level after a certain duration  \citep[e.g.,][]{Meyer_1995,Kolsrud_et_al_2018}. Examples of dynamic treatment assignments also include sequential medical interventions, educational programs, and marketing strategies. 

When implementing a dynamic policy, policymakers aim to optimize treatment assignments across multiple stages to maximize its social impact.
The effects of treatment at each stage are usually heterogeneous with respect to past treatments, intermediate outcomes, and individual characteristics. Hence, to maximize its social impact, treatment assignment to each individual at each stage should depend on the individual's accumulated information up to the corresponding stage.\footnote{
For example, in the context of a sequential job training program, the interest is in which training regimen to be assigned to each individual at each stage depending on their history of prior training participation, associated labor outcomes, and other observed characteristics. An important question in the unemployment insurance policy context is when and to whom to reduce the insurance level, given a recipient's characteristics and past effort toward a job search.
}

This paper proposes a statistical decision approach to solve dynamic treatment choice problems using data from an experimental or quasi-experimental study. 
We assume dynamic unconfoundedness \citep{Robins_1997}, 
meaning that the treatment assignment at each stage is independent of current and future potential outcomes given the history of treatment assignments and state variables. 
Under this assumption, we construct an approach to estimate the optimal dynamic treatment regime (DTR)\footnote{Borrowing from the terminology of statistics literature, we call the dynamic treatment assignment rule DTR.} building upon the concept of empirical welfare maximization (EWM) \citep{Kitagawa_Tetenov_2018a}. 
We call it the dynamic empirical welfare maximization (DEWM). The DEWM approach estimates the optimal DTR by maximizing empirical welfare, a sample mean of propensity-score-weighted outcomes, over a pre-specified class of feasible DTRs. True (estimated) propensity scores are used in the experimental (observational) data setting.

When designing public policy, considering external constraints such as interpretability or fairness in treatment allocation is crucial.  
The DEWM approach offers favorable features to accommodate exogenous policy constraints by restricting the class of feasible DTRs. 
Moreover, it can be applied to various dynamic treatment choice problems, such as optimal starting and stopping problems, by properly constraining the class of DTRs.

We present two approaches to estimate the optimal DTR. The first estimates the optimal DTR through backward induction, which solves the treatment choice problem from the final to the first stage, supposing at each stage that the optimal treatments are chosen in future stages. The second estimates the optimal DTR simultaneously across all stages, solving the empirical welfare maximization problem at once for the entire DTR.\footnote{Without specifying the direct and indirect effects of the treatment on future outcomes, each approach accounts for these effects within its EWM process.}

We reveal that the two approaches complement each other. The backward estimation method is computationally efficient; however, its consistency is ensured only when a pre-specified class of DTRs contains the first-best rule that assigns the best treatment for any history at each stage (except the first stage). 
Conversely, the simultaneous estimation method consistently estimates the optimal DTR on a pre-specified class of DTRs, irrespective of the feasibility of the first-best rule, at the cost of computational efficiency.

In practical terms, dynamic policies often impose budget or capacity constraints on treatment allocation over time. An ideal DTR should allocate limited resources effectively across stages to maximize welfare. We extend the simultaneous estimation method to problems with intertemporal budget/capacity constraints. We show that the resulting DTR approximately maximizes welfare while satisfying these constraints.

We evaluate the statistical properties of the DEWM approaches in terms of average welfare regret.\footnote{The average welfare regret is the average welfare loss relative to the maximum welfare achievable in the pre-specified class of DTRs.}
We derive finite-sample and distribution-free upper bounds on the average welfare regret of the DTR estimated by each of the backward-induction and simultaneous optimization methods. The resulting bounds depend on the sample size $n$ and a measure of complexity of the class of DTRs. Our main theorem shows that the average welfare regret for each method converges to zero at rate $n^{-1/2}$ in the experimental data setting. Furthermore, we show that this convergence rate is optimal.\footnote{To my knowledge, this is the first work to formally show the minimax rate optimality of welfare regrets in estimating optimal DTRs.}
For the budget/capacity constrained problem, we also analyze the excess implementation cost of the estimated DTR relative to the actual budget/capacity. We derive finite-sample and distribution-free upper bounds on both the welfare regret and the excess cost of the estimated DTR.

\subsection*{Related Literature \label{sec:related literature}}

This paper contributes to the literature on statistical decision of treatment choice, although much of existing work focuses on the static problem.\footnote{A partial list of works in that literature includes \citet{Manski_2004}, \citet{Dehejia_2005}, \citet{Hirano_Porter_2009},
\citet{Stoye_2009,Stoye_2012},
\citet{Bhattacharya_Dupas_2012}, \citet{Chamberlain_2012},
\citet{Tetenov_2012}, \citet{Kitagawa_Tetenov_2018a}, \citet{Athey_Wager_2020}, \cite{Kitagawa_Tetenov_2018b}, \citet{Mbakop_Tabord-Meehan_2019}, and \cite{kitagawa_et_al_2021}.}
Policy learning methods by \cite{Kitagawa_Tetenov_2018a}, \cite{Athey_Wager_2020}, and \cite{Mbakop_Tabord-Meehan_2019} build on the similarity of the empirical welfare maximizing treatment choice and the empirical risk-minimizing classification. 
\cite{Athey_Wager_2020} apply doubly robust estimators to static policy learning, and show that an $n^{-1/2}$-asymptotic upper bound on regret can be achieved even in the observational data setting.  

In the dynamic treatment framework, \citet{Han_2018} relaxes the sequential randomization assumption, allowing for noncompliance, and studies point identification of the average dynamic treatment effects and optimal non-additive DTR. 
\citet{Han_2019} proposes a method to characterize the sharp partial ordering of the counterfactual welfares of DTRs in an instrumental variable setting. 
\citet{Heckman_Navarro_2007} and \citet{Heckman_et_al_2016} use exclusion restrictions to identify the dynamic treatment effect, but their focus do not extend to the identification of the optimal DTRs.

Estimation of the optimal DTRs has been widely studied in the biostatistics and statistics literature.\footnote{\citet{Chakraborty_Moodie_2013}, \citet{Chakraborty_Murphy_2014}, \citet{Laber_et_al_2014}, and \cite{Tsiatis_et_al_2019} review the developments in this field.} Some dominant approaches exist, such as G-estimation \citep{Robins_1989,Robins_et_al_1992} and Q-learning \citep{Murphy_2005,Moodie_et_al_2012}. 
A potential drawback of these approaches is the risk of misspecification of the models relevant to the counterfactual outcomes. By contrast, the DEWM approach does not need to specify any model relevant to the counterfactual outcomes.

Building on the similarity between treatment choice and classification, 
\citet{Zhao_et_al_2015} develop estimation methods for the optimal DTRs using the support vector machine with propensity score weighted outcomes. Their approach is computationally efficient because it uses a convex surrogate loss. However, it cannot accommodate exogenous constraints on a class of DTRs.\footnote{The hinge loss approach in \citet{Zhao_et_al_2015} loses consistency and computational efficiency, for example, under budget or fairness constraints. Moreover, \cite{Laha_et_al_2024} show that using a smooth convex surrogate loss or hinge loss in the simultaneous maximization approach can fail to consistently estimate the optimal DTRs.} In contrast, our focus lies on estimating the optimal DTRs with exogenous constraints on the class of DTRs, a scenario more commonly encountered in public policy-making.\footnote{In the static setting, \cite{kitagawa_et_al_2021} show that the surrogate hinge loss approach has consistency in constrained treatment choice problems, which could be extended to our dynamic setting. Their main result for consistency applies when constraints are imposed on the level set of a treatment rule, whereas we consider more general constraints on the functional form of treatment rules.}
Beyond the results of \cite{Zhao_et_al_2015}, we reveal a tradeoff between imposing constraints on the dynamic treatment choice and the consistency of the backward-induction approach, and formally show the minimax rate optimality of the proposed methods in the context of dynamic treatment choice. 

This work is also related to the literature on optimal stopping  \citep[e.g.,][]{van_1976,Rust_1987,Jacka_1991,Goel_et_al_2017,Nie_et_al_2019}.  Most works in the literature rely either on a known stochastic model \citep{van_1976,Rust_1987,Jacka_1991} or on a generator of system dynamics \citep{Goel_et_al_2017}.\footnote{\citet{Nie_et_al_2019} propose doubly robust estimation method for the optimal stopping/starting problem.} 
The methods proposed in our study can estimate the optimal stopping/starting policies from batch data by properly specifying the class of DTRs.

Finally, the dynamic treatment framework we study differs from the bandit problem, for example, studied by \cite{Kock_Thyrsgaard_2018}. In the bandit problem, different individuals receive treatment at different stages. By contrast, in our dynamic framework, the same individuals progress through different stages and receive sequential treatment interventions across these stages. Additionally, in bandit problems, the treatment effect is explored and exploited across sequential stages, whereas, in our framework, the effects of sequential treatments are estimated before the allocation task.\footnote{The bandit problem is an online learning problem, whereas we study an off-line learning problem.}\footnote{\cite{Kallus_2020} study the bandit problem with DTRs, considering the problem of developing and exploiting the optimal DTR in an online setting.}

\subsection*{Structure of the Paper }
The remainder of this paper proceeds as follows. Section \ref{sec:setup} defines the dynamic treatment choice problem. Section \ref{sec:DEWM}
presents the two DEWM methods and shows their statistical properties. Section \ref{sec:budget constraint} extends the simultaneous estimation method to accommodate intertemporal budget/capacity constraints. Section \ref{sec:estimated propensity score} proposes estimation methods for the observational data setting. Section \ref{sec:simulation} shows the results of a simulation study. 
In Section \ref{sec:empirical application}, we apply the proposed methods to the Project STAR (Steps to Achieving Resilience) data, where we estimate an optimal DTR to allocate each student to a class with or without a teacher aide in multiple grades. Section \ref{sec:conclusion} concludes this paper. 

\section{Setup\label{sec:setup}}

Section \ref{sec:dynamic treatment framework} introduces the dynamic treatment framework, following Robins's dynamic counterfactual outcomes framework \citep{Robins_1986,Robins_1997}. Subsequently, we define the dynamic treatment choice problem in Section \ref{sec:dynamic treatment choice problem}. In this study, we denote by $E_{P}\left[\cdot\right]$ the expectation with respect to a distribution function $P$.

\subsection{Dynamic Treatment Framework \label{sec:dynamic treatment framework}}

We suppose $T$ $(T < \infty)$ stages of binary treatment assignment. 
Let $D_{t}\in \left\{ 0,1\right\}$, for $t=1,\ldots,T$, denote the binary
treatment at stage $t$. At the end of each stage $t$, we observe an outcome $Y_{t}$. Let $X_{t}$ be a $k$-dimensional vector of covariates observed before treatment assignment at stage $t$. The distribution of
$X_{t}$ may depend on past treatments, outcomes, and covariates.
$X_{1}$ represents pre-treatment information,
containing individuals' demographic characteristics observed before policy implementation.  Throughout this paper, for any time-dependent object $A_{t}$,
we denote by $\text{\ensuremath{\text{\ensuremath{\underline{A}}}}}_{t}\equiv \left(A_{1},\ldots,A_{t}\right)$
a history of the object up to stage $t$, and denote by $\text{\ensuremath{\underline{A}}}_{s:t} \equiv \left(A_{s},\ldots,A_{t}\right)$,
for $s\leq t$, a partial history of the object from stage $s$ up to stage $t$. For example, the treatment history up to stage $t$ is denoted by $\text{\ensuremath{\underline{D}}}_{t}=\left(D_{1},\ldots,D_{t}\right)$. 
Let $Z \equiv \left(\underline{D}_T,\underline{Y}_T,\underline{X}_T\right)$ be the vector containing all observed variables. We define the history in stage $t$ by $H_{t} \equiv \left(\underline{D}_{t-1},\underline{Y}_{t-1},\underline{X}_t\right)$, which is available information for the policymaker when she chooses
a treatment assignment at stage $t$. Note that $H_{s}\subseteq H_{t}$
for any $s\leq t$, and $H_{1}=\left(X_{1}\right)$. We denote
the support of $H_{t}$ and  $Z$ by ${\cal H}_{t}$ and $\MZ$, respectively.

We illustrate the dynamic treatment framework with an example of a sequential job training from \cite{Rodriguez_et_al_2018}. They study the effect of sequential training in Chile's “Franquicia Tributaria” program, where a worker can sequentially participate in multiple training sessions. They consider two stages ($T=2$) with ``$D_1=1$'' and ``$D_2=1$'' indicating participation in the first and second stages, respectively. $Y_1$ and $Y_2$ are the monthly salaries observed after training for each stage. $X_1$ includes age, gender, initial wage, and education variables, while there are no time-varying covariates $X_2$.

To formalize our results, we employ the framework of dynamic potential outcomes \citep{Robins_1986,Murphy_2003}. 
Let $Y_{t}\left(\text{\ensuremath{\underline{d}}}_{t}\right)$ denote the potential outcome of $\underline{d}_t \in \{0,1\}^t$ at stage $t$, representing the outcome for stage $t$ that is realized when the history of treatment up to stage $t$ coincides with $\underline{d}_{t}$. We implicitly assume that the potential outcomes are not influenced by future treatments, that is, a no-anticipation condition. 
Given that the covariates $X_t$ may be influenced by past treatments, we define potential covariates as $X_t(\underline{d}_{t-1})$ for each $t \geq 2$ and $\underline{d}_{t-1} \in \{0,1\}^{t-1}$. We denote $X_{1}(\underline{d}_{0})=X_1$ when $t=1$.
The observed outcomes and covariates are defined as $Y_t \equiv Y_t(\underline{D}_t)$ and $X_t \equiv X_t(\underline{D}_{t-1})$, respectively. Denoting $\underline{Y}_{t}(\underline{d}_{t}) \equiv (Y_1(d_1),\ldots,\underline{Y}_{t}(\underline{d}_{t}))$ and $\underline{X}_{t}(\underline{d}_{t-1}) \equiv (X_1,X_2(d_1),\ldots,\underline{X}_{t}(\underline{d}_{t-1})$, a vector $H_t(\underline{d}_{t-1}) \equiv \left(\underline{d}_{t-1},\underline{Y}_{t-1}(\underline{d}_{t-1}),\underline{X}_{t}(\underline{d}_{t-1})\right)$ represents the potential history that is realized when prior treatments are $\underline{d}_{t-1}$. We denote $H_{1}(\underline{d}_{0}) = H_1$ when $t=1$.  The observed history is defined as $H_{t} \equiv H_t(\underline{D}_{t-1})$.
Let $P$ be the distribution of all underlying variables $\left(\underline{D}_T,\{\underline{Y}_{T}(\underline{d}_{T})\}_{\underline{d}_{T} \in \{0,1\}^{T}},\{\underline{X}_{T}(\underline{d}_{T-1})\}_{\underline{d}_{T-1} \in \{0,1\}^{T-1}}\right)$.

From an experimental or observational study, we observe $Z_{i}\equiv \left(D_{it},Y_{it},X_{it}\right)_{t=1}^{T}$ for individuals $i=1,\ldots,n$, where
$Y_{it}  \equiv Y_{it}(\underline{D}_{it})$ and $X_{it}  \equiv X_{it}(\underline{D}_{i,t-1})$
with $Y_{it}(\underline{d}_{t})$ and $X_{it}(\underline{d}_{t-1})$ being a potential outcome and covariates for individual $i$ at stage $t$. We suppose that the vectors of underlying random variables $V_i \equiv \left(\underline{D}_{iT},\{\underline{Y}_{iT}(\underline{d}_{T})\}_{\underline{d}_{T} \in \{0,1\}^{T}},\{\underline{X}_{iT}(\underline{d}_{T-1})\}_{\underline{d}_{T-1} \in \{0,1\}^{T-1}}\right)$,
$i=1,\ldots,n$, are independent and identically distributed (i.i.d) with the distribution $P$. We denote by $P^{n}$ the joint distribution of $\left\{V_i:i=1,\ldots,n\right\}$.

Let $e_{t}\left(d_{t},h_{t}\right)\equiv \Pr\left(D_{t}=d_{t}\mid H_{t}=h_{t}\right)$
be a propensity score of treatment at stage $t$ given the history up to that point. We suppose that the propensity scores are known in the experimental study but are unknown in the observational study.
These settings are considered in Sections \ref{sec:DEWM}-\ref{sec:budget constraint} and Section \ref{sec:estimated propensity score}, respectively.

In this study, we suppose that the following assumptions hold.

\medskip

\begin{assumption}[Sequential Independence Assumption]\label{asm:sequential independence} For any $t=1,\ldots,T$ and $\text{\ensuremath{\underline{d}}}_{T}\in\left\{ 0,1\right\} ^{T}$,
$
\left(Y_{t}\left(\text{\ensuremath{\underline{d}}}_{t}\right),\dots,Y_{T}\left(\text{\ensuremath{\underline{d}}}_{T}\right),X_{t+1}\left(\underline{d}_t\right),\ldots,X_{T}\left(\underline{d}_{T-1}\right)\right) \indep D_{t} \mid H_{t} \mbox{\ a.s.}
$
\end{assumption}
\begin{assumption}[Bounded Outcomes]\label{asm:bounded outcome}
There exists $M_{t}<\infty$ such that the support of $Y_{t}$ is
contained in $\left[-M_{t}/2,M_{t}/2\right]$ for $t=1,\ldots,T$.
\end{assumption}

\medskip

Assumption
\ref{asm:sequential independence} is known as a dynamic unconfoundedness
assumption or sequential/dynamic conditional independence assumption elsewhere,
and is commonly used in the literature on dynamic treatment
effect analysis \citep{Robins_1997,Murphy_2003}.
This assumption means that the treatment assignment at each stage is independent of the current and future potential outcomes and future covariates conditional on the history up to that point.
This is typically satisfied in sequential randomization experiments. In observational studies, this assumption is often controversial but can be satisfied if a sufficient set of confounders is available. Assumption \ref{asm:bounded outcome} is a common assumption in the literature on statistical treatment choice \citep[e.g.,][]{Manski_2004,Stoye_2009,Kitagawa_Tetenov_2018a}.

\subsection{Dynamic Treatment Choice Problem\label{sec:dynamic treatment choice problem}}

We aim to develop methods to estimate the optimal DTRs from experimental or observational data with sequential treatment assignment. We denote a treatment rule for each stage $t$ by $g_{t}:{\cal H}_{t}\mapsto\left\{ 0,1\right\} $, a map from the history up to stage $t$ to a binary treatment. We define the DTR by $g\equiv\left(g_{1},\ldots,g_{T}\right)$,
a sequence of stage-specific treatment rules. The DTR guides policymakers in selecting treatment for each individual at each stage based on their history up to that point. 

We define the counterfactual outcome of a sequence of treatment rules $\underline{g}_t$ for each stage $t$ as $\widetilde{Y}_t\left(\underline{g}_t \right) \equiv \sum_{\underline{d}_{t} \in \{0,1\}^t} Y_t(\underline{d}_t)\cdot \prod_{s=1}^{t}1\left\{g_s\left(H_{s}(\underline{d}_{s-1})\right) = d_s\right\}$.
This is the counterfactual outcome for stage $t$ that is realized when the sequential treatment assignment up to stage $t$ follows the sequence of treatment rules $\underline{g}_t$.

We then define the welfare of a DTR $g$ by the population mean of a weighted sum of outcomes as follows: 
\begin{align}
W\left(g\right) &\equiv  E_{P}\left[\sum_{t=1}^{T}\gamma_{t}\widetilde{Y}_{t}\left(\underline{g}_{t}\right)\right]  =  \sum_{t=1}^{T}E_{P}\left[\gamma_{t}\widetilde{Y}_{t}\left(\underline{g}_{t}\right)\right], \label{eq:additive_welfare_function}
\end{align}
where the weight
$\gamma_{t}$, for $t=1,\ldots,T$, lies in $\left[0,1\right]$ and
is chosen by the policy-maker. If the policymaker targets a time-discounted welfare, the weight at each stage is $\gamma_{t}=\gamma^{T-t}$ with
$\gamma$ being a time-discount factor that lies in $\left(0,1\right)$. If the policymaker targets the outcome for the last stage only, $\gamma_{T}=1$ and $\gamma_{t}=0$ for all $t \neq T$.

Given the propensity scores $\left\{ e_{t}\left(d_{t},h_{t}\right)\right\} _{t=1}^{T}$
and under Assumption \ref{asm:sequential independence}, the welfare function can be identified by the observables only:
\begin{align}
W\left(g\right) & =\sum_{t=1}^{T}E_{P}\left[\frac{\left(\prod_{s=1}^{t}1\left\{ D_{s} = g_{s}\left(H_{s}\right)\right\} \right)\gamma_{t}Y_{t}}{\prod_{s=1}^{t}e_{s}\left(D_{s},H_{s}\right)}\right].\label{eq:original welfare}
\end{align}

We suppose that the policymaker chooses a DTR from a pre-specified class of feasible DTRs, denoted by ${\cal G}\equiv{\cal G}_{1}\times\cdots\times{\cal G}_{T}$, where ${\cal G}_{t}$ is a class of feasible treatment rules at stage
$t$ (i.e., a class of measurable functions $g_{t}:\MH_{t}\rightarrow \{0,1\}$). Therefore, the ultimate goal of the analysis is to choose an optimal DTR that maximizes the welfare function $W\left(\cdot\right)$ over $\mathcal{G}$. 
\footnote{
In the context of sequential job training \citep{Lechner_2009,Rodriguez_et_al_2018}, $g_t(h_t)$ decides whether an individual with history $h_t$ should receive job training at stage $t$. The history $h_t$ may include information on past trainings, pre-training and intermediate wages, and educational backgrounds. When $Y_t$ represents the wage at stage $t$, the optimal DTR is the optimal sequence of treatment rules for determining participation in job training at each stage, to maximize the population mean of the total weighted wages $W(g)$.
}

In this study, we constrain the complexity of the class of feasible DTRs in terms of VC-dimension.\footnote{The definition of VC-dimension is given in Definition \ref{def:VC-dimension of a class of binary functions} in Appendix along with some examples.}
The following assumption restricts the complexity of the class of feasible DTRs $\MG$ in terms of the VC-dimension of $\MG_t$ for each $t=1,\ldots,T$.
\medskip

\begin{assumption}[VC-class]\label{asm:vc-class} 
The class of feasible DTRs $\MG$ has the form of $\MG=\MG_1 \times \cdots \times \MG_T$. For $t=1,\ldots,T$, ${\cal G}_{t}$ is a VC-class of functions
and has VC-dimension $v_{t}<\infty$. 
\end{assumption}

\medskip

This assumption restricts the complexity of the class of DTRs $\MG$ by restricting the class of feasible treatment rules $\MG_{t}$ for each specific stage. By restricting the complexity, we can select a DTR that is simple to explain/interpret DTR, and can keep estimated DTRs from overfitting the data. Although Assumption \ref{asm:vc-class} excludes nonparametric classes of $\MG_t$, our framework can accommodate nonparametric approaches by appropriately controlling the growth rate of VC-dimension with the sample size $n$. 

We can incorporate arbitrary exogenous policy constraints for ethical or political reasons into DTRs by specifying the form of $\MG_{t}$ for each $t$.\footnote{Although a treatment rule $g_{t}(h_t)$ depends on the full-history of covariates $\underline{x}_{t}$ from stage 1 to t, we can also consider treatment rules that do not depend on the past covariates $\underline{x}_{t-1}$ by restricting the class $\MG_t$ such that for any $g_t \in \MG_t$, $g_{t}(h_t) = g_{t}(h_{t}^\prime)$ for any $h_t$ and $h_{t}^{\prime}$ such that $h_{t} \backslash \underline{x}_{t-1} = h_{t}^{\prime} \backslash \underline{x}_{t-1}^{\prime}$. Similar constraints can also be imposed for the treatment history $\underline{d}_t$ and outcome history $\underline{y}_{t}$.}
Some examples of practically relevant classes of DTRs are linear treatment rules and decision tree rules.

Aside from the constraint on the functional form, we can specify various dynamic treatment choice problem by restricting the intertemporal relationship of treatment rules across stages. Some examples are as follows.

\medskip

\begin{example}[Optimal Starting/Stopping Problem]
 If the policymaker aims to decide when to start consecutive treatment assignments for each individual, the restriction $d_{s}\leq g_{t}(\cdot)$ for all $s\leq t$
should be imposed on ${\cal G}_{t}$. Similarly, the problem of deciding when to stop consecutive treatment assignments can be specified by imposing the restriction $d_{s}\geq g_{t}(\cdot)$ on ${\cal G}_{t}$ for all $s\leq t$.
\end{example}

\smallskip

\begin{example}[One-Shot Treatment]
If the problem is to decide when to assign a one-shot treatment to each individual, the analyst should impose the restriction $\sum_{s=1}^{t-1}d_{s}+g_{t}(\cdot)\leq1$
on ${\cal G}_{t}$ for each $t$.
\end{example}

\smallskip
\noindent
The VC-dimension of an additionally restricted class does not exceed that of the original class.

Given a class of feasible DTRs $\MG$, we assume the following overlap condition holds for the propensity scores $\{e_{t}(d_{t},h_{t})\}_{t=1}^{T}$.

\medskip

\begin{assumption}[Overlap Condition]\label{asm:overlap} 
For $t=1,\ldots,T$, there exists $\kappa_{t} \in (0,1)$ for which $\kappa_{t} \leq e_{t}(d_{t},h_{t})$ holds for any pair $(d_{t},h_{t}) \in \{0,1\}\times \MH_{t}$ such that there exists $g_{t} \in \MG_{t}$ that satisfies $g_{t}(h_{t})=d_{t}$. 
\end{assumption}

\medskip

When $\MG$ is structurally constrained, Assumption \ref{asm:overlap} is weaker than a common overlap condition that requires the overlap $e_{t}(d_{t},h_{t}) \in (0,1)$ for all $(d_{t},h_{t})\in \{0,1\}\times \MH_{t}$ and $t=1,\ldots,T$.\footnote{For example, in the optimal stopping problem, Assumption \ref{asm:overlap} does not require $e_{t}(1,h_{t})>0$ for any $h_t$ such that $d_s$ in $h_t$ is equal to $0$ for some $s < t$.}   
This assumption also guides how to design experiments given $\MG$; that is, in an experiment, the treatment $d_t$ does not need to be assigned to individuals with any $h_t$ such that $d_t$ is not achievable by $g_t(h_t)$ for any $g_t \in \MG_t$ (i.e., $d_t \neq g_{t}(h_t)$ for any $g_t \in \MG_t$).\footnote{For example, in the optimal stopping problem, $d_t = 1$ does not need to be assigned to any individuals who were already untreated (i.e., individuals with $d_s = 0$ for some $s < t$).}
Assumption \ref{asm:overlap} is satisfied in the experimental data setting, for example, when the treatment $D_t$ is randomly assigned without any dependence on the history $H_t$.

We denote the highest welfare that is attainable in the class of feasible DTRs ${\cal G}$ by
\begin{align}
W_{{\cal G}}^{\ast}\equiv  \max_{g\in{\cal G}}W\left(g\right).\label{eq:welfare maximization}
\end{align}
We consider estimating the optimal DTR  that maximizes the welfare $W(\cdot)$ over $\MG$ from the sample $\left\{ Z_{i}:i=1,\ldots,n\right\} $. In the subsequent section, we present two methods to estimate the optimal DTR, and show their statistical properties.


\section{Dynamic Empirical Welfare Maximization \label{sec:DEWM}}

This section proposes two DEWM methods. One method employs backward induction to solve the dynamic treatment choice problem sequentially from the final to initial stage. The other method involves the simultaneous maximization of $W\left(\cdot \right)$ over the entire class of DTRs $\MG$ across all stages. The backward-induction approach is computationally efficient; however, we will see that it may not consistently estimate the optimal DTR when $\MG_{t}$ does not contain the first-best treatment rule for all $t \geq 2$. 
By contrast, the simultaneous maximization method can consistently estimate the optimal DTR irrespective of whether $\MG_{t}$ contains the first-best rule at each stage $t$, though it is computationally less efficient.\footnote{It is worth noting that our study, focused on the consistent estimation of the optimal DTR, differs from the literature on ``dynamic (in)consistency'' in economics  \citep[e.g.,][]{Epstein_et_al_2003,Hansen_Sargent_2022}, because the notions of consistency are different between our work and works regarding ``dynamic (in)consistency'' in economics.}
We explain the backward-induction and simultaneous-maximization methods in Sections \ref{sec:backward DEWM} and \ref{sec:simultaneous DEWM}, respectively.

\subsection{Backward Dynamic Empirical Welfare Maximization \label{sec:backward DEWM}}

We first explain the backward-induction approach. 
To present the idea, we here suppose that the generative distribution function $P$ is known
and the pair $\left(P,{\cal G}\right)$ satisfies Assumptions \ref{asm:sequential independence} and \ref{asm:overlap}. 

The backward-induction approach in the population problem proceeds as follows. First, for the final stage
$T$, we obtain
\begin{align}
g_{T}^{\ast} & \in\argmax_{g_{T}\in{\cal G}_{T}}E_{P}\left[Q_{T}\left(H_{T},g_{T}(H_T)\right)\right],
\label{eq:backward_induction_first}
\end{align}
where $Q_{T}\left(h_{T},d_{T}\right)\equiv E_{P}\left[\gamma_{T}Y_{T}\mid H_{T}=h_{T},D_{T}=d_{T}\right]$
is the conditional mean of the weighted final outcome $\gamma_{T}Y_{T}$ given the history $h_{T}$ and treatment $d_T$. 

Then, recursively, from $t=T-1$ to $1$, we obtain 
\begin{align}
g_{t}^{\ast} & \in\argmax_{g_{t}\in{\cal G}_{t}}E_{P}\left[Q_{t}\left(H_{t},g_{t}(H_t)\right)\right], \label{eq:backward_induction}
\end{align}
with 
$Q_{t}\left(h_{t},d_{t}\right)  \equiv E_{P}\left[\gamma_{t}Y_{t}+Q_{t+1}\left(H_{t+1},g_{t+1}^{\ast}(H_{t+1})\right)\mid H_{t}=h_{t},D_{t}=d_{t}\right].$
The function $Q_{t}\left(h_{t},d_{t}\right)$ is the action value function for stage $t$ and represents the expected welfare that is realized when the history is $h_{t}$, the treatment at stage $t$ is $d_t$, and the future treatments follow $(g_{t+1}^{\ast},\ldots,g_{T}^{\ast})$.

Given the propensity scores $\left\{ e_{t}\left(d_{t},h_{t}\right)\right\} _{t=1}^{T}$
and under Assumption \ref{asm:sequential independence}, $E_{P}\left[Q_{t}\left(H_{t},g_{t}(H_t)\right)\right]$
can be identified as 
\begin{align*}
E_{P}\left[Q_{t}\left(H_{t},g_{t}\right)\right] & =E_{P}\left[q_{t}\left(Z,g_{t};g_{t+1}^{\ast},\ldots,g_{T}^{\ast}\right)\right],
\end{align*}
where 
\begin{align}
q_{t}\left(Z,g_{t};g_{t+1},\ldots,g_{T}\right)\equiv  \sum_{s=t}^{T}\left\{ \frac{\left(\prod_{\ell =t}^{s}1\left\{ D_{\ell} = g_{\ell}\left(H_{\ell}\right)\right\}\right) \gamma_{s}Y_{s}}{\prod_{\ell =t}^{s}e_{\ell}\left(D_{\ell},H_{\ell}\right)}\right\} .\nonumber
\end{align}
Hence, the objective function $E_P[Q_t(H_t,g_t)]$ can be expressed by the observables only.

Using the inverse propensity score weighting, we propose the estimation method based on the empirical analogue of the above backward induction procedure. We refer to this method as the backward DEWM method. The backward DEWM method first estimates $g_{T}^{\ast}$ by 
\begin{align*}
\hat{g}_{T}^{B} & \in\argmax_{g_{T}\in{\cal G}_{T}}\frac{1}{n}\sum_{i=1}^{n}q_{T}\left(Z_{i},g_{T}\right).
\end{align*}
Then, recursively, from $t=T-1$ to $1$, the method estimates $g_{t}^{\ast}$ by 
\begin{align}
\hat{g}_{t}^{B} & \in\argmax_{g_{t}\in{\cal G}_{t}}\frac{1}{n}\sum_{i=1}^{n}q_{t}\left(Z_{i},g_{t};\hat{g}_{t+1}^{B},\ldots,\hat{g}_{T}^{B}\right). \label{eq:optimization_backward DEWM}
\end{align}
We denote by $\hat{g}^{B}\equiv \left(\hat{g}_{1}^{B},\ldots,\hat{g}_{T}^{B}\right)$
the DTR obtained from this procedure. 

The resulting DTR $\hat{g}^{B}$ does not necessarily have consistency to the optimal one, $g_{opt}^{\ast} \in \argmax_{g \in \MG} W(g)$, unless the class $\MG_t$ of treatment rules for each $t\geq 2$ contain the first-best rule that globally maximizes $E_{p}[Q_t(H_t,g_t(H_t))]$ over all measurable functions of $g_t$. 
For any $s<t$, let 
\begin{align*}
\widetilde{Y}_{t}\left(\underline{d}_{s},\underline{g}_{(s+1):t}\right) \equiv \sum_{\underline{d}_{(s+1):t} \in \{0,1\}^{t-s}} Y_t(\underline{d}_{s},\underline{d}_{(s+1):t})\cdot \prod_{\ell=s+1}^{t}1\left\{g_{\ell}\left(H_{\ell}\left(\underline{d}_{\ell-1}\right)\right) = d_\ell\right\},
\end{align*}
which is the outcome in stage $t$ that is realized when the treatment assignments from stage $1$ to stage $s$ are fixed to $\underline{d}_{s}$, and the subsequent sequential treatment assignment follows $\underline{g}_{(s+1):t}$.\footnote{We denote $\widetilde{Y}_{t}\left(\underline{d}_{t},\underline{g}_{(t+1):t}\right) = Y_{t}\left(\underline{d}_{t}\right)$ when $s=t$.} 
To ensure consistent estimation with a given distribution $P$, the following assumption requires that the first-best treatment rule is attainable at all but the first stage.

\medskip{}

\begin{assumption}[First-Best Treatment Rule] 
\label{asm:first-best} 
For any $t=2,\ldots,T$, there exists $g_{t,FB}^{\ast}\in{\cal G}_{t}$ such that the following holds:
\begin{align*}
&E_{P}\left[\sum_{s=t}^{T}\gamma_{s}\widetilde{Y}_{s}\left(\text{\underline{D}}_{t-1},\underline{g}_{t:s,FB}^{\ast}\right) \middle| H_{t}\right] \geq \max_{d_t \in \{0,1\}}E_{P}\left[\sum_{s=t}^{T}\gamma_{s}\widetilde{Y}_{s}\left(\text{\underline{D}}_{t-1},d_t,\underline{g}_{(t+1):T,FB}^{\ast}\right) \middle| H_{t} \right] \mbox{\ a.s.}
\end{align*}
\end{assumption}
\medskip{}

\noindent 
We refer to $g_{t,FB}^{\ast}$, which satisfies Assumption \ref{asm:first-best}, as the first-best treatment rule at stage $t$. The first-best rule $g_{t,FB}^{\ast}$ always chooses the best treatment for any history $h_t$ given that the first-best rules are followed in the future stages. Assumption \ref{asm:first-best} is satisfied when $\MG_{t}$, $t=2,\ldots,T$, are rich enough or are correctly specified in the sense that they contain the first-best rule. Note that there is a trade-off between the simplicity of a class of DTRs and the feasibility of Assumption \ref{asm:first-best}; while a simpler class of DTRs is often preferable in practice, it is less likely to contain the first-best rule.\footnote{A tension also exists between restrictions on information sets versus restrictions on functional classes. Imposing functional restrictions can restrict the information set, potentially causing dynamic inconsistency.} Assumption \ref{asm:first-best} does not require the class of treatment rules for the first stage $\MG_1$ to contain the first-best.

When the first-best rule is not attainable in $\MG_{t}$ for some $t\geq 2$, the solution $g_{s}^{\ast}$ of the backward induction for $s\leq t$ does not necessarily correspond to the optimal treatment rule. We illustrate this issue with a simple example in the following remark  (and also in the simulation study in Section \ref{sec:simulation}).

\medskip

\begin{remark}
Suppose that $T=2$ and the data-generating process (DGP) $P$ satisfies the following:
\begin{align}
    & E_{P}[Y_{2}(1,1)]=1.0,\ E_{P}[Y_{2}(1,0)]=0.5,\ E_{P}[Y_{2}(0,1)]=0.0,\ E_{P}[Y_{2}(0,0)]=0.6; \nonumber \\
    & \mbox{$D_{1}$ and $D_{2}$ are independently distributed as $Ber(1/2)$}. \label{eq:example_DGP} 
\end{align}
We set the target welfare to $$W(g) = E_P\left[\widetilde{Y}_2(g_1,g_2)\right] =E_{P}\left[\sum_{(d_1,d_2)\in \{0,1\}^2}Y_2\left(d_1,d_2\right)\cdot 1\{g_1(H_1) = d_1,g_2(H_2(d_1))=d_2\}\right].$$ Suppose that the history information are $H_1 = \emptyset$ and $H_2 = (D_1)$.

As an example of a constrained class of DTRs, we consider a class of uniform DTRs; that is $\MG_t = \{c_{t}^{0},c_{t}^{1}\}$, for $t=1,2$, where $c_{t}^{0}$ and $c_{t}^{1}$ denote constant functions such that $c_{t}^{0}(h_t)=0$ and $c_{t}^{1}(h_t)=1$ for any $h_t$. Under the supposed DGP $P$, the first-best rule for $t=2$ is $g_{2,FB}^{\ast}(d_1)=d_1$. Hence $\MG_2$ does not contain the first-best. 

The optimal DTR over the class of constant DTRs is
\begin{align*}
    (g_{1,opt}^{\ast},g_{2,opt}^{\ast}) = \argmax_{(g_1,g_2) \in \{c_{1}^{0},c_{1}^{1}\} \times \{c_{2}^{0},c_{2}^{1}\}} E\left[\widetilde{Y}_{2}(g_1,g_2)\right] = (c_{1}^{1},c_{2}^{1}),
\end{align*}
and its welfare is $W(g_{1,opt}^{\ast},g_{2,opt}^{\ast})=E[Y_2(1,1)]=1.0$.
On the other hand, the solution $(g_{1}^{\ast},g_{2}^{\ast})$ of the backward-induction approach is $(c_{1}^{0},c_{2}^{0})$ because
\begin{align*}
    \mbox{(1st step)\ \ \ \ }g_{2}^{\ast} &= \argmax_{g_2 \in \{c_{2}^{0},c_{2}^{1}\}}E_{P}\left[\widetilde{Y}_2(D_1,g_2)\right] = c_{2}^{0};  \nonumber \\
    \mbox{(2nd step)\ \ \ \ }g_{1}^{\ast} &= \argmax_{g_1 \in \{c_{1}^{0},c_{1}^{1}\}}E_{P}\left[\widetilde{Y}_2(g_1,g_{2}^{\ast})\right]  = c_{1}^{0}.  
\end{align*}
Hence, the backward-induction solution $g^{\ast}=(c_{1}^{0},c_{2}^{0})$ differs from the optimal solution $g_{opt}^{\ast}=(c_{1}^{1},c_{2}^{1})$ over $\MG$, resulting in a suboptimal welfare $W(g^{\ast})=E[Y_2(0,0)]=0.6$.
\par
The above example suggests that when the first-best rule is not feasible in $\MG_{t}$ ($t\geq 2$), the backward-induction solution does not necessarily correspond to the optimal one. 
This happens because the backward-induction solution $g_{t}^{\ast}$ depends on the DGP $P$ of the observed data in which the distribution of treatment assignments $(D_1,D_2)$ is decided by the experimental design.
This DGP differs from the DGP that arises when the treatment assignments, except for stage $t$, follow the optimal treatment rules. 
However, when the first-best rule is feasible in $\MG_{t}$ for each $t \geq 2$, the backward-induction solution $g_{t}^{\ast}$ at each stage corresponds to the first-best rule, under the overlap condition, irrespective of the distribution of $(D_1,D_2)$.
\par 
Finally, note that the infeasibility of the first-best rule does not necessarily cause the suboptimality of the backward-induction approach for a fixed DGP. Suppose that the DGP $P$ satisfies the condition (\ref{eq:example_DGP}) with $E_P[Y_2(0,1)]=0.0$ replaced by $E_P[Y_2(0,1)]=0.4$. In this case, the backward-induction solution becomes $g^{\ast}=(c_{1}^{1},c_{2}^{1})$ and corresponds to the optimal one $g_{opt}^{\ast} = (c_{1}^{1},c_{2}^{1})$.\footnote{There is also another example. Consider decision rules that rely solely on a discretized version of the history space. In such a scenario, backward induction can still achieve the optimal decision rule within this discretized class, treating the discretized history space as a new set of covariates.}
\end{remark}

\subsection{Simultaneous Dynamic Empirical Welfare Maximization \label{sec:simultaneous DEWM}}

The second approach is a sample analogue of the entire welfare maximization problem (\ref{eq:welfare maximization}). We refer to the proposed method as the simultaneous DEWM method, as it simultaneously estimates the optimal treatment rules across all stages. The method estimates the optimal DTR through the maximization of the sample analogue of (\ref{eq:original welfare}): 
\begin{align}
\left(\hat{g}_{1}^{S},\dots,\hat{g}_{T}^{S}\right) & \in\argmax_{g\in{\cal {\cal G}}}\sum_{t=1}^{T}\left[\frac{1}{n}\sum_{i=1}^{n}w_{t}^{S}(Z_{i},\text{\ensuremath{\underline{g}}}_{t})\right],\label{eq:optim_SDEWM}
\end{align}
where $\underline{g}_{t} \equiv (g_1,\ldots,g_{t})$ is the vector of treatment rules up to stage $t$ and
\begin{align*}
w_{t}^{S}(Z_{i},\text{\ensuremath{\underline{g}}}_{t}) & \equiv \frac{\left(\prod_{s=1}^{t}1\left\{ D_{is} = g_{s}\left(H_{is}\right)\right\}\right) \gamma_{t}Y_{it}}{\prod_{s=1}^{t}e_{s}\left(D_{is},H_{is}\right)}.
\end{align*}
In equation (\ref{eq:optim_SDEWM}), $n^{-1}\sum_{i=1}^{n}w_{t}^{S}(Z_{i},\text{\ensuremath{\underline{g}}}_{t})$ corresponds to the sample analogue of the $t$-th term in (\ref{eq:original welfare}).
We denote by $\hat{g}^{S}\equiv \left(\hat{g}_{1}^{S},\ldots,\hat{g}_{T}^{S}\right)$
the DTR obtained from this procedure. 
Theorem \ref{thm:upper bound} below shows that this method can consistently estimate the optimal DTR on $\MG$ even when $\MG_t$ does not contain the first-best rule for some $t$ (i.e., Assumption \ref{asm:first-best} does not hold).

\medskip

\begin{remark}[Optimization]\label{rem:MILP}
When $\MG_{t}$ ($t=1,\ldots,T$) are classes of the linear treatment rules, the optimization problems (\ref{eq:optimization_backward DEWM}) for the backward DEWM and (\ref{eq:optim_SDEWM}) for the simultaneous DEWM can be formulated as mixed integer linear programming (MILP) problems. Appendix \ref{appendix:computation} gives details.
\end{remark}

\smallskip

\begin{remark}[Q-learning]
The Q-learning method is also based on the idea of backward induction \citep{Murphy_2005,Moodie_et_al_2012}.
In the first step, the method estimates Q-function for stage $T$, $Q_{T}^{\dagger}\left(h_{t},d_{t}\right) \equiv E_{P}[Y_{T}| H_{T}=h_{T},D_{T}=d_{t}]$, through regression of $Y_T$ on $(H_T,D_{T})$ and obtain its estimate $\widehat{Q}_{T}^{\dagger}\left(h_{t},d_{t}\right)$. Then it estimates the optimal treatment rule for stage $T$ as $\hat{g}_{T}^{Q}(h_{T}) = \argmax_{d_{T} \in \{0,1\}} \widehat{Q}_{T}^{\dagger}\left(h_{t},d_{t}\right)$. Recursively, from $t=T-1$ to $1$, the method estimates the Q-function (optimal action-value function) for stage $t$, $Q_{t}^{\dagger}(h_t,d_{t}) \equiv E_{P}\left[Y_{t} + \gamma_{t+1}\max_{d_{t+1}}Q_{t+1}^{\dagger}(h_{t+1},d_{t+1})| H_{t}=h_{t},D_{t}=d_{t}\right]$, by regressing $Y_{t} + \gamma_{t+1}\max_{d_{t+1}}\widehat{Q}_{t+1}^{\dagger}(h_{t+1},d_{t+1})$ on $(H_{t},D_{t})$, and obtain its estimate $\widehat{Q}_{t}^{\dagger}(h_t,d_{t})$.\footnote{Linear regression is typically used to estimate the Q-functions.} Then it estimates the optimal treatment rule for stage $t$ as $\hat{g}_{t}^{Q}(h_{t}) = \argmax_{d_{t} \in \{0,1\}} \widehat{Q}_{t}^{\dagger}\left(h_{t},d_{t}\right)$.
The method yields a DTR $\hat{g}^{Q}\equiv \left(\hat{g}_{1}^{Q},\ldots,\hat{g}_{T}^{Q}\right)$. 

Q-learning is simple to implement and computationally tractable. Moreover, it does not require overlap conditions of propensity scores. However, it requires the correct specification of the Q-functions for consistent estimation of the optimal DTRs, even when experimental data is used.
Our proposed methods do not require the specification of the Q-functions; instead, they use the propensity scores. Additionally, while the backward DEWM requires the specified class of DTRs to include the first-best rules, the simultaneous DEWM does not.
\end{remark}

\smallskip

\begin{remark}[Non-Linear Social Welfare]
So far we have considered the linear form (\ref{eq:additive_welfare_function}) of the welfare function. However, some important social welfare criteria (e.g., Gini social welfare \citep{Blackorby_Donaldson_1978,Weymark_1981}) are represented by non-linear social welfare functions. In Appendix \ref{appendix:non-additive_welfare_function}, we consider the 
equality-minded rank-dependent social welfare functions introduced by \cite{Meyer_1995} and \cite{Weymark_1981} and studied by \cite{Kitagawa_Tetenov_2021}: 
\begin{align}
    W_{\Lambda}(F) \equiv \int_{0}^{\infty}\Lambda(F(y))dy, \label{eq:rank-dependent_welfare_function_1}
\end{align}
where $F(y)$ is the distribution of an outcome and $\Lambda(\cdot):[0,1] \rightarrow [0,1]$ is a non-increasing, non-negative function with $\Lambda(0)=1$ and $\Lambda(1)=0$. 
An important family of social welfare functions represented by (\ref{eq:rank-dependent_welfare_function_1}) is the extended Gini family \citep{Donaldson_Weymark_1980,Donaldson_Weymark_1983,Aaberge_et_al_2013}:
$W_k(F) \equiv \int_{0}^{\infty}(1 - F(y))^{k-1}dy$. When $k=3$, $W_k(F)$ corresponds to the standard Gini social welfare function \citep{Blackorby_Donaldson_1978,Weymark_1981}: $W_{Gini}(F) = E(Y)(1 - I_{Gini}(F))$ with $I_{Gini}(F) = 1 - (\int_{0}^{1}F^{-1}(\tau)\cdot 2(1-\tau)d\tau)/E(Y)$.

For any DTR  $g=(g_1, \ldots, g_T)$, let $F_{g}(\cdot)$ denote the distribution of $\sum_{t=1}^{T}\gamma_{t} \widetilde{Y}_{t}(\underline{g}_{t})$, and we define the rank-dependent SWF of $g$ by $W_{\Lambda}(g) \equiv W_{\Lambda} (F_g)$.
Appendix \ref{appendix:non-additive_welfare_function} presents a simultaneous DEWM approach to estimate the optimal DTR that maximizes the non-linear social welfare function $W_{\Lambda}(g)$ over $\MG$, and shows its statistical properties.

\end{remark}

\smallskip

\begin{remark}[Multiple Treatment]
    We have so far considered DTRs with binary treatment in each stage. 
    Suppose that there are $K$ treatments in each stage. The discussion so far and the presented procedures are easily extendable to the multiple treatment setting by replacing the binary treatment class $\{0,1\}$ with the multiple one $\{1,\ldots,K\}$. In this case, the treatment rule $g_t$ becomes a map from $\MH_t$ to $\{1,\ldots,K\}$. Appendix \ref{appendix:multiple} elaborates on this extension.
\end{remark}

\subsection{Statistical Properties \label{sec:statistical property}}

As in much of the literature that follows Manski (2004), we evaluate the statistical properties of the two DEWM methods in terms of the average welfare regret, that is, the average welfare loss relative to the maximum feasible welfare $W_{\MG}^{\ast}$. Following \citet{Kitagawa_Tetenov_2018a}, we focus on the non-asymptotic
upper bounds of the worst-case average welfare regret, $\sup_{P\in{\cal P}\left(M, \kappa, \MG\right)}E_{P^{n}}\left[W_{{\cal G}}^{\ast}-W\left(\hat{g}\right)\right]$, where ${\cal P}\left(M, \kappa, \MG\right)$
is a class of distributions of $\left(\underline{D}_T,\{\underline{Y}_{T}(\underline{d}_{T})\}_{\underline{d}_{T} \in \{0,1\}^{T}},\{\underline{X}_{T}(\underline{d}_{T-1})\}_{\underline{d}_{T-1} \in \{0,1\}^{T-1}}\right)$ that satisfy Assumptions \ref{asm:sequential independence}, \ref{asm:bounded outcome}, and \ref{asm:overlap}
with $M\equiv\left(M_{1},\ldots,M_{T}\right)^{\prime}$,  $\kappa\equiv\left(\kappa_{1},\ldots,\kappa_{T}\right)^{\prime}$, and a fixed $\MG$.

The following theorem provides a finite-sample upper bound on the worst-case average welfare regret and shows its dependence on the sample size $n$, the VC-dimension of $\mathcal{G}_{t}$ for each $t$, and the number of stages $T$.

\medskip{}

\begin{theorem}\label{thm:upper bound} Suppose that Assumptions
\ref{asm:sequential independence}, \ref{asm:bounded outcome}, and \ref{asm:overlap} hold for any distribution
$P\in{\cal P}\left(M, \kappa, \MG\right)$ and Assumption \ref{asm:vc-class}
holds for $\MG$.\\
 (i) For the simultaneous DEWM method, there holds
\begin{align*}
\sup_{P\in{\cal P}\left(M, \kappa, \MG\right)}E_{P^{n}}\left[W_{{\cal G}}^{\ast}-W\left(\hat{g}^{S}\right)\right] & \leq C\sum_{t=1}^{T}\left\{ \frac{\gamma_{t}M_{t}}{\prod_{s=1}^{t}\kappa_{s}}\sqrt{\frac{\sum_{s=1}^{t}v_{s}}{n}}\right\} ,
\end{align*}
where $C$ is some universal constant.\\
 (ii) Suppose, in addition, that Assumption \ref{asm:first-best} holds
for a pair of $\mathcal{G}$ and any $P\in{\cal P}\left(M, \kappa, \MG\right)$.
Then, for the backward DEWM method, there holds
\begin{align*}
\sup_{P\in{\cal P}\left(M, \kappa, \MG\right)}E_{P^{n}}\left[W_{{\cal G}}^{\ast}-W\left(\hat{g}^{B}\right)\right] & \leq  C\sum_{t=1}^{T}\left\{ \frac{\gamma_{t}M_{t}}{\prod_{s=1}^{t}\kappa_{s}}\sqrt{\frac{\sum_{s=1}^{t}v_{s}}{n}}\right\} \\
 & +C\sum_{t=2}^{T}\frac{2^{t-2}}{\prod_{s=1}^{t-1}\kappa_{s}}\left(\sum_{s=t}^{T}\left\{ \frac{\gamma_{s}M_{s}}{\prod_{\ell=t}^{s}\kappa_{\ell}}\sqrt{\frac{\sum_{\ell=t}^{s}v_{\ell}}{n}}\right\} \right),
\end{align*}
where $C$ is the same universal constant. \end{theorem}

\begin{proof}
See Appendix \ref{appendix:proofs-1}.
\end{proof}

\medskip{}

This theorem shows that the convergence rates of the worst-case average welfare regrets of the two methods are not slower than $n^{-1/2}$. The upper bounds increase with the VC-dimension of $\MG_{t}$, implying that as the candidate treatment rules become more complex, the estimated DTR tends to overfit the data (the distribution of welfare regret becomes more dispersed).\footnote{ 
When the VC-dimension $v_t$ increases with the sample size $n$, Theorem \ref{thm:upper bound} implies that the rate of convergence of the welfare regrets depends on this growth rate.}
The upper bound for the backward DEWM method is greater than that for the simultaneous DEWM method, though neither bound is necessarily sharp. Technically, the difference between these bounds arises from the property of the sequential estimation of the backward DEWM, which leads to additional uncertainty in the estimation.

The next theorem shows a lower bound on the maximum average welfare regret for any data-driven DTR. To present the theorem formally, let $v_{s:t}$, for $s\leq t$, denote the VC-dimension of the following class of indicator functions on $\MZ$:
\begin{align*}
    \left\{ f(z)=1\left\{ g_{s}\left(h_{s}\right)=d_{s},\ldots,g_{t}\left(h_{t}\right)=d_{t}\right\} :\left(g_{s},\ldots,g_{t}\right)\in\mathcal{G}_{s}\times\cdots\times\mathcal{G}_{t}\right\} .
\end{align*}
Note that $v_{s:t}\leq\sum_{\ell=s}^{t}v_{\ell}$
holds (see Lemma \ref{lem:vc_subclass}).

\medskip{}

\begin{theorem}\label{thm:lower bound} Suppose that Assumptions \ref{asm:sequential independence}, \ref{asm:bounded outcome}, and \ref{asm:overlap} hold for any distribution $P\in{\cal P}\left(M, \kappa, \MG\right)$ and Assumption \ref{asm:vc-class} holds for ${\cal G}$. Then, for any DTR $\hat{g}\in\mathcal{G}$ as a function of $\left(Z_{1},\ldots,Z_{n}\right)$, there holds
\begin{align*}
\sup_{P\in{\cal P}\left(M, \kappa, \MG\right)}E_{P^{n}}\left[W_{{\cal G}}^{\ast}-W\left(\hat{g}\right)\right] & \geq \frac{1}{2}\exp\left(-4\right)\max_{t \in \{1,\ldots,T\}}\left\{ \gamma_{t}M_{t}\sqrt{\frac{v_{1:t}}{n}}\right\} 
\end{align*}
for all $n\geq16v_{1:T}$. This result holds irrespective of whether or not Assumption \ref{asm:first-best} additionally holds for a pair of $\MG$ and any $P \in \MP(M, \kappa, \MG)$.
\end{theorem}

\begin{proof}
See Appendix \ref{appendix:proof}.
\end{proof}

\medskip{}

This theorem, along with Theorem \ref{thm:upper bound}, shows that both $\hat{g}^{S}$ and $\hat{g}^{B}$ are minimax rate optimal over the class of DGPs $\MP\left(M, \kappa, \MG\right)$. Optimality here means that the convergence rates of the upper bounds of the worst-case average welfare regrets in Theorem \ref{thm:upper bound} align with the convergence rate of the universal lower bound concerning the sample size $n$. 
The convergence rate is also optimal with respect to the VC-dimension $v_{t}$ for each $t$.
In Theorem \ref{thm:lower bound}, the maximum of $\gamma_{t}M_{t}\sqrt{v_{1:t}/n}$ over $t=1,\ldots,T$, rather than its summation over $t=1,\ldots,T$, appears in the lower bound, which is due to the simplicity of the derivation of the lower bound in its proof.
\medskip{}

\begin{remark}
\label{rem:deman outcome}
The finite sample optimization problems (\ref{eq:optimization_backward DEWM}) and (\ref{eq:optim_SDEWM}) are not invariant to adding a constant, which can affect the estimated DTR by manipulating the outcome variables.  
Following \cite{Kitagawa_Tetenov_2018a}, we suggest using the demeaned outcomes $Y_{it}-(1/n)\sum_{i=1}^{n}Y_{it}$, instead of the original ones $Y_{t}$, in the optimization problems (\ref{eq:optimization_backward DEWM}) and (\ref{eq:optim_SDEWM}), because it is invariant to adding a constant to the original outcome.  
\end{remark}

\section{Budget/Capacity Constraints \label{sec:budget constraint}}

We consider budget/capacity constraints that limit the proportion of the population receiving treatment. In dynamic treatment policy, these constraints may be imposed intertemporally, meaning that they affect treatment assignment across multiple stages. A policymaker faces an intertemporal budget/capacity constraint when managing a budget that spans across multiple stages or a fixed amount of treatment to distribute over multiple stages.\footnote{In the static setting, \cite{Bhattacharya_Dupas_2012} propose a method to estimate the optimal treatment rule under a budget constraint. As its application, they estimate the optimal allocation policy for subsidies of anti-malaria bed nets under budget constraints.} For instance, the job training program studied by \cite{Rodriguez_et_al_2018} subsidizes training courses at off-site providers across multiple stages, where, when the subsidy budget is limited, the program faces intertemporal budget constraints, limiting the number of individuals participating in training across multiple stages.


Similar to the definition of $\widetilde{Y}_{t}\left(\underline{g}_{t}\right)$, we define a counterfactual history as 
\begin{align*}
    \widetilde{H}_t\left(\underline{g}_{t-1}\right) &\equiv \sum_{\underline{d}_{t-1} \in \{0,1\}^{t-1}}H_t(\underline{d}_{t-1})\cdot \prod_{s=1}^{t-1}1\left\{g_s\left(H_{s}(\underline{d}_{s-1})\right) = d_s\right\},
\end{align*} 
which is the counterfactual history in stage $t$ that is realized when the prior treatments $\underline{d}_{t-1}$ are decided by $\underline{g}_{t-1}$. We denote $\widetilde{H}_{1}\left(\underline{g}_0\right) = H_1$ when $t=1$.
We suppose that the policymaker faces the following $B$ constraints: 
\begin{align}
\sum_{t=1}^{T}K_{tb}E_{P}\left[g_{t}\left(\widetilde{H}_t\left(\underline{g}_{t-1}\right)\right)\right] & \leq C_{b}\ \ \mbox{for }\ b=1,\ldots,B,\label{eq:budget constraint}
\end{align}
where $K_{tb}\in\left[0,1\right]$ and $C_{b}\geq0$. As a scale
normalization, we assume $\sum_{t=1}^{T}K_{tb}=1$ for all $b$.
The left-hand side of equation (\ref{eq:budget constraint}) represents the implementation cost of the DTR $g$, where
the weights $K_{1b},\ldots,K_{Tb}$ represent the relative
costs of treatments across stages, and $C_{b}$ represents the total
budget or capacity.  If at least two of $K_{1b},\cdots,K_{Tb}$ take non-zero values, the $b$-th constraint is an intertemporal budget/capacity constraint; otherwise, the $b$-th constraint is a temporal one. In the context of the two-stage job training program with an intertemporal budget constraint ($B=1$), $k_{11}$ and $k_{21}$ represent costs of job training for the first and second stages, respectively, and $C_{1}$ represents the intertemporal budget of the program.\footnote{In reality, the time periods of individuals receiving the treatment would not be aligned. For example, different individuals take job training (for each stage) at different times. In such cases, the formulation (\ref{eq:budget constraint}) of budget constraints can be considered as follows. Suppose that a provider of the treatments (e.g., government) has a fixed budget that can be expended in a fixed fiscal period. The provider (correctly) predicts the number of participants of the program during the fiscal period. We also suppose that the budget can be expended on treatment for any stage for those who participate in the program in any time during the fiscal period. Subsequently, given the budget, the provider can decide the fraction of people who can receive treatment at each stage, as formulated as (\ref{eq:budget constraint}).}

Our aim is to maximize the welfare $W(g)$ under the budget/capacity constraints (\ref{eq:budget constraint}) across the class of feasible DTRs $\mathcal{G}$. The population welfare maximization problem is then formulated as 
\begin{align}
& W_{{\cal G}}^{\ast,bdgt}  =\max_{g\in{\cal G}}W\left(g\right)\label{eq:budget constrained welfare}\\
\mbox{s.t. } & \sum_{t=1}^{T}K_{tb}E_{P}\left[g_{t}\left(\widetilde{H}_{t}\left(\underline{g}_{t-1}\right)\right)\right]\leq C_{b}\ \mbox{for }b=1,\ldots,B.\nonumber 
\end{align}
The goal of the analysis is to choose a DTR from ${\cal G}$ that maximizes the welfare $W(\cdot)$ subject to the budget/capacity constraints (\ref{eq:budget constraint}).  

To this end, we incorporate the sample analogues of the budget/capacity constraints (\ref{eq:budget constraint}) into the simultaneous DEWM.\footnote{We here do not consider the backward DEWM with the budget/capacity constraints because the first-best rule is likely to be unachievable under such constraints.}
The simultaneous DEWM method with the budget/capacity constraints solves the following problem: 
\begin{align}
& \left(\hat{g}_{1}^{bdgt},\dots,\hat{g}_{T}^{bdgt}\right)  \in\argmax_{g\in{\cal {\cal G}}}\frac{1}{n}\sum_{i=1}^{n}\sum_{t=1}^{T}w_{t}^{S}\left(Z_{i},\underline{g}_{t}\right)\label{eq:budget constrained SDEWM}\\
\mbox{s.t. } & \sum_{t=1}^{T}K_{tb}\widehat{E}\left[g_{t}\left(\widetilde{H}_{t}\left(\underline{g}_{t-1}\right)\right)\right]\leq C_{b}+\alpha_{n}  \mbox{\ for }b=1,\ldots,B, \label{eq:empirical budget constraint}
\end{align}
where 
\begin{align*}
\widehat{E}\left[g_{t}\left(\widetilde{H}_{t}\left(\underline{g}_{t-1}\right)\right)\right]
\equiv \frac{\sum_{i=1}^{n}\left(\prod_{s=1}^{t-1}1\left\{ D_{is}=g_{s}\left(H_{is}\right)\right\}\right) g_{t}\left(H_{it}\right)}{\sum_{i=1}^{n}\left(\prod_{s=1}^{t-1}1\left\{ D_{is}=g_{s}\left(H_{is}\right)\right\}\right) }.
\end{align*}
We denote $\hat{g}^{bdgt} \equiv \left(\hat{g}_{1}^{bdgt},\dots,\hat{g}_{T}^{bdgt}\right)$.

The inequality constraints (\ref{eq:empirical budget constraint}) are empirical budget/capacity constraints, where $\alpha_{n}$ is a tuning parameter dependent on the sample size $n$. $\alpha_{n}$ may be either positive or negative, and converges to zero as $n$ increases. As $\alpha_{n}$ decreases, the empirical budget/capacity constraints become tighter. A sufficiently large value of $\alpha_{n}$ ensures that the optimal DTR (a solution of (\ref{eq:budget constrained welfare})) is attainable under the empirical budget/capacity constraints with high probability. When $\MG_{t}$ is the class of linear treatment rules for all $t$, the optimization problem (\ref{eq:budget constrained SDEWM}) can be formulated as an MILP problem (see Appendix \ref{appendix:computation}). 

Subsequently, we evaluate the resulting welfare regret $W_{{\cal G}}^{\ast,bdgt}-W\left(\hat{g}^{bdgt}\right)$ and the budget excess $\sum_{t=1}^{T}K_{tb}E_{P}\left[\hat{g}_{t}^{bdgt}\left(H_{t}\left(\underline{\hat{g}}_{t-1}^{bdgt}\right)\right)\right]-C_b$ of the estimated DTR with high probability, rather than evaluating their expected values, $E_{P^n}\left[W_{{\cal G}}^{\ast,bdgt}-W\left(\hat{g}^{bdgt}\right)\right]$ and $\sum_{t=1}^{T}K_{tb}E_{P^{n}}\left[E_{P}\left[\hat{g}_{t}^{bdgt}\left(H_{t}\left(\underline{\hat{g}}_{t-1}^{bdgt}\right)\right)\right]\right]-C_b$.\footnote{Note that $W\left(\hat{g}^{bdgt}\right)$ and $E_{P}\left[\hat{g}^{bdgt}\left(H_{t}\left(\underline{\hat{g}}_{t-1}^{S}\right)\right)\right]$ are random variables depending on the random sample $\{Z_i:i=1,\ldots,n\}$. } We adopt this approach because the actual value of the budget excess is typically of greater concern than its expected value in practice. 

The following theorem shows the finite-sample properties of the welfare regret and the budget excess of $\hat{g}^{bdgt}$.

 \medskip{}

\begin{theorem}\label{thm:budget constraint} Suppose that the underlying distribution $P$ satisfies Assumptions \ref{asm:sequential independence} and \ref{asm:bounded outcome}, $\MG$ satisfies Assumption $\ref{asm:vc-class}$, and that the pair $(P,\MG)$ satisfies Assumption \ref{asm:overlap}. 
Let $W_{{\cal G}}^{\ast,bdgt}$ be defined in (\ref{eq:budget constrained welfare}) and $\hat{g}^{bdgt}$ be a solution of (\ref{eq:budget constrained SDEWM}) subject to (\ref{eq:empirical budget constraint}). Let $\delta$ be any value in $(0,1)$ and $C$ be the same constant as in Theorem \ref{thm:upper bound}. Let $k_{(B,n,\delta)} := \sqrt{\log\left(6B/\delta\right)/\left(2n\right)}$, and $W_{\MG,\alpha_{n}}^{\ast,bdgt}$ be the optimal value of the optimization problem (\ref{eq:budget constrained welfare}) with $C_{b}$ replaced by $C_{b} - k_{(B,n,\delta)} + \alpha_n$, assuming that such an optimal value exists. Then the following holds with probability at least $1-\delta$:
\begin{align}
\left|W_{{\cal G}}^{\ast,bdgt}-W\left(\hat{g}^{bdgt}\right)\right|
&\leq  
\left(W_{{\cal G}}^{\ast,bdgt}-W_{\MG,\alpha_{n}}^{\ast,bdgt}\right) \notag \\ 
& +  \frac{1}{\sqrt{n}}
\sum_{t=1}^{T}\left[\left(\frac{\gamma_{t}M_{t}}{\prod_{s=1}^{t}\kappa_{s}}\right)\left(2C\sqrt{\sum_{s=1}^{t}v_{s}}+\sqrt{2\log\left(6/\delta\right)}\right)\right]  \label{eq:maximum welfare loss_2}
\end{align}
and, for any $b \in \{1,\ldots,B\}$,
\begin{align}
\sum_{t=1}^{T}K_{tb}E_{P}\left[\hat{g}_{t}^{S}\left(\widetilde{H}_{t}\left(\underline{\hat{g}}_{t-1}^{bgdt}\right)\right)\right]-C_{b} 
\leq  \frac{1}{\sqrt{n}} \sum_{t=1}^{T}\left[K_{tb}\left(C\sqrt{\sum_{s=1}^{t}v_{s}}+\sqrt{\frac{\log\left(2B/\delta\right)}{2}}\right)\right] + \alpha_n. \label{eq:excess budget_2}
\end{align}
\end{theorem}
\begin{proof}
See Appendix \ref{appendix:proofs-2}.
\end{proof}

\medskip{}

Equations (\ref{eq:maximum welfare loss_2}) and (\ref{eq:excess budget_2}) evaluate the welfare regret and the budget excess of the estimated DTR over the $b$-th budget/capacity, respectively. 
In equation (\ref{eq:maximum welfare loss_2}), $W_{{\cal G}}^{\ast,bdgt}-W_{\MG,\alpha_{n}}^{\ast,bdgt} \leq 0$ holds when $\alpha_n \leq k_{(B,n,\delta)}$, and $W_{{\cal G}}^{\ast,bdgt}-W_{\MG,\alpha_{n}}^{\ast,bdgt} \geq 0$ holds otherwise. Hence, when we set $\alpha_n$ such that $\alpha_n \leq k_{(B,n,\delta)}$, the results in Theorem \ref{thm:budget constraint} holds with equation (\ref{eq:maximum welfare loss_2}) replaced by
\begin{align*}
\left|W_{{\cal G}}^{\ast,bdgt}-W\left(\hat{g}^{bdgt}\right)\right|
&\leq  
\frac{1}{\sqrt{n}}
\sum_{t=1}^{T}\left[\left(\frac{\gamma_{t}M_{t}}{\prod_{s=1}^{t}\kappa_{s}}\right)\left(2C\sqrt{\sum_{s=1}^{t}v_{s}}+\sqrt{2\log\left(6/\delta\right)}\right)\right],
\end{align*}
where the welfare regret converges to zero as $n$ increases.
The theorem suggests that with sufficiently large sample sizes, both the welfare regret and the budget excess are likely to be small, diminishing at the rate of $1/\sqrt{n}$ when $\alpha_{n}$ is chosen such that $\alpha_n = O(1/\sqrt{n})$ and that $\alpha_n \leq  k_{(B,n,\delta)}$.\footnote{When $\alpha_n - k_{(B,n,\delta)} \searrow 0$,  whether $W_{\MG}^{\ast} - W_{\MG,\alpha_{n}}^{\ast,bdgt}$ in (\ref{eq:maximum welfare loss_2}) converges to zero depends on the properties of the class of DTR $\MG$ and distribution $P$. When $\MG$ consists of a finite number of functions, the maximum welfare subject to budget constraints may not be continuous with respect to the budget under some $P$, and hence $W_{\MG}^{\ast} - W_{\MG,\alpha_{n}}^{\ast}$ may not converge to zero when $\alpha_n - k_{(B,n,\delta)} \searrow 0$.}

The tuning parameter $\alpha_{n}$ decides the strictness of the budget constraint. A smaller $\alpha_{n}$ implies that the estimated DTR tightly satisfies the budget constraint, as seen in (\ref{eq:excess budget_2}), but leads to lower welfare. Conversely, a larger $\alpha_{n}$ results in less strict adherence to the budget constraint and higher welfare. Thus, the choice of $\alpha_n$ involves a trade-off between maximizing welfare and minimizing budget excess. 

We here propose two approaches to choose $\alpha_n$. When aiming to satisfy the budget constraints with a certain level of budget excesses and a particular probability, we can analytically choose the proper value of $\alpha_n$ through Theorem \ref{thm:budget constraint}. For example, for any $\varepsilon \in (0,1)$, Theorem \ref{thm:budget constraint} guarantees that the excess budget for the $b$-th constraint is equal to or smaller than $\varepsilon$ with probability at least $1-\delta$ when we choose $\alpha_n = \varepsilon - \frac{1}{\sqrt{n}} \sum_{t=1}^{T}\left[K_{tb}\left(C\sqrt{\sum_{s=1}^{t}v_{s}}+\sqrt{\log\left(2B/\delta\right)/2}\right)\right]$. We can also use cross-validation to choose $\alpha_n$, wherein the validation data evaluates the welfare and budget excess for each DTR estimated with candidate values of $\alpha_n$. Note that $k_{(B,n,\delta)}$ is not a tuning parameter to be selected. Theorem \ref{thm:budget constraint} also guides the selection of the sample size $n$ so that the budget excess is constrained to a certain level with a particular probability.

\section{Observational Study \label{sec:estimated propensity score}}

We consider the observational data setting where the propensity scores are not known but can be estimated from data. We modify the backward and simultaneous DEWM methods to use the estimated propensity scores, following the e-hybrid EWM rule proposed by \citet{Kitagawa_Tetenov_2018a}. We also discuss construction of a doubly robust approach for estimating the optimal DTRs with technical exposition and theoretical results presented.

Let $\hat{e}_{t}\left(d_{t},h_{t}\right)$ be an estimated version of the propensity score $e_{t}\left(d_{t},h_{t}\right)$. For the estimators of the propensity scores, we suppose the following high-level assumption.

\medskip{}

\begin{assumption}\label{asm:estimated propensity score} (i) Define
\begin{align*}
\tau_{t}\left(\underline{d}_{t},H_t\right)\equiv\left\{ \frac{\left(\prod_{s=1}^{t}1\left\{D_{s}=d_{s}\right\}\right) \gamma_{t}Y_{t}}{\prod_{s=1}^{t}e_{s}\left(d_{s},H_{s}\right)}\right\} \mbox{ and }\ \hat{\tau}_{t}\left(\underline{d}_{t},H_t\right)\equiv\left\{ \frac{\left(\prod_{s=1}^{t}1\left\{D_{s}=d_{s}\right\}\right) \gamma_{t}Y_{t}}{\prod_{s=1}^{t}\hat{e}_{s}\left(d_{s},H_{s}\right)}\right\} ,
\end{align*}
where $\hat{e}_{t}\left(d_{t},H_{t}\right)$ is an estimated propensity
score taking a value in $\left(0,1\right)$. For a class of data generating
processes ${\cal P}_{e}$, there exists a sequence $\phi_{n}\rightarrow\infty$
such that 
\begin{align*}
\underset{P\in{\cal P}_{e}}{\sup}\sup_{t\in\left\{ 1,\ldots,T\right\} }\sum_{\text{\ensuremath{\underline{d}}}_{t}\in\left\{ 0,1\right\} ^{t}}E_{P^{n}}\left[\frac{1}{n}\sum_{i=1}^{n}\left|\hat{\tau}_{t}\left(\text{\ensuremath{\underline{d}}}_{t},H_{it}\right)-\tau_{t}\left(\text{\ensuremath{\underline{d}}}_{t},H_{it}\right)\right|\right]  =O\left(\phi_{n}^{-1}\right).
\end{align*}
(ii) Define 
\begin{align*}
\eta_{t}\left(\text{\ensuremath{\underline{d}}}_{t:T},H_{T}\right) \equiv & \sum_{s=t}^{T}\left\{ \frac{\left(\prod_{\ell=t}^{s}1\left\{ D_{\ell}=d_{\ell}\right\}\right) \gamma_{s}Y_{s}}{\prod_{\ell=t}^{s}e_{\ell}\left(d_{\ell},H_{\ell}\right)}\right\},\\
\hat{\eta}_{t}\left(\text{\ensuremath{\underline{d}}}_{t:T},H_{T}\right)\equiv &  \sum_{s=t}^{T}\left\{ \frac{ \left(\prod_{\ell=t}^{s}1\left\{ D_{\ell}=d_{\ell}\right\}\right) \gamma_{s}Y_{s}}{\prod_{\ell=t}^{s}\hat{e}_{\ell}\left(d_{\ell},H_{\ell}\right)}\right\} .
\end{align*}
For a class of data-generating processes $\widetilde{\MP}_{e}$, there exists
a sequence $\xi_{n}\rightarrow\infty$ such that 
\begin{align*}
\underset{P\in \widetilde{\MP}_{e}}{\sup}\sup_{t\in\left\{ 1,\ldots,T\right\} }\sum_{\text{\ensuremath{\underline{d}}}_{t:T}\in\left\{ 0,1\right\} ^{T-t+1}}E_{P^{n}}\left[\frac{1}{n}\sum_{i=1}^{n}\left|\hat{\eta}_{t}\left(\text{\ensuremath{\underline{d}}}_{t:T},H_{iT}\right)-\eta_{t}\left(\text{\ensuremath{\underline{d}}}_{t:T},H_{iT}\right)\right|\right]  = O\left(\xi_{n}^{-1}\right).
\end{align*}
\end{assumption}

\medskip{}

Note that $E_{P}\left[\tau_{t}\left(\text{\ensuremath{\underline{d}}}_{t},H_{t}\right)\right]=E_{P}\left[\gamma_{t}Y_{t}(\underline{d}_{t})\right]$ and $E_{P}\left[\eta_{t}\left(\text{\ensuremath{\underline{d}}}_{t:T},H_{T}\right)\right]=E_{P}\left[\sum_{s=t}^{T}\gamma_{s}Y_{s}(\underline{d}_{s})\right]$ hold under Assumption \ref{asm:sequential independence} and $n^{-1}\sum_{i=1}^{n}\hat{\tau}_{t}\left(\text{\ensuremath{\underline{d}}}_{t}\right)$ and $n^{-1}\sum_{i=1}^{n}\hat{\eta}_{t}\left(\text{\ensuremath{\underline{d}}}_{t:T}\right)$ are estimators of these, respectively. 
We do not explore lower-level conditions that satisfy Assumption \ref{asm:estimated propensity score}. When the propensity scores are consistently estimated with parametric estimators, they are estimated at rate $n^{-1/2}$. 

When the estimated propensity scores are used, the backward DEWM method
solves the following problem, recursively, from $t=T$ to $1$: 
\begin{align*}
&\hat{g}_{t,e}^{B}  \in\argmax_{g_{t}\in{\cal G}_{t}}\frac{1}{n}\sum_{i=1}^{n}\hat{q}_{t}\left(H_{it},g_{t};\hat{g}_{t+1,e}^{B},\ldots,\hat{g}_{T,e}^{B}\right)\\
\mbox{with}\ \ &\hat{q}_{t}\left(h_{t},g_{t};g_{t+1},\ldots,g_{T}\right)\equiv  \sum_{s=t}^{T}\left\{ \frac{\left(\prod_{\ell=t}^{s}1\left\{ D_{\ell} = g_{\ell}\left(H_{\ell}\right)\right\}\right) \cdot\gamma_{s}Y_{s}}{\prod_{\ell=t}^{s}\hat{e}_{\ell}\left(D_{\ell},H_{\ell}\right)}\right\} .
\end{align*}
We denote by $\hat{g}_{e}^{B} \equiv \left(\hat{g}_{1,e}^{B},\ldots,\hat{g}_{T,e}^{B}\right)$ the DTR obtained by this procedure. 

Similarly, the simultaneous DEWM method solves the following
problem: 
\begin{align*}
&\left(\hat{g}_{1,e}^{S},\dots,\hat{g}_{T,e}^{S}\right) \in\argmax_{g\in{\cal {\cal G}}}\sum_{t=1}^{T}\left[\frac{1}{n}\sum_{i=1}^{n}\hat{w}_{t}^{S}(Z_{i},\text{\ensuremath{\underline{g}}}_{t})\right]
\end{align*}
where $\hat{w}_{t}^{S}(Z_{i},\text{\ensuremath{\underline{g}}}_{t})  \equiv \left\{\left(\prod_{s=1}^{t}1\left\{ D_{is} = g_{s}\left(H_{is}\right)\right\}\right) \cdot\gamma_{t}Y_{it}\right\}/\left\{\prod_{s=1}^{t}\hat{e}_{s}\left(D_{is},H_{is}\right)\right\}$
uses the estimated propensity scores.
We denote the resulting DTR by $\hat{g}_{e}^{S} \equiv \left(\hat{g}_{1,e}^{S},\ldots,\hat{g}_{T,e}^{S}\right)$.

The following theorem shows the uniform convergence rate bounds
on the worst-case average welfare regret for the two estimation methods.

\medskip{}

\begin{theorem}\label{thm:estimated propensity score}

Suppose that Assumptions \ref{asm:sequential independence}, \ref{asm:bounded outcome}, and \ref{asm:overlap}
hold for any distribution $P\in{\cal P}\left(M, \kappa, \MG\right)$ and
Assumption \ref{asm:vc-class} holds for ${\cal {G}}$.\\
 (i) Suppose further that Assumption \ref{asm:estimated propensity score} (i)
holds for any distribution $P\in \MP_{e}$. For the Simultaneous
DEWM method, there holds
\begin{align*}
\sup_{P\in \MP_{e}\bigcap{\cal P}\left(M, \kappa, \MG\right)}E_{P^{n}}\left[W_{{\cal G}}^{\ast}-W\left(\hat{g}_{e}^{S}\right)\right] & \leq C\sum_{t=1}^{T}\left\{ \frac{\gamma_{t}M_{t}}{\prod_{s=1}^{t}\kappa_{s}}\sqrt{\frac{\sum_{s=1}^{t}v_{s}}{n}}\right\} +O\left(\phi_{n}^{-1}\right),
\end{align*}
where $C$ is the same universal constant as that introduced in Theorem
\ref{thm:upper bound}.\\
 (ii) Suppose that Assumption \ref{asm:estimated propensity score} (ii)
holds for any distribution $P\in \widetilde{\MP}_{e}$ and Assumption
\ref{asm:first-best} holds for a pair $\left(P,{\cal {G}}\right)$
for any $P\in{\cal P}\left(M, \kappa, \MG\right)$. Then, for the backward
DEWM method, there holds
\begin{align*}
\sup_{P\in \widetilde{\MP}_{e}\bigcap{\cal P}\left(M, \kappa, \MG\right)}E_{P^{n}}\left[W_{{\cal G}}^{\ast}-W\left(\hat{g}_{e}^{B}\right)\right] &\leq  C\sum_{t=1}^{T}\left\{ \frac{\gamma_{t}M_{t}}{\prod_{s=1}^{t}\kappa_{s}}\sqrt{\frac{\sum_{s=1}^{t}v_{s}}{n}}\right\} \\
 & +C\sum_{t=2}^{T}\frac{2^{t-2}}{\prod_{s=1}^{t-1}\kappa_{s}}\left(\sum_{s=t}^{T}\left\{ \frac{\gamma_{s}M_{s}}{\prod_{\ell=t}^{s}\kappa_{\ell}}\sqrt{\frac{\sum_{\ell=t}^{s}v_{\ell}}{n}}\right\} \right) \\ &+O\left(\xi_{n}^{-1}\right).
\end{align*}
\end{theorem}
\begin{proof}
See Appendix \ref{appendix:proof}.
\end{proof}

\medskip{}

The theorem implies that the convergence rate of the worst-case average regret for each method depends on that of the propensity scores' estimators. If the propensity scores are correctly specified and parametrically estimated, both methods achieve the optimal $n^{-1/2}$-convergence rate of the worst-case average regret.

When the propensity scores are not consistently estimated, the IPW approaches do not consistently estimate the optimal DTRs. We hence propose a doubly robust approach; it is robust to misspecification of either propensity scores or models relevant to outcomes, and can achieve the optimal $n^{-1/2}$-rate of the welfare regret even when nuisance components are nonparametrically estimated. The following remark briefly discusses this approach while technical exposition and theoretical results are presented in Appendix \ref{appendix:DR_estimation}.

\medskip

\begin{remark}[Doubly Robust Approach]\label{rem:doubly_robust}
For the static treatment choice problem, \citet{Athey_Wager_2020} and \cite{Zhou_et_al_2023} show that using the augmented inverse probability weighting (AIPW) estimator of the welfare function can improve the convergence rate of the welfare regret relative to the e-hybrid EWM rule.\footnote{\citet{Nie_et_al_2019} extend this approach to the problem of optimal starting/stopping decision.} In the dynamic setting, we consider extension of the simultaneous maximization approach to doubly robust approach.\footnote{\cite{Sakaguchi_2024} propose a doubly robust method with backward induction for estimating optimal DTRs.} 
The approach presented in Appendix \ref{appendix:DR_estimation} combines the estimated propensity scores and estimators of Q-functions, $$Q_{t}^{\underline{g}_{(t+1):T}}(h_{t},d_{t}) \equiv E_{P}\left[\gamma_{t}Y_{t} + \sum_{s=t+1}^{T}\gamma_{s}\widetilde{Y}_s(\underline{D}_{t},\underline{g}_{(t+1):s})\middle|H_t = h_t,A_t = d_t\right],$$ to construct an AIPW estimator of the welfare function $W(g)$, and then maximizes it over $\MG$ to estimate the optimal DTR. The cross-fitting is also used.
This approach consistently estimates the optimal DTR if either a propensity score or Q-function for each stage is consistently estimated.  
The results in Appendix \ref{appendix:DR_estimation} show that the welfare regret $W_{\MG}^{\ast} - W(\hat{g}^{AIPW})$ converges to $0$ with the optimal rate of $n^{-1/2}$ under mild conditions on the convergence rates of the estimators of the propensity scores and Q-functions. 

This approach, however, faces an optimization challenge. Since $Q_{t}^{\underline{g}_{(t+1):T}}(h_{t},d_{t})$ is specific to a sequence of treatment rules $\underline{g}_{(t+1):T}$, when we apply this approach, we have to estimate $\left\{Q_{t}^{\underline{g}_{(t+1):T}}(h_{t},d_{t})\right\}_{t=1,\ldots,T}$ for every possible DTR $g$ in $\MG$. This is computationally challenging unless the class of DTRs is sufficiently small (e.g., a finite class of a moderate number of DTRs).\footnote{When covariates are exongenous and intermediate outcomes are not used, a doubly robust approach with lower computational cost can be constructed. Appendix \ref{appendix:DR_approach_2} gives details.} 
\end{remark}

\section{Simulation Study \label{sec:simulation}}

We conduct a simulation study to examine the finite sample performance of the proposed methods. We compare the performance of backward DEWM, simultaneous DEWM, and Q-learning.

We consider DGPs that consist of two stages of treatment assignment ($D_{1},D_{2}$), associated
potential outcomes $\left(Y_{1}\left(d_{1}\right),Y_{2}\left(d_{1},d_{2}\right)\right)_{\left\{ d_{1},d_{2}\right\} \in\left\{ 0,1\right\} ^{2}}$,
and a covariate $X_{1}$ observed at the first stage. The potential
outcomes are generated as 
\begin{align*}
Y_{1}\left(d_{1}\right)&=  \phi_{01}+\phi_{11}X_{1}+\left(\psi_{01}+\psi_{11}X_{1}\right)d_{1}+U_{1},\\
Y_{2}\left(d_{1},d_{2}\right) &=  \phi_{02}+\phi_{12}Y_{1}\left(d_{1}\right)+\left(\psi_{02}+\psi_{12}d_{1}+\sum_{j=1}^{3}\psi_{j+1,2}\left(Y_{1}\left(d_{1}\right)\right)^{j}\right)d_{2}+U_{2}
\end{align*}
for $(d_{1},d_{2}) \in \{0,1\}^{2}$. We consider three DGPs labeled DGPs 1-3. In all the DGPs, $X_{1}$,
$U_{1}$, and $U_{2}$ are independently drawn from $N\left(0,1\right)$;
$D_{1}$ and $D_{2}$ are independently drawn from $Ber\left(1/2\right)$;
and $\left(\phi_{01},\phi_{11},\psi_{01},\psi_{11}\right)=\left(0.5,-1.0,1.0,1.5\right)$
and $\left(\phi_{02},\phi_{12},\psi_{02},\psi_{12}\right)=\left(0.5,0.5,0.5,0.5\right)$.
Regarding the other parameters, we set $\left(\psi_{22},\psi_{32},\psi_{42}\right)=\left(0,0,0\right)$
in DGP1, $\left(\psi_{22},\psi_{32},\psi_{42}\right)=\left(1,0,0\right)$
in DGP2, and $\left(\psi_{22},\psi_{32},\psi_{42}\right)=\left(0.3,0.3,-0.4\right)$ in DGP3. We set the target welfare to maximize as $W(g_1,g_2)=E_{P}\left[Y_{2}(g_{1},g_{2})\right]$. The treatment effect of $D_{2}$ does not depend on the past outcome in DGP1, but it does in DGPs 2 and 3.

For the backward and simultaneous DEWM methods, we use a class of
DTRs $\mathcal{G}=\mathcal{G}_{1}\times\mathcal{G}_{2}$ that consists
of the following classes of linear treatment rules: 
\begin{align*}
{\cal G}_{1}& =  \left\{ 1\left\{ \left(1,X_{1}\right)^{\prime}\boldsymbol{\beta}_{1}\geq0\right\} :\boldsymbol{\beta}_{1}=\left(\beta_{01},\beta_{11}\right)^{\prime}\in\mathbb{R}^{2}\right\} ,\\
{\cal G}_{2}& =  \left\{ 1\left\{ \left(1,D_{1},Y_{1}\right)^{\prime}\boldsymbol{\beta}_{2}\geq0\right\} :\boldsymbol{\beta}_{2}=\left(\beta_{02},\beta_{12},\beta_{22}\right)^{\prime}\in\mathbb{R}^{3}\right\} .
\end{align*}
$\MG_{2}$ contains the first-best rule under DGPs 1 and 2 but not under DGP3. Thus, the backward DEWM method can consistently estimate the optimal DTR under DGPs 1 and 2 but cannot under DGP3. We solve the optimization problems for each DEWM method through MILPs as discussed in Remark \ref{rem:MILP}.

For Q-learning, we assume that the conditional outcomes are specified as 
\begin{align*}
E\left[Y_{1}\mid H_{1},D_{1};\boldsymbol{\alpha}_{1},\boldsymbol{\gamma}_{1}\right]&=  \alpha_{01}+\alpha_{11}X_{1}+\left(\gamma_{01}+\gamma_{11}X_{1}\right)D_{1},\\
E\left[Y_{2}\mid H_{2},D_{2};\boldsymbol{\alpha}_{2},\boldsymbol{\gamma}_{2}\right]&=  \alpha_{02}+\alpha_{12}Y_{1}+\left(\gamma_{02}+\gamma_{12}D_{1}+\gamma_{22}Y_{1}\right)D_{2},
\end{align*}
where $\boldsymbol{\alpha}_{t}^{\prime}=\left(\alpha_{0t},\alpha_{1t}\right)^{\prime}$
for each $t=1,2$, $\boldsymbol{\gamma}_{1}^{\prime}=\left(\gamma_{01},\gamma_{11}\right)^{\prime}$,
and $\boldsymbol{\gamma}_{2}^{\prime}=\left(\gamma_{02},\gamma_{12},\gamma_{22}\right)^{\prime}$.
This specification is correct under DGPs 1 and 2 but is not under
DGP3.

Table \ref{table:Monte Carlo Simulation Results} presents the results of 500 simulations with sample sizes $n=$
200, 500, and 800. The table shows the mean and median welfare achieved by each estimated DTR calculated with 3,000 observations randomly drawn from the same DGP. The results show that Q-learning performs better than the backward and simultaneous DEWM methods in DGPs 1 and 2 in terms of the population mean welfare. However, both the backward and simultaneous DEWM methods exhibit superior performance to Q-learning in DGP3, where the outcome model used by Q-learning is misspecified. In DGPs 1 and 2, the backward and simultaneous DEWM methods demonstrate similar welfare performance. However, in DGP3, the simultaneous DEWM method achieves higher welfare than the backward DEWM method. 
Table \ref{table:Monte Carlo Simulation Results} also presents the average CPU time to calculate DTR per simulation iteration.\footnote{We use Julia version 1.8.1 with Gurobi Optimizer version 9.5.2. The hardware is 12th Gen Intel(R) Core(TM) i7-12700 2.10 GHz.} The simultaneous DEWM takes the longest time but remains feasible at these scales of simulation. Appendix \ref{appendix:simulation_results} provides additional simulation results for DGPs where $D_1$ and $D_2$ are not independent of $U_1$ and $U_2$ (i.e., the sequential independence assumption (Assumption \ref{asm:sequential independence}) does not hold).

\begin{table}[hh]
\centering \caption{Monte Carlo simulation results}
\label{table:Monte Carlo Simulation Results}
\scalebox{0.8}{ 
\begin{tabular}{cccccccccccccccc}
\hline 
 &  &  & n=200  &  &  &  & & n=500  &  &  &  & & n=800  &  & \tabularnewline
\hline 
 & DGP  & Mean  & Med  & SD & Time  &  & Mean  & Med  & SD & Time  &  & Mean  & Med  & SD & Time  \tabularnewline
\hline 
Q-learning  & 1  & 2.27  & 2.27  & 0.04 & 0.01 &  & 2.28  & 2.28  & 0.07 & 0.01 &  & 2.27  & 2.27  & 0.07  & 0.01 \tabularnewline
B-DEWM      & 1  & 2.05  & 2.15  & 0.27 & 0.75 &  & 2.15  & 2.21  & 0.21 & 5.35 &  & 2.15  & 2.21  & 0.23  & 13.69 \tabularnewline
S-DEWM      & 1  & 2.01  & 2.12  & 0.29 & 9.13 &  & 2.15  & 2.21  & 0.25 & 61.95 &  & 2.18  & 2.22  & 0.18  & 209.82 \tabularnewline
\hdashline 
Q-learning  & 2  & 3.97  & 3.97  & 0.07 & 0.01 &  & 3.98  & 3.97  & 0.12 & 0.01 &  & 3.98  & 3.98  & 0.12  & 0.01 \tabularnewline
B-DEWM      & 2  & 3.56  & 3.72  & 0.50 & 0.60 &  & 3.71  & 3.78  & 0.38 & 3.17 &  & 3.69  & 3.79  & 0.43  & 9.31 \tabularnewline
S-DEWM      & 2  & 3.50  & 3.67  & 0.53 & 6.11 &  & 3.71  & 3.80  & 0.37 & 38.54 &  & 3.75  & 3.82  & 0.39  & 97.29 \tabularnewline
\hdashline 
Q-learning  & 3  & 1.64  & 1.63  & 0.11 & 0.01 &  & 1.63  & 1.61  & 0.15 & 0.01 &  & 1.63  & 1.61  & 0.15  & 0.01 \tabularnewline
B-DEWM      & 3  & 1.77  & 1.80  & 0.13 & 0.51 &  & 1.79  & 1.80  & 0.13 & 4.68 &  & 1.79  & 1.81  & 0.09  & 14.32 \tabularnewline
S-DEWM      & 3  & 1.86  & 1.89  & 0.15 & 6.32 &  & 1.89  & 1.90  & 0.17 & 42.26 &  & 1.91  & 1.92  & 0.13  & 143.70 \tabularnewline
\hline 
\end{tabular}
}
\begin{tablenotes} \footnotesize
\item Notes: Mean and Med represent the mean
and median of the population mean welfares achieved by the estimated
DTRs across the simulations, respectively. SD is the standard deviation of the population mean welfares across the simulations.
The population mean welfare is calculated using 3,000 observations
randomly drawn from the corresponding DGP. B-DEWM and S-DEWM represent the Backward and Simultaneous DEWM methods, respectively. The columns of ``Time'' show average CPU time to estimate DTR per iteration for each method, DGP, and sample size.
\end{tablenotes} 
\end{table}

\section{Empirical Application \label{sec:empirical application}}

We apply the proposed methods to data from Project STAR  \citep[e.g.,][]{krueger_1999, Gerber_et_al_2001, Schanzenbach_2006, Chetty_et_al_2011}. In this experimental project, out of 1,346 kindergarten students not belonging to small classes, 672 students were randomly allocated to regular-size classes with a full-time teacher aide, and the others were allocated to regular-size classes without a teacher aide. Upon their progression to grade 1, the enrolled students were randomly shuffled into regular class-size classes with or without a teacher aide and remained in the allocated classes until the end of grade 3.

We study optimal allocation of students to two types of classes (regular-size classes with or without a teacher aide) in grades $K$ and $1$, based on their socioeconomic information and intermediate academic performance.\footnote{We focus on allocating students to regular-size classes with a teacher aide, rather than small-size classes, because the allocation of students to regular-size classes with or without a teacher aide in the experiment matches the sequential randomization design; however, the allocation to small-size classes in the experiment does not.} 
We aim to maximize the population average of scores from mathematics test that students took at the end of grade 1.\footnote{We focus on the test score at the end of grade 1 rather than at the end of grade 3 because we found attending a class with a teacher aide in kindergarten has little effect on the test score at the end of grade 3 even when treatment effect heterogeneity is considered.}
We set the first and second stages ($t=1$ and $2$) to grades $K$ and $1$, respectively. The treatment variable $D_{t}$, for $t=1,2$, takes the value one if the student is assigned to a class with a teacher aide at stage $t$ and zero otherwise. The potential intermediate outcome $Y_{1}(d_1)$ and final outcome $Y_{2}(d_1,d_1)$ represent the mathematics test scores at the end of grades K and 1 for treatments $d_1$ and $d_2$.

Because it is plausible that assignment to a class with a teacher aide outperforms assignment to a class without one, we consider the cost of having a teacher aide.
We consider a fictitious cost for teacher aide. 
This cost for each grade is set as $c \equiv E[Y_2(1,1) - Y_2(0,0) ]/2$, half of the expected welfare gain of assigning every student to teacher-aide classes in both grades.  We use the estimate of $E[Y_2(1,1) - Y_2(0,0) ]/2$, which is $13.57$, for the cost $c$. Letting $Y_{2}^{c}(d_1,d_2) := Y_2(d_1,d_2)  - c(d_1 + d_2)$ represent the individual welfare contribution, we aim to maximize welfare defined as
$$W(g_1,g_2) = E\left[
    \sum_{(d_1,d_2) \in \{0,1\}^2} Y_{2}^{c}(d_1,d_2)\cdot 1\left\{g_1(H_1)=d_1,g_2(H_2(d_1))=d_2\right\}\right],$$ which represents the average of the total test score at the end of grade 1 with subtraction of the cost.
    Note that in this setting, uniformly assigning every student to a class with a teacher aide in both grades results in zero welfare gain.

The socioeconomic information we use for treatment choice are the qualification for free or reduced-price school lunches and school location (rural or non-rural).\footnote{While we have access to student information such as sex and race, using such information in treatment choice is discriminatory and prohibited.}
A binary variable $X_{Lunch,t}$ takes $1$ if the student is eligible for free or reduced-price school lunch at stage $t$ and $0$ otherwise. A binary variable $X_{Rural,t}$ takes $1$ if the student attends a school located in a rural area at stage $t$ and $0$ otherwise.

We employ a set of class-allocation policies represented by treatment rules $\mathcal{G}=\mathcal{G}_{1}\times\mathcal{G}_{2}$, where $\mathcal{G}_{1}$ and $\mathcal{G}_{2}$ constitute a class of linear treatment rules:\footnote{It is testable whether $\MG_2$ contains the first-best rule or not. For example, we can estimate the first-best rule for the second stage as $\hat{g}_{2}^{\ast,FB}(h_{2}) = 1\{\hat{\tau}_{2}(h_2) \geq 2\}$ with $\hat{\tau}_{2}(h_2)$ being a (nonparametric) estimator of the conditional average treatment effect $\tau_{2}(h_2) = E[Y_{2}(D_1,1) - Y_2(D_1,0) | H_2 = h_2]$. We can then check whether $\MG_2$ contains the first-best rule by examining if the optimal policy in $\MG_2$ achieves the same expected outcome value as $\hat{g}_{2}^{\ast,FB}(h_{2})$.}
\begin{align*}
{\cal G}_{1}& =  \left\{ 1\left\{ \beta_{1}x_{Lunch,1} + \beta_{2}x_{Rural,1} \geq c_{1}\right\} :\beta_{1} \geq 0,\ \left(\beta_{2},c_{1}\right)^{\prime}\in\mathbb{R}^{2}  \right\} ,\\
{\cal G}_{2} &=  \left\{ 1\left\{ \gamma_{1}x_{Lunch,2}+ \gamma_{2}x_{Rural,2} +\gamma_{3}\left(1-d_{1}\right)y_{1}+\gamma_{4}d_{1}y_{1}\geq c_{2}\right\} :\gamma_{1}\geq 0,\ \left(\gamma_{2},\gamma_{3},\gamma_{4},c_{2}\right)^{\prime}\in\mathbb{R}^{4}\right\} .
\end{align*}
In the formulations of $\MG_{1}$ and $\MG_{2}$, the coefficients of $x_{Lunch,1}$ and $x_{Lunch,2}$ are constrained to be non-negative. This ensures that students eligible for free or reduced-price school lunches are not less likely to be allocated to a class with a teacher aide, given that the other information is fixed.
The interaction terms $\left(1-d_{1}\right)y_{1}$ and $d_{1}y_{1}$ in $\mathcal{G}_{2}$ enable the eligibility score to evaluate the intermediate outcome differently based on class allocation at kindergarten.
We solve the optimization problems for the backward and simultaneous DEWM through MILPs as discussed in Remark \ref{rem:MILP}. 

For a DTR $g\in\mathcal{G}$, we define the welfare gain of $g$ as $W\left(g\right)-E[Y_{2}^{c}(0,0)]$, the welfare increase achieved by allocating students according to the DTR $g$ rather than allocating every student to regular classes without teacher aides at all stages.\footnote{Two factors enhance the students' academic achievement in classroom allocation: optimal matching between each student and classroom type (with or without a teacher aide) and peer effects among students. The optimal DTR considered here exploits the former but not the latter, as it does not utilize peer effects among students to determine classroom allocation.}
Applying the backward and simultaneous DEWM methods, we estimate the optimal DTR over $\mathcal{G}$ and its welfare gain as well as the treatment ratio at each stage. 
To avoid overfitting biases, we adopt two-fold random sample splitting with a fixed seed: one third of the sample is used as the training set to estimate the optimal DTRs, and the remaining is used as the test set to estimate the welfare gains and treatment ratios.

The DTR estimated by the backward DEWM method is $\hat{g}^{B}=\left(\hat{g}_{1}^{B},\hat{g}_{2}^{B}\right)$ where $\hat{g}_{1}^{B}(h_{1})=  1$ and $\hat{g}_{2}^{B}(h_{2})=  1\left\{-0.104\left(1-d_{1}\right)y_{1}+0.382d_{1}y_{1}\geq - 0.449\right\}$.
The DTR estimated by the simultaneous DEWM method is $\hat{g}^{S}=\left(\hat{g}_{1}^{S},\hat{g}_{2}^{S}\right)$ where $\hat{g}_{1}^{S}(h_{1}) =  1\left\{x_{Rural,1}= 1\right\}$ and $\hat{g}_{2}^{S}(h_{2}) = 1\left\{ 0.985x_{Rural,2} + 0.148d_{1}y_{1}\geq 0\right\}$. 
$\hat{g}_{1}^B$ assigns every student to a class with a teacher aide in grade K, while $\hat{g}_{1}^S$ assigns only students in rural areas to classes with teacher aides in grade K. Under both treatment rules $\hat{g}_{2}^{B}$ and $\hat{g}_{2}^{S}$ for grade 1, a student who attends a class with a teacher aide and attains a high test score in grade K is more likely to be assigned to a class with a teacher aide in grade 1.

Table \ref{table:Estimated welfare gain} reports the estimated welfare gains and shares of the population to be treated at each stage for the estimated DTRs $\hat{g}^{B}$ and $\hat{g}^{S}$ and three uniform DTRs $\left(g_{1},g_{2}\right)=\left(1,0\right),\left(0,1\right),\left(1,1\right)$.
\footnote{With some abuse of notation, we denote by $(g_1,g_2)=(d_1,d_2)$ the uniform DTR that allocates every student to class types $d_1$ and $d_2$ in stages 1 and 2, respectively.} For example, the DTR $\left(g_{1},g_{2}\right)=\left(1,0\right)$ assigns every student to a class with a teacher aide in kindergarten but assigns none in grade 1. The results indicate that both the backward and simultaneous DEWM methods lead to higher welfare gains than all the uniform DTRs.  

\begin{table}[h]
\begin{centering}
\caption{Estimated welfare gains}
\label{table:Estimated welfare gain}
\begin{tabular}{cccc}
\hline 
  & \multicolumn{2}{c}{Share of population to be treated} &  \\ \cline{2-3}
Dynamic treatment regime & 1st stage  & 2nd stage  & Estimated welfare gain \\
\hline 
$\left(g_{1},g_{2}\right)=\left(1,0\right)$  & $1$  & $0$  & $1.93$\\

$\left(g_{1},g_{2}\right)=\left(0,1\right)$  & $0$  & $1$  & $3.58$\\

$\left(g_{1},g_{2}\right)=\left(1,1\right)$  & $1$  & $1$  & $0.00$\\

$\left(\hat{g}_{1}^{B},\hat{g}_{2}^{B}\right)$  & $1.0$  & $0.50$  & $16.68$\\

$\left(\hat{g}_{1}^{S},\hat{g}_{2}^{S}\right)$  & $0.67$  & $0.73$  & $14.91$\\
\hline 
\end{tabular}
\par\end{centering}
\centering{}
\begin{tablenotes} 
\footnotesize 
\item Notes: The SD of $Y_{2}$ in the sample is 83.27. We use the two-fold sample splitting with a fixed seed. The training sample is used to estimate the DTRs $\hat{g}^{B}$ and $\hat{g}^{S}$. The test sample is used to estimate shares of population to be treated and welfare gains of $\hat{g}^{B}$ and $\hat{g}^{S}$.
\end{tablenotes} 
\end{table}

Next, we consider the decision problem of when each student should begin attending a class with a teacher aide. To this aim, we impose a constraint $g_{2}\left(h_{2}\right)\geq d_{1}$ for all $h_{2}\in\mathcal{H}_{2}$ on $\MG_2$. Under this constraint, the DTR estimated by the backward DEWM method is $\hat{g}^{B}=\left(\hat{g}_{1}^{B},\hat{g}_{2}^{B}\right)$ with $\hat{g}_{1}^{B}(h_{1})  = 0$ and $\hat{g}_{2}^{B}(h_{2}) =  1\left\{ -0.104\left(1-d_{1}\right)y_{1} + 0.382 d_{1}y_{1}\geq -0.449\right\}$;  
the DTR estimated by the simultaneous DEWM method is $\hat{g}^{S}=\left(\hat{g}_{1}^{S},\hat{g}_{2}^{S}\right)$ with $\hat{g}_{1}^{S}(h_{1})  =  1\left\{ x_{Rural,1}= 1\right\}$ and $\hat{g}_{2}^{S}(h_{2})  =  1$.

Table \ref{table:Estimated welfare gain for start time decision problem} reports the estimated welfare gains and shares of population to be treated by $\hat{g}^{B}$ and $\hat{g}^{S}$ and two uniform DTRs $(g_{1},g_{2})=(0,1),(1,1)$, which satisfy the monotonicity constraint. 
Both the Simultaneous and backward DEWM methods lead to higher welfare gains than the uniform DTRs. The simultaneous DEWM method leads to a slightly higher welfare gain than the backward DEWM method.

\begin{table}[h]
\begin{centering}
\caption{Estimated welfare gains for the start-time decision problem}
\label{table:Estimated welfare gain for start time decision problem}
\begin{tabular}{cccc}
\hline 
 & \multicolumn{2}{c}{Share of population to be treated} & \\ \cline{2-3}
Dynamic treatment regime & 1st stage  & 2nd stage  & Estimated welfare gain \\
\hline 
$\left(g_{1},g_{2}\right)=\left(0,1\right)$  & $0$  & $1$  & $3.58$\\

$\left(g_{1},g_{2}\right)=\left(1,1\right)$  & $1$  & $1$  & $0.00$\\

$(\hat{g}_{1}^{B},\hat{g}_{2}^{B})$  & $0.0$  & $0.50$  & $9.05$\\

$(\hat{g}_{1}^{S},\hat{g}_{2}^{S})$  & $0.67$ & $1.0$ & $10.43$ \\
\hline 
\end{tabular}
\par\end{centering}
\centering{}
\begin{tablenotes} 
\footnotesize  
\item Notes: The SD of $Y_{2}$ in the sample is 83.27. We use the two-fold sample splitting with a fixed seed. The training sample is used to estimate the DTRs $\hat{g}^{B}$ and $\hat{g}^{S}$. The test sample is used to estimate shares of population to be treated and welfare gains of $\hat{g}^{B}$ and $\hat{g}^{S}$. The welfare gains of the uniform policies are estimated with the whole sample.
\end{tablenotes} 
\end{table}


\section{Conclusion \label{sec:conclusion}}

This study proposes empirical methods to estimate the optimal DTR over a pre-specified class of feasible DTRs based on the EWM approach. We proposed two estimation methods, the backward DEWM and simultaneous DEWM methods, which estimate the optimal DTR through backward induction and simultaneous maximization, respectively. The former is computationally efficient, but it may not consistently estimate the optimal DTR when the class of feasible DTRs does not include the first-best rule at all stages except for the first stage. Conversely, the latter method can consistently estimate the optimal DTR irrespective of the feasibility of the first-best rule, though it is computationally less efficient. These methods can accommodate exogenous constraints on the class of DTRs and specify different types of dynamic treatment choice problems. 
We show that each method can achieve the optimal $n^{-1/2}$ rate of convergence of the regret in the experimental data setting. We also modified the simultaneous DEWM to accommodate intertemporal budget/capacity constraints.


\bibliographystyle{ecta}
\bibliography{ref_DTR,ref_DTR_appendix,ref_surrogate_loss}

\newpage

\appendix
\part*{Appendix}


\section{Proof of Theorem \ref{thm:upper bound} \label{appendix:proofs-1}}

This appendix presents the proof of Theorem \ref{thm:upper bound} along with some auxiliary lemmas.
Let $\mbox{SG}\left(\mathcal{F}\right)\equiv \{\mbox{SG}\left(f\right): f\in \mathcal{F} \}$ be a collection of subgraphs over a class of functions $\mathcal{F}$, where the subgraph of a real-valued function
$f:\mathcal{Z}\rightarrow\mathbb{R}$ is defined as the set $\mbox{SG}\left(f\right)\equiv  \left\{ \left(z,t\right)\in\MZ \times\mathbb{R}:t\leq f\left(z\right)\right\}$.
We consider the VC-dimension of $\mbox{SG}(\MF)$ as a complexity measure of $\mathcal{F}$, where its definition is given in Appendix \ref{appendix:proof}.

The following lemma establishes the link between the VC-dimension of a class of feasible DTRs and the VC-dimension of a class of subgraphs of functions on ${\cal Z}$.

\begin{lemma}\label{lem:vc_subclass} 
Suppose that Assumption \ref{asm:vc-class} holds. Let $r:\MZ \rightarrow \Real$
be any function.
For any integers $s$ and $t$ with $1\leq s\leq t\leq T$, a class of functions
from ${\cal Z}$ to $\mathbb{R}$ 
\begin{align*}
\MF_{s:t}\equiv  \left\{ f\left(z\right)=1\left\{ g_{s}\left(h_{s}\right)=d_{s},\ldots,g_{t}\left(h_{t}\right)=d_{t}\right\} \cdot r(z):\left(g_{s},\ldots,g_{t}\right)\in\MG_{s}\times\cdots\times\mathcal{G}_{t}\right\} 
\end{align*}
is a VC-subgraph class of functions with $VC(\mbox{SG}(\MF_{s:t})) \leq \sum_{j=s}^{t}v_{j}$. 
\end{lemma}

\begin{proof}
    The proof is presented in Appendix \ref{appendix:proof}.
\end{proof}

The next lemma, which corresponds to Lemma A.4 of \cite{Kitagawa_Tetenov_2018a}, gives a uniform upper bound for the mean of a supremum of centered empirical processes indexed by a VC-subgraph class of functions. This is a fundamental result in the literature on empirical process theory and its proof can be found, for example, in \cite{van_der_Vaart_Wellner_1996} and \cite{Kitagawa_Tetenov_2018a}. 

\medskip

\begin{lemma} \label{lem:KT} (Lemma A.4 in \cite{Kitagawa_Tetenov_2018a}) Let ${\cal F}$ be a class of uniformly bounded functions on $\MZ$, that is, there exists $\bar{F}<\infty$ such that
$\left\Vert f\right\Vert _{\infty}\leq\bar{F}$ for all $f\in{\cal F}$.
Assume that ${\cal F}$ is a VC-subgraph of functions with VC-dimension $v<\infty$.
Then there is a universal constant $C$ such that 
\begin{align}
E_{P^{n}}\left[\sup_{f\in{\cal F}}\left|E_{n}\left(f\right)-E_{P}\left(f\right)\right|\right] & \leq C\bar{F}\sqrt{\frac{v}{n}}\nonumber
\end{align}
holds for all $n\geq1$. \end{lemma}

\medskip

Before proceeding to the proofs of the main theorems, we define 
\begin{align}
\tilde{Q}_{t}\left(g_{t},\ldots,g_{T}\right) & \equiv E_{P}\left[q_{t}\left(Z,g_{t},\ldots,g_{T}\right)\right] \notag \\
 & =E_{P}\left[\sum_{s=t}^{T}\left\{ \frac{(\prod_{\ell=t}^{s}1\left\{ D_{\ell}=g_{\ell}\left(H_{\ell}\right)\right\}) \gamma_{s}Y_{s}}{\prod_{\ell=t}^{s}e_{\ell}\left(D_{\ell},H_{\ell}\right)}\right\} \right],
 \label{eq:Q-function with propensity scofre}\\
 \tilde{Q}_{nt}\left(g_{t},\ldots,g_{T}\right) & \equiv E_{n}\left[q_{t}\left(Z,g_{t},\ldots,g_{T}\right)\right]\notag\\
 & =E_{n}\left[\sum_{s=t}^{T}\left\{ \frac{(\prod_{\ell =t}^{s}1\left\{ D_{\ell} = g_{\ell}\left(H_{\ell}\right)\right\}) \gamma_{s}Y_{s}}{\prod_{\ell =t}^{s}e_{\ell}\left(D_{\ell},H_{\ell}\right)}\right\} \right].\notag
\end{align}
We further define 
\begin{align*}
 \Delta\tilde{Q}_{t}&\equiv\tilde{Q}_{t}\left(g_{t}^{\ast},\ldots,g_{T}^{\ast}\right)-\tilde{Q}_{t}\left(\hat{g}_{t}^{B},\ldots,\hat{g}_{T}^{B}\right), \\   
 \Delta \tilde{Q}_{t}^{\dagger}&\equiv\tilde{Q}_{t}\left(g_{t}^{\ast},\hat{g}_{t+1}^{B},\ldots,\hat{g}_{T}^{B}\right)-\tilde{Q}_{t}\left(\hat{g}_{t}^{B},\ldots,\hat{g}_{T}^{B}\right).
\end{align*}

The following lemma will be used in the proof of Theorem \ref{thm:upper bound} (ii) for the backward DEWM method.

\medskip

\begin{lemma}\label{lem:backward DEWM} Suppose that Assumptions
\ref{asm:sequential independence}, \ref{asm:overlap}, and \ref{asm:first-best} hold for a pair $(P,\MG)$. Then the following hold: (i) for any $t=1,\ldots,T-1$ and $s=t+1,\ldots,T$, 
\begin{align*}
\tilde{Q}_{t}\left(g_{t}^{\ast},\ldots,g_{T}^{\ast}\right)-\tilde{Q}_{t}\left(g_{t}^{\ast},\ldots,g_{s}^{\ast},\hat{g}_{s+1}^{B},\ldots,\hat{g}_{T}^{B}\right) & \leq\frac{1}{\prod_{\ell =t}^{s}\kappa_{\ell}}\Delta\tilde{Q}_{s+1};
\end{align*}
(ii) 
\begin{align*}
\Delta\tilde{Q}_{1} & \leq \Delta \tilde{Q}_{1}^{\dagger}+\sum_{s=1}^{T-1}\frac{2^{s-1}}{\prod_{t=1}^{s}\kappa_{t}}\Delta \tilde{Q}_{s+1}^{\dagger}.
\end{align*}

\end{lemma}

\begin{proof}
 
\noindent (i) Let $\tilde{Q}_{t}\left(g_{t},\ldots,g_{T};h_t\right) \equiv E_{P}\left[q_{t}\left(Z,g_{t},\ldots,g_{T}\right)| H_t = h_t\right]$.
For any integers $s$ and $t$ such that $1 \leq t <s \leq T$, it follows that
\begin{align*}
 & \tilde{Q}_{t}\left(g_{t}^{\ast},\ldots,g_{T}^{\ast}\right)-\tilde{Q}_{t}\left(g_{t}^{\ast},\ldots,g_{s}^{\ast},\hat{g}_{s+1}^{B},\ldots,\hat{g}_{T}^{B}\right)\notag \\
&=  E_{P}\left[\frac{\prod_{\ell =t}^{s}1\left\{ D_{\ell}=g_{\ell}^{\ast}\left(H_{\ell}\right)\right\} }{\prod_{\ell =t}^{s}e_{\ell}\left(D_{\ell},H_{\ell}\right)}\left(\tilde{Q}_{s+1}\left(g_{s+1}^{\ast},\ldots,g_{T}^{\ast};H_{s+1}\right)-\tilde{Q}_{s+1}\left(\hat{g}_{s+1}^{B},\ldots,\hat{g}_{T}^{B};H_{s+1}\right) \right)\right] \\
&\leq  \frac{1}{\prod_{\ell =t}^{s}\kappa_{\ell}}E_{P}\left[\tilde{Q}_{s+1}\left(g_{s+1}^{\ast},\ldots,g_{T}^{\ast};H_{s+1}\right)-\tilde{Q}_{s+1}\left(\hat{g}_{s+1}^{B},\ldots,\hat{g}_{T}^{B};H_{s+1}\right)\right]\\
&=  \frac{1}{\prod_{\ell =t}^{s}\kappa_{\ell}}\Delta\tilde{Q}_{s+1},
\end{align*}
where the first equality follows from Assumption \ref{asm:sequential independence} and the inequality follows from Assumption \ref{asm:overlap} and because
$\tilde{Q}_{s+1}\left(g_{s+1}^{\ast},\ldots,g_{T}^{\ast};H_{s+1}\right)-\tilde{Q}_{s+1}\left(\hat{g}_{s+1}^{B},\ldots,\hat{g}_{T}^{B};H_{s+1}\right) \geq0$
holds a.s. under Assumptions \ref{asm:sequential independence} and \ref{asm:first-best}. 

\medskip
\noindent
(ii) Note that 
\begin{align*}
\Delta\tilde{Q}_{T} & =\tilde{Q}_{T}\left(g_{T}^{\ast}\right)-\tilde{Q}_{T}\left(\hat{g}_{T}^{B}\right)=\Delta \tilde{Q}_{T}^{\dagger}.
\end{align*}
Then, for $t=T-1$, we have
\begin{align*}
\Delta\tilde{Q}_{T-1} & =\tilde{Q}_{T-1}\left(g_{T-1}^{\ast},g_{T}^{\ast}\right)-\tilde{Q}_{T-1}\left(\hat{g}_{T-1}^{B},\hat{g}_{T}^{B}\right)\\
 & =\tilde{Q}_{T-1}\left(g_{T-1}^{\ast},g_{T}^{\ast}\right)-\tilde{Q}_{T-1}\left(g_{T-1}^{\ast},\hat{g}_{T}^{B}\right)+\tilde{Q}_{T-1}\left(g_{T-1}^{\ast},\hat{g}_{T}^{B}\right)-\tilde{Q}_{T-1}\left(\hat{g}_{T-1}^{B},\hat{g}_{T}^{B}\right)\\
 & \leq\frac{1}{\kappa_{T-1}}\Delta \tilde{Q}_{T}^{\dagger}+\Delta \tilde{Q}_{T-1}^{\dagger},
\end{align*}
where the inequality follows from Lemma \ref{lem:backward DEWM} (i).

 Generally, for any $k=1,\ldots,T-1$, it follows that 
\begin{align*}
\Delta\tilde{Q}_{T-k} & =\tilde{Q}_{T-k}\left(g_{T-k}^{\ast},\ldots,g_{T}^{\ast}\right)-\tilde{Q}_{T-k}\left(\hat{g}_{T-k}^{B},\ldots,\hat{g}_{T}^{B}\right)\\
 & =\sum_{s=T-k}^{T}\left[\tilde{Q}_{T-k}\left(g_{T-k}^{\ast},\ldots,g_{s}^{\ast},\hat{g}_{s+1}^{B},\ldots,\hat{g}_{T}^{B}\right)-\tilde{Q}_{T-k}\left(g_{T-k}^{\ast},\ldots,g_{s-1}^{\ast},\hat{g}_{s}^{B},\ldots,\hat{g}_{T}^{B}\right)\right]\\
 & \leq\sum_{s=T-k}^{T}\left[\tilde{Q}_{T-k}\left(g_{T-k}^{\ast},\ldots,g_{T}^{\ast}\right)-\tilde{Q}_{T-k}\left(g_{T-k}^{\ast},\ldots,g_{s-1}^{\ast},\hat{g}_{s}^{B},\ldots,\hat{g}_{T}^{B}\right)\right]\\
 & =\sum_{s=T-k+1}^{T}\left[\tilde{Q}_{T-k}\left(g_{T-k}^{\ast},\ldots,g_{T}^{\ast}\right)-\tilde{Q}_{T-k}\left(g_{T-k}^{\ast},\ldots,g_{s-1}^{\ast},\hat{g}_{s}^{B},\ldots,\hat{g}_{T}^{B}\right)\right]+\Delta \tilde{Q}_{T-k}^{\dagger}\\
 & \leq\sum_{s=T-k+1}^{T}\frac{1}{\prod_{\ell=T-k}^{s-1}\kappa_{\ell}}\Delta\tilde{Q}_{s}+\Delta \tilde{Q}_{T-k}^{\dagger},
\end{align*}
where the second line follows by taking a telescope sum;
the third line follows from the fact that $\left(g_{s+1}^{\ast},\ldots,g_{T}^{\ast}\right)$
maximizes $\tilde{Q}_{T-k}\left(g_{T-k}^{\ast},\ldots,g_{s}^{\ast},\cdot,\ldots,\cdot\right)$
over ${\cal G}_{s+1}\times\cdots\times{\cal G}_{T}$ under Assumption
\ref{asm:first-best}; the last line follows from Lemma \ref{lem:backward DEWM} (i). 

Then, recursively, the following hold:
\begin{align*}
\Delta\tilde{Q}_{T-1} & \leq\frac{1}{\kappa_{T-1}}\Delta\tilde{Q}_{T}+\Delta \tilde{Q}_{T-1}^{\dagger}=\frac{1}{\kappa_{T-1}}\Delta \tilde{Q}_{T}^{\dagger} + \Delta \tilde{Q}_{T-1}^{\dagger},\\
\Delta\tilde{Q}_{T-2} & \leq\frac{1}{\kappa_{T-2}}\Delta\tilde{Q}_{T-1}+\frac{1}{\kappa_{T-2}\kappa_{T-1}}\Delta\tilde{Q}_{T}+ \Delta \tilde{Q}_{T-2}^{\dagger}\\
 & \leq\frac{2}{\kappa_{T-2}\kappa_{T-1}}\Delta \tilde{Q}_{T}^{\dagger}+\frac{1}{\kappa_{T-2}}\Delta \tilde{Q}_{T-1}^{\dagger}+ \Delta \tilde{Q}_{T-2}^{\dagger},\\
 & \vdots\\
\Delta\tilde{Q}_{T-k} & \leq\sum_{s=1}^{k}\frac{2^{k-s}}{\prod_{t=T-k}^{T-s}\kappa_{t}}\Delta \tilde{Q}_{T-s+1}^{\dagger}+\Delta \tilde{Q}_{T-k}^{\dagger}.
\end{align*}
Therefore, when $k=T-1$, we have 
\begin{align*}
\Delta\tilde{Q}_{1} & \leq\Delta \tilde{Q}_{1}^{\dagger}+\sum_{s=1}^{T-1}\frac{2^{T-1-s}}{\prod_{t=1}^{T-s}\kappa_{t}}\Delta \tilde{Q}_{T-s+1}^{\dagger}\\
 & =\Delta \tilde{Q}_{1}^{\dagger}+\sum_{s=1}^{T-1}\frac{2^{s-1}}{\prod_{t=1}^{s}\kappa_{t}}\Delta \tilde{Q}_{s+1}^{\dagger}.
\end{align*}
\end{proof}

\medskip

We are now prepared to give the proof of Theorem \ref{thm:upper bound}. We first give the proof for the simultaneous DEWM method.

\medskip

\begin{proof}[Proof of Theorem \ref{thm:upper bound} (i).]
Let $P\in{\cal P}\left(M, \kappa, \MG\right)$ be fixed. Define $W_{t}(\text{\ensuremath{\underline{g}}}_{t}) \equiv E_{P}\left[\gamma_t \widetilde{Y}_t\left(\underline{g}_t \right)\right]$. Note that $W_{t}(\text{\ensuremath{\underline{g}}}_{t})=E_{P}\left[w_{t}^{S}(Z,\text{\ensuremath{\underline{g}}}_{t})\right]$ holds under Assumption \ref{asm:sequential independence}, where $w_{t}^{S}(Z,\underline{g}_{t})$ is defined in Section \ref{sec:simultaneous DEWM}. Note also that
$W\left(g\right)=\sum_{t=1}^{T}W_{t}(\text{\ensuremath{\underline{g}}}_{t})$. 
Let $W_{nt}(\underline{g}_{t})$
and $W_{n}\left(g\right)$ be defined as 
$W_{nt}(\text{\ensuremath{\underline{g}}}_{t})\equiv \frac{1}{n}\sum_{i=1}^{n}w_{t}^{S}(Z_{i},\text{\ensuremath{\underline{g}}}_{t})$ and $
 W_{n}\left(g\right)\equiv\sum_{t=1}^{T}W_{nt}(\text{\ensuremath{\underline{g}}}_{t})$, respectively.

It follows, for any $g\in \MG$, that 
\begin{align}
E_{P^n}\left[W\left(g\right)-W\left(\hat{g}^{S}\right)\right] & =E_{P^n}\left[W\left(g\right)-W_{n}\left(g\right)\right] + E_{P^n}\left[W_{n}\left(g\right)-W\left(\hat{g}^{S}\right)\right]\nonumber \\
 & \leq E_{P^n}\left[W\left(g\right)-W_{n}\left(g\right)\right] + E_{P^n}\left[W_{n}\left(\hat{g}^{S}\right)-W\left(\hat{g}^{S}\right)\right]\nonumber \\
 & \leq2E_{P^n}\left[\sup_{g\in{\cal G}}\left|W_{n}\left(g\right)-W\left(g\right)\right|\right]\nonumber \\
  &=2E_{P^n}\left[\sup_{g\in{\cal G}}\left| \sum_{t=1}^{T}\left( W_{nt}( \underline{g}_{t}) - W_{t}( \underline{g}_{t})\right) \right|\right]\nonumber \\
 & \leq2\sum_{t=1}^{T}E_{P^n}\left[\sup_{\text{\ensuremath{\underline{g}}}_{t} \in \MG_{1}\times \cdots \times \MG_{t}} \left|W_{nt}(\text{\ensuremath{\underline{g}}}_{t})-W_{t}(\text{\ensuremath{\underline{g}}}_{t})\right|\right],\label{eq:decomposition}
\end{align}
where the second line follows from the fact that $\hat{g}^{S}$ maximizes $W_{n}\left(\cdot\right)$ over ${\cal G}$, and the fourth line follows
from the definition of $W_{n}\left(\cdot\right)$ and equation
(\ref{eq:original welfare}). 

Applying Lemma \ref{lem:KT}, combined with Lemma \ref{lem:vc_subclass}, to each term in (\ref{eq:decomposition}) leads to the following: for each $t=1,\ldots,T$,
\begin{align*}
E_{P^{n}}\left[\sup_{\text{\ensuremath{\underline{g}}}_{t} \in \MG_{1}\times \cdots \times \MG_{t}}\left|W_{nt}(\text{\ensuremath{\underline{g}}}_{t})-W_{t}(\text{\ensuremath{\underline{g}}}_{t})\right|\right] & \leq C\frac{\gamma_{t}M_{t}/2}{\prod_{s=1}^{t}\kappa_{s}}\sqrt{\frac{\sum_{s=1}^{t}v_{s}}{n}},
\end{align*}
where $C$ is the same universal constant that appears in Lemma \ref{lem:KT}. Combining this result with (\ref{eq:decomposition}),
we obtain
\begin{align*}
E_{P^{n}}\left[W_{{\cal G}}^{\ast}-W\left(\hat{g}^{S}\right)\right] & \leq C\sum_{t=1}^{T}\left\{ \frac{\gamma_{t}M_{t}}{\prod_{s=1}^{t}\kappa_{s}}\sqrt{\frac{\sum_{s=1}^{t}v_{s}}{n}}\right\} .
\end{align*}
Since this upper bound does not depend on $P\in{\cal P}\left(M, \kappa, \MG\right)$,
the upper bound is uniform over ${\cal P}\left(M, \kappa, \MG\right)$. 
\end{proof}
 \medskip
We next present the proof for the backward DEWM method.
\medskip

\begin{proof}[ Proof of Theorem \ref{thm:upper bound} (ii).]
Let $P\in{\cal P}\left(M, \kappa, \MG\right)$ be fixed. Let $g^\ast$ be defined in Section \ref{sec:backward DEWM}. It follows under Assumptions \ref{asm:sequential independence} and \ref{asm:first-best} that 
\begin{align*}
W_{\MG}^{\ast}-W\left(\hat{g}^{B}\right)=  \tilde{Q}_{1}\left(g^\ast\right)-\tilde{Q}_{1}\left(\hat{g}^{B}\right)\leq\Delta\tilde{Q}_{1}.
\end{align*}
Then, from Lemma \ref{lem:backward DEWM} (ii),
\begin{align*}
W_{\MG}^{\ast}-W\left(\hat{g}^{B}\right) & \leq\Delta \tilde{Q}_{1}^{\dagger}+\sum_{s=1}^{T-1}\frac{2^{s-1}}{\prod_{t=1}^{s}\kappa_{t}}\Delta \tilde{Q}_{s+1}^{\dagger}.
\end{align*}
Thus, we have
\begin{align}
    E_{P^{n}}\left[W_{{\cal G}}^{\ast}-W\left(\hat{g}^{B}\right)\right]&\leq
    E_{P^n}\left[\Delta \tilde{Q}_{1}^{\dagger}\right] + \sum_{s=1}^{T-1}\frac{2^{s-1}}{\prod_{t=1}^{s}\kappa_{t}} E_{P^n}\left[\Delta \tilde{Q}_{s+1}^{\dagger}\right].
    \label{eq:BDEWM_proof}
\end{align}
Regarding $\Delta \tilde{Q}_{t}^{\dagger}$ for each $t$, it follows that 
\begin{align}
\Delta \tilde{Q}_{t}^{\dagger}&=  \tilde{Q}_{t}\left(g_{t}^{\ast},\hat{g}_{t+1}^{B},\ldots,\hat{g}_{T}^{B}\right)-\tilde{Q}_{t}\left(\hat{g}_{t}^{B},\ldots,\hat{g}_{T}^{B}\right) \notag \\
&=  \tilde{Q}_{t}\left(g_{t}^{\ast},\hat{g}_{t+1}^{B},\ldots,\hat{g}_{T}^{B}\right)-\tilde{Q}_{nt}\left(g_{t}^{\ast},\hat{g}_{t+1}^{B},\ldots,\hat{g}_{T}^{B}\right) +\tilde{Q}_{nt}\left(g_{t}^{\ast},\hat{g}_{t+1}^{B},\ldots,\hat{g}_{T}^{B}\right)-\tilde{Q}_{t}\left(\hat{g}_{t}^{B},\ldots,\hat{g}_{T}^{B}\right) \notag \\
&\leq  \tilde{Q}_{t}\left(g_{t}^{\ast},\hat{g}_{t+1}^{B},\ldots,\hat{g}_{T}^{B}\right)-\tilde{Q}_{nt}\left(g_{t}^{\ast},\hat{g}_{t+1}^{B},\ldots,\hat{g}_{T}^{B}\right)+\tilde{Q}_{nt}\left(\hat{g}_{t}^{B},\ldots,\hat{g}_{T}^{B}\right)-\tilde{Q}_{t}\left(\hat{g}_{t}^{B},\ldots,\hat{g}_{T}^{B}\right) \notag \\
& \leq  2\sup_{\left(g_{t},\ldots,g_{T}\right)\in{\cal G}_{t}\times\cdots\times{\cal G}_{T}}\left|\tilde{Q}_{nt}\left(g_{t},\ldots,g_{T}\right)-\tilde{Q}_{t}\left(g_{t},\ldots,g_{T}\right)\right|, \label{eq:centered empirical process for BDEWM}
\end{align}
where the first inequality follows from the fact that $\hat{g}_{t}^{B}$
maximizes $\tilde{Q}_{nt}\left(\cdot,\hat{g}_{t+1}^{B},\ldots,\hat{g}_{T}^{B}\right)$ over ${\cal G}_{t}$. Because $\left\Vert \tilde{Q}_{t}\left(g_{t},\ldots,g_{T}\right)\right\Vert _{\infty}\leq  \sum_{s=t}^{T}(\gamma_{s}M_{s}/2)/\left(\prod_{\ell =t}^{s}\kappa_{\ell}\right)$
holds under Assumptions \ref{asm:bounded outcome} and \ref{asm:overlap}, by applying Lemmas \ref{lem:vc_subclass} and \ref{lem:KT} to the following class of
functions: 
\begin{align*}
    \left\{ \sum_{s=t}^{T}\left\{ \frac{\left(\prod_{\ell =t}^{s}1\left\{ D_{\ell} = g_{\ell}\left(H_{\ell}\right)\right\}\right) \gamma_{s}Y_{s}}{\prod_{\ell =t}^{s}e_{\ell}\left(D_{\ell},H_{\ell}\right)}\right\} :\left(g_{t},\ldots,g_{T}\right)\in\MG_{t}\times \cdots \times \MG_{T}\right\}, 
\end{align*}
we have 
\begin{align*}
E_{P^{n}}\left[\sup_{\left(g_{t},\ldots,g_{T}\right)\in{\cal G}_{t}\times\cdots\times{\cal G}_{T}}\left|\tilde{Q}_{nt}\left(g_{t},\ldots,g_{T}\right)-\tilde{Q}_{t}\left(g_{t},\ldots,g_{T}\right)\right|\right] & \leq C\left(\sum_{s=t}^{T}\frac{\gamma_{s}M_{s}/2}{\prod_{\ell =t}^{s}\kappa_{\ell}}\right)\sqrt{\frac{\sum_{s=t}^{T}v_{s}}{n}}.
\end{align*}
Combining this with equations (\ref{eq:BDEWM_proof}) and (\ref{eq:centered empirical process for BDEWM}) leads to 
\begin{align*}
E_{P^{n}}\left[W_{{\cal G}}^{\ast}-W\left(\hat{g}^{B}\right)\right] & \leq  C\sum_{t=1}^{T}\left\{ \frac{\gamma_{t}M_{t}}{\prod_{s=1}^{t}\kappa_{s}}\sqrt{\frac{\sum_{s=1}^{t}v_{s}}{n}}\right\} \\
 & +\sum_{t=2}^{T}\frac{2^{t-2}}{\prod_{s=1}^{t-1}\kappa_{s}}\left(C\sum_{s=t}^{T}\left\{ \frac{\gamma_{s}M_{s}}{\prod_{\ell =t}^{s}\kappa_{\ell}}\sqrt{\frac{\sum_{\ell =t}^{s}v_{\ell}}{n}}\right\} \right),
\end{align*}
where $C$ is the same universal constant that appears in Lemma
\ref{lem:KT}. Since this upper bound does not depend on $P\in{\cal P}\left(M, \kappa, \MG\right)$,
the upper bound is uniform over ${\cal P}\left(M, \kappa, \MG\right)$. 
\end{proof}

\section{Proof of Theorem \ref{thm:budget constraint}
\label{appendix:proofs-2}}

This appendix presents the proof of Theorem \ref{thm:budget constraint}. We first introduce several lemmas that will be used in the proof of Theorem \ref{thm:budget constraint}. 

The following lemma provides a concentration inequality that is frequently used in the literature on statistical learning theory, the proof of which can be found, for example, in \cite{Mohri_et_al_2012}.

\medskip

\begin{lemma}{(McDiarmid's Inequality):}\label{lem:McDiarmid's inequality}
Let $S=(Z_{1},\ldots,Z_{n})\in{\cal Z}^{n}$ be a set of $n$ independent
random variables, and $g$ be a mapping from ${\cal Z}^{n}$ to $\mathbb{R}$
such that there exist $c_{1},\ldots,c_{n}>0$ that satisfy the following
conditions: 
\begin{align*}
\left|g\left(z_{1},\ldots,z_{i},\ldots,z_{n}\right)-g\left(z_{1},\ldots,z_{i}^{\prime},\ldots,z_{n}\right)\right|  <c_{i}
\end{align*}
for any $n+1$ points $z_{1},\ldots,z_{n},z_{i}^{\prime}$ in $\MZ$ and all $i\in\left\{ 1,\ldots,n\right\}$.
Let $g\left(S\right)$ denote $g\left(Z_{1},\ldots,Z_{n}\right)$.
Then the following inequalities hold for all $\epsilon>0$: 
\begin{align*}
\Pr\left(g\left(S\right)-E\left[g\left(S\right)\right]\geq\epsilon\right) & \leq\exp\left(\frac{-2\epsilon^{2}}{\sum_{i=1}^{n}c_{i}^{2}}\right),\\
\Pr\left(g\left(S\right)-E\left[g\left(S\right)\right]\leq-\epsilon\right) & \leq\exp\left(\frac{-2\epsilon^{2}}{\sum_{i=1}^{n}c_{i}^{2}}\right).
\end{align*}
\end{lemma}

\medskip

The following lemma gives a finite-sample upper bound on $\sup_{g \in \MG}\left|W\left(g\right)-W_{n}\left(g\right)\right|$ that holds with a high probability.
\medskip
\begin{lemma}\label{prop:upper bound} Suppose that the underlying distribution $P$ satisfies Assumptions \ref{asm:sequential independence}, \ref{asm:bounded outcome}, and \ref{asm:overlap}
and that ${\cal {G}}$ satisfies Assumption \ref{asm:vc-class}. Then,
for any $\delta\in\left(0,1\right)$, the following holds with probability
at least $1-\delta$: 
\begin{align}
  \sup_{g \in \MG}\left|W\left(g\right)-W_{n}\left(g\right)\right|
\leq  \frac{1}{\sqrt{n}}\sum_{t=1}^{T}\left[\frac{\gamma_{t}M_{t}}{\prod_{s=1}^{t}\kappa_{s}} \cdot \left( C \sqrt{\sum_{s=1}^{t}v_{s}}+\sqrt{\frac{\log\left(1/\delta\right)}{2}}\right) \right]. \label{eq:from prop-1}
\end{align}
\end{lemma}

\begin{proof}

The proof follows a similar argument as that of Corollary 3.4 of \cite{Mohri_et_al_2012}.
We will evaluate $\sup_{g\in{\cal G}}\left|W\left(g\right)-W_{n}\left(g\right)\right|$.
Let $S=\left(Z_{1},\ldots,Z_{n}\right)$ be the sample and define $A\left(S\right) \equiv\sup_{g\in{\cal G}}\left\{ W\left(g\right)-W_{S}\left(g\right)\right\}$, where, for any sample $S$ with size $n$, $W_{S}\left(g\right)$ is defined as $W_{n}\left(g\right)$ that uses the sample $S$. 

Introduce $S^{\prime}=\left(Z_{1},\ldots,Z_{n-1},Z_{n}^{\prime}\right)$,
an i.i.d. sample that is different from $S$ with respect to the final component.
Then, it follows that 
\begin{align*}
A\left(S\right)-A\left(S^{\prime}\right) & =\sup_{g\in{\cal G}}\inf_{g^{\prime}\in{\cal G}}\left\{ W\left(g\right)-W_{S}\left(g\right)-W\left(g^{\prime}\right)+W_{S^{\prime}}\left(g^{\prime}\right)\right\} \\
 & \leq\sup_{g\in{\cal G}}\left\{ W\left(g\right)-W_{S}\left(g\right)-W\left(g\right)+W_{S^{\prime}}\left(g\right)\right\} \\
 & =\frac{1}{n}\sup_{g\in{\cal G}}\left\{ \sum_{t=1}^{T}w_{t}^{S}\left(Z_{n},\underline{g}_{t}\right)-\sum_{t=1}^{T}w_{t}^{S}\left(Z_{n}^{\prime},\underline{g}_{t}\right)\right\} \\
 & \leq\frac{1}{n}\sum_{t=1}^{T}\sup_{g\in \MG}\left\{ w_{t}^{S}\left(Z_{n},\underline{g}_{t}\right)-w_{t}^{S}\left(Z_{n}^{\prime},\underline{g}_{t}\right)\right\} \\
 & \leq\frac{1}{n}\sum_{t=1}^{T}\left(\frac{\gamma_{t}M_{t}}{\prod_{s=1}^{t}\kappa_{s}}\right),
\end{align*}
where the last inequality follows from the fact that under Assumptions \ref{asm:bounded outcome} and \ref{asm:overlap}, $w_{t}^{S}\left(Z_{i},\underline{g}_{t}\right)$ is bounded from above by $\left(\gamma_{t}M_{t}/2\right)/\left(\prod_{s=1}^{t}\kappa_{s}\right)$.

Since we have 
\begin{align*}
    \left| A\left(S\right)-A\left(S^{\prime}\right)\right|\leq \frac{1}{n}\sum_{t=1}^{T}\frac{\gamma_{t}M_{t}}{\prod_{s=1}^{t}\kappa_{s}},
\end{align*}
applying Lemma \ref{lem:McDiarmid's inequality} leads to
\begin{align*}
P\left( \left|A\left(S\right)-E_{P^{n}}\left[A\left(S\right)\right]\right|\geq\epsilon\right)   \leq\exp\left(\frac{-2n\epsilon^{2}}{\left( \sum_{t=1}^{T}\frac{\gamma_{t}M_{t}}{\prod_{s=1}^{t}\kappa_{s}}\right) ^{2}}\right)
\end{align*}
for any $\epsilon >0$. This is equivalent to the following inequality: for any $\delta\in\left(0,1\right)$,
\begin{align}
P\left( \left|A\left(S\right)-E_{P^{n}}\left[A\left(S\right)\right]\right|\leq\left(\sum_{t=1}^{T}\frac{\gamma_{t}M_{t}}{\prod_{s=1}^{t}\kappa_{s}}\right)\sqrt{\frac{\log\left(1/\delta\right)}{2n}}\right)   \geq1-\delta.\label{eq:proof-prop-2}
\end{align}

Subsequently, we will evaluate $E_{P^{n}}\left[A\right(S\left)\right]$.
Since 
\begin{align*}
E_{P^{n}}[A(S)] & = E_{P^n}\left[\sup_{g\in{\cal G}}\left| \sum_{t=1}^{T}\left( W_{nt}( \underline{g}_{t}) - W_{t}( \underline{g}_{t})\right) \right|\right]\nonumber \\
 & \leq \sum_{t=1}^{T}E_{P^n}\left[\sup_{\text{\ensuremath{\underline{g}}}_{t} \in \MG_{1}\times \cdots \times \MG_{t}} \left|W_{nt}(\text{\ensuremath{\underline{g}}}_{t})-W_{t}(\text{\ensuremath{\underline{g}}}_{t})\right|\right],
\end{align*}
applying Lemma \ref{lem:KT} combined with Lemma \ref{lem:vc_subclass} leads to
\begin{align}
E_{P^{n}}[A(S)] \leq C \sum_{t=1}^{T}\left\{ \frac{\gamma_{t}M_{t}}{\prod_{s=1}^{t}\kappa_{s}}\sqrt{\frac{\sum_{s=1}^{t}v_{s}}{n}}\right\}, \label{eq:proof-prop-3}
\end{align}
where $C$ is the same constant that appears in Lemma \ref{lem:KT}.

Consequently, combining (\ref{eq:proof-prop-2})
and (\ref{eq:proof-prop-3}), for any $\delta\in\left(0,1\right)$, it follows 
with probability at least $1-\delta$ that 
\begin{align*}
\sup_{g \in \MG}\left|W\left(g\right)-W_{n}\left(g\right)\right|
&\leq  C\sum_{t=1}^{T}\left[\frac{\gamma_{t}M_{t}}{\prod_{s=1}^{t}\kappa_{s}}\sqrt{\frac{\sum_{s=1}^{t}v_{s}}{n}}\right]+\left(\sum_{t=1}^{T}\frac{\gamma_{t}M_{t}}{\prod_{s=1}^{t}\kappa_{s}}\right)\sqrt{\frac{\log\left(1/\delta\right)}{2n}}\\
&=  \frac{1}{\sqrt{n}}\sum_{t=1}^{T}\left[\frac{\gamma_{t}M_{t}}{\prod_{s=1}^{t}\kappa_{s}} \cdot \left( C \sqrt{\sum_{s=1}^{t}v_{s}}+\sqrt{\frac{\log\left(1/\delta\right)}{2}}\right) \right].
\end{align*}
\end{proof}
 \medskip

The following lemma shows that a class of feasible DTRs that satisfy the empirical budget/capacity constraints (\ref{eq:empirical budget constraint}) contains the optimal DTR with high probability. 

\medskip

\begin{lemma} \label{lem:budget constraint} Suppose that the underlying distribution $P$ satisfies Assumption \ref{asm:sequential independence} and that
$\sum_{t=1}^{T}K_{tb}=1$ holds for all $b=1,\ldots,B$. For $k>0$, let $\tilde{g}^{\ast} = (\tilde{g}_{1}^{\ast},\ldots,\tilde{g}_{T}^{\ast})$ be a solution of the constrained maximization problem (\ref{eq:budget constrained welfare}) with $C_b$ replaced by $C_{b} - k + \alpha_{n}$, where we suppose that such a solution exists. Define 
\begin{align*}
\MG_{\alpha_{n}}^{S}\equiv &\left\{  g\in\MG:  \sum_{t=1}^{T}K_{tb}\hat{E}\left[g_{t}\left(\widetilde{H}_{t}\left(\underline{g}_{t-1}\right)\right)\right]\leq C_{b}+\alpha_{n}  \mbox{\mbox{ for }}b=1,\ldots,B \right\} ,
\end{align*}
which is a subset of DTRs that satisfy the sample budget constraints
(\ref{eq:empirical budget constraint}). Then, for any $\delta\in\left(0,1\right)$,
$   P\left(g^{\ast} \in {\cal G}_{\alpha_{n}}^{S}\right) \geq 1- B\cdot \exp\left(-2nk^{2}\right)$ holds. 
\end{lemma}

\begin{proof}
It follows that 
\begin{align*}
 P\left(\tilde{g}^{\ast}\notin{\cal G}_{\alpha_{n}}^{S}\right)  &= P\left(\max_{b=1,\ldots,B}\left\{\sum_{t=1}^{T}K_{tb}\hat{E}\left[\tilde{g}_{t}^{\ast}\left(\widetilde{H}_{t}\left(\underline{\tilde{g}}_{t-1}^{\ast}\right)\right)\right]-C_{b}\right\}>\alpha_{n} \right)\\
  &\leq \sum_{b=1}^{B}P\left(\sum_{t=1}^{T}K_{tb}\hat{E}\left[\tilde{g}_{t}^{\ast}\left(\widetilde{H}_{t}\left(\underline{\tilde{g}}_{t-1}^{\ast}\right)\right)\right]-C_{b}>\alpha_{n}\right)\\
  &\leq  \sum_{b=1}^{B}P\left(\sum_{t=1}^{T}K_{tb}\hat{E}\left[\tilde{g}_{t}^{\ast}\left(\widetilde{H}_{t}\left(\underline{\tilde{g}}_{t-1}^{\ast}\right)\right)\right] -\sum_{t=1}^{T}K_{tb}E_{P}\left[\tilde{g}_{t}^{\ast}\left(\widetilde{H}_{t}\left(\underline{\tilde{g}}_{t-1}^{\ast}\right)\right)\right]>k\right),
\end{align*}
where the second inequality follows from the fact that $\tilde{g}^{\ast}$ satisfies
the population budget/capacity constraints (\ref{eq:budget constraint}) with $C_b$ replaced by $C_{b} - k + \alpha_{n}$.

By Hoeffding's inequality, it follows for each $b=1,\ldots,B$ that 
\begin{align*}
&P\left(\sum_{t=1}^{T}K_{tb}\hat{E}\left[\tilde{g}_{t}^{\ast}\left(\widetilde{H}_{t}\left(\underline{\tilde{g}}_{t-1}^{\ast}\right)\right)\right] -\sum_{t=1}^{T}K_{tb}E_{P}\left[g_{t}^{\ast}\left(\widetilde{H}_{t}\left(\underline{g}_{t-1}^{\ast}\right)\right)\right]>k \right) \\
 \leq & \exp\left\{ -\frac{2n k^{2}}{\left(\sum_{t=1}^{T}K_{tb}\right)^{2}}\right\} 
  = \exp\left(-2n\alpha_{n}^{2}\right),
\end{align*}
where the equality follows from the scale
normalization $\sum_{t=1}^{T}K_{tb}=1$. Thus, we have $P\left(g^{\ast} \notin {\cal G}_{\alpha_{n}}^{S}\right) \leq  B\cdot \exp\left(-2nk^{2}\right)$.
Therefore, $ P\left(g^{\ast} \in {\cal G}_{\alpha_{n}}^{S}\right) = 1 - P\left(g^{\ast} \notin {\cal G}_{\alpha_{n}}^{S}\right) \geq 1- B\cdot \exp\left(-2nk^{2}\right)$.
\end{proof}
 \medskip

\begin{proof}[ Proof of Theorem \ref{thm:budget constraint} (i).]

We use the notation $A\leq_{\delta}B$ to denote that $A\leq B$ holds with probability at least $1-\delta$. 
Let $g_{\alpha_n}^{\ast}$ be a solution of the constrained maximization problem (\ref{eq:budget constrained welfare}) with $C_b$ replaced by $C_{b} - k_{(B,n,\delta)} + \alpha_{n}$. 

By Lemma \ref{lem:budget constraint}, 
we have $P\left(g_{\alpha_n}^{\ast}\in{\cal G}_{\alpha_{n}}^{S}\right) \geq 1 - \delta/6$. Thus,  $W_{n}\left(g_{\alpha_n}^{\ast}\right)\leq_{\delta/6}W_{n}\left(\hat{g}^{bdgt}\right)$ holds
because $\hat{g}^{bdgt}$ maximizes $W_{n}\left(\cdot\right)$ over ${\cal G}_{\alpha_{n}}^{S}$. Note that $W\left(g_{\alpha_n}^{\ast}\right) = W_{{\cal G}}^{\ast,bdgt}$.
By combining  the fact that $W_{n}\left(g_{\alpha_n}^{\ast}\right)\leq_{\delta/6}W_{n}\left(\hat{g}^{bdgt}\right)$ with (\ref{eq:from prop-1}), it follows that
\begin{align*}
  W_{{\cal G}}^{\ast,bdgt} &= W\left(g_{\alpha_n}^{\ast}\right) \\
& \leq_{\delta/6}  W_{n}\left(g_{\alpha_n}^{\ast}\right)+\frac{1}{\sqrt{n}}\sum_{t=1}^{T}\left[\frac{\gamma_{t}M_{t}}{\prod_{s=1}^{t}\kappa_{s}}\cdot \left( C \sqrt{\sum_{s=1}^{t}v_{s}}+\sqrt{\frac{\log\left(6/\delta\right)}{2}}\right) \right]\\
&\leq_{\delta/6}  W_{n}\left(\hat{g}^{bdgt}\right)+\frac{1}{\sqrt{n}}\sum_{t=1}^{T}\left[\frac{\gamma_{t}M_{t}}{\prod_{s=1}^{t}\kappa_{s}}\cdot \left( C \sqrt{\sum_{s=1}^{t}v_{s}}+\sqrt{\frac{\log\left(6/\delta\right)}{2}}\right) \right]\\
&\leq_{\delta/6}  W\left(\hat{g}^{bdgt}\right)+\frac{2}{\sqrt{n}}\sum_{t=1}^{T}\left[\frac{\gamma_{t}M_{t}}{\prod_{s=1}^{t}\kappa_{s}}\cdot \left( C \sqrt{\sum_{s=1}^{t}v_{s}}+\sqrt{\frac{\log\left(6/\delta\right)}{2}}\right) \right].
\end{align*}
The first inequality follows from the inequality in (\ref{eq:from prop-1});
the second inequality follows from the fact that $\hat{g}^{bdgt}$ maximizes $W_{n}\left(\cdot\right)$ over ${\cal G}_{\alpha_{n}}^{S}$ and $g_{\alpha_n}^{\ast}\in{\cal G}_{\alpha_{n}}^{S}$ holds with probability at least $1-\delta/6$; the third inequality follows from the inequality in (\ref{eq:from prop-1}).  Overall, we have
\begin{align}
 W_{{\cal G}}^{\ast,bdgt} 
\leq_{\delta/2}  W\left(\hat{g}^{bdgt}\right) 
+\frac{1}{\sqrt{n}}\sum_{t=1}^{T}\left[\frac{\gamma_{t}M_{t}}{\prod_{s=1}^{t}\kappa_{s}}\cdot \left( 2C \sqrt{\sum_{s=1}^{t}v_{s}}+\sqrt{2\log\left(6/\delta\right)}\right) \right]. \label{eq:probability bound_welfare deviation}
\end{align}

Applying the same argument as in the proof of Lemma \ref{prop:upper bound}, it follows for each $b=1,\ldots,B$ that 
\begin{align}
 & \left|E_{n}\left[\sum_{t=1}^{T}K_{tb}\hat{g}_{t}^{bdgt}\left(\widetilde{H}_{t}\left(\underline{\hat{g}}_{t-1}^{bgdt}\right)\right)\right]-E_{P}\left[\sum_{t=1}^{T}K_{tb}\hat{g}_{t}^{bdgt}\left(\widetilde{H}_{t}\left(\underline{\hat{g}}_{t-1}^{bgdt}\right)\right)\right]\right|\nonumber \\
\leq & \sum_{t=1}^{T} \sup_{\underline{g}_{t} \in \MG_{1} \times \cdots \times \MG_{t}}\left|E_{n}\left[K_{tb}g_{t}\left(\widetilde{H}_{t}\left(\underline{g}_{t-1}\right)\right)\right]-E_{P}\left[K_{tb}g_{t}\left(\widetilde{H}_{t}\left(\underline{g}_{t-1}\right)\right)\right]\right|\nonumber \\
\leq_{\delta} & \frac{1}{\sqrt{n}}\sum_{t=1}^{T}\left[K_{tb}\cdot \left( C \sqrt{\sum_{s=1}^{t}v_{s}}+\sqrt{\frac{\log\left(1/\delta\right)}{2}}\right) \right].\label{eq:from prop-2}
\end{align}
Furthermore, for each $b=1,\ldots,B$, 
\begin{align*}
E_{P}\left[\sum_{t=1}^{T}K_{tb}\hat{g}_{t}^{bdgt}\left(\widetilde{H}_{t}\left(\underline{\hat{g}}_{t-1}^{bgdt}\right)\right)\right]
 & \leq_{\delta/\left(2B\right)}E_{n}\left[\sum_{t=1}^{T}K_{tb}\hat{g}_{t}^{bdgt}\left(\widetilde{H}_{t}\left(\underline{\hat{g}}_{t-1}^{bgdt}\right)\right)\right] \\ &+\frac{1}{\sqrt{n}}\sum_{t=1}^{T}\left[K_{tb}\cdot \left( C \sqrt{\sum_{s=1}^{t}v_{s}}+\sqrt{\frac{\log\left(2B/\delta\right)}{2}}\right) \right]\\
 & \leq C_{b}+\alpha_{n}+\frac{1}{\sqrt{n}}\sum_{t=1}^{T}\left[K_{tb}\cdot \left( C \sqrt{\sum_{s=1}^{t}v_{s}}+\sqrt{\frac{\log\left(2B/\delta\right)}{2}}\right) \right],
\end{align*}
where the first inequality follows from the inequality in (\ref{eq:from prop-2}), and the second inequality follows from the fact that $\hat{g}^{bdgt}\in{\cal G}_{\alpha_{n}}^{S}$. Thus, the following holds with probability at least $1-\delta$: for any $b \in \{1,\ldots,B\}$,
\begin{align}
 & E_{P}\left[\sum_{t=1}^{T}K_{tb}\hat{g}_{t}^{bdgt}\left(\widetilde{H}_{t}\left(\underline{\hat{g}}_{t-1}^{bgdt}\right)\right) - C_{b}\right] \notag \\
 & \leq_{\delta/(2B)} \alpha_{n}+\frac{1}{\sqrt{n}}\sum_{t=1}^{T}\left[K_{tb}\cdot \left( C\sqrt{\sum_{s=1}^{t}v_{s}}+\sqrt{\frac{\log\left(2B/\delta\right)}{2}}\right) \right].
 \label{eq:probabiloty bound_budget deviation}
\end{align}

The result follows from combining the probability inequalities (\ref{eq:probability bound_welfare deviation}) and (\ref{eq:probabiloty bound_budget deviation}) for all $b=1,\ldots,B$.
\end{proof}

\section{Non-additive Welfare Function} \label{appendix:non-additive_welfare_function}

In this appendix, we consider a non-additive \textit{social welfare function} (SWF) and provide a simultaneous dynamic EWM approach to estimate the optimal DTR. We consider the 
equality-minded rank-dependent SWFs introduced by \cite{Meyer_1995} and \cite{Weymark_1981} and studied by \cite{Kitagawa_Tetenov_2021}: 
\begin{align}
    W_{\Lambda}(F) \equiv \int_{0}^{\infty}\Lambda(F(y))dy, \label{eq:rank-dependent_welfare_function}
\end{align}
where $F(y)$ is the distribution of an outcome and $\Lambda(\cdot):[0,1] \rightarrow [0,1]$ is a non-increasing, non-negative functions with $\Lambda(0)=1$ and $\Lambda(1)=0$. 

An important family of SWFs represented by (\ref{eq:rank-dependent_welfare_function}) is the extended Gini family \citep{Donaldson_Weymark_1980,Donaldson_Weymark_1983,Aaberge_et_al_2013}:
\begin{align*}
    W_k(F) &\equiv \int_{0}^{\infty}(1 - F(y))^{k-1}dy = \int_{0}^{\infty}\Lambda_{k}\left(F(y)\right)dy \\
    & = \int_{0}^{1}F^{-1}(\tau)\omega_{k}(\tau)d\tau,
\end{align*}
where $\Lambda_{k}(\tau) \equiv (1 - \tau)^{k-1}$ and $\omega_{k}(\tau)\equiv (k-1)(1-\tau)^{k-2}$. The standard Gini social welfare function \citep{Blackorby_Donaldson_1978,Weymark_1981} corresponds to the extended Gini social welfare function when $k=3$, which can also be written as
\begin{align*}
    W_{Gini}(F) = E(Y)(1 - I_{Gini}(F)),
\end{align*}
where $I_{Gini}(F) = 1 - \frac{\int_{0}^{1}F^{-1}(\tau)\cdot 2(1-\tau)d\tau}{E(Y)}$ is the widely used Gini inequality index.

Without loss of generality, we suppose that the target outcome is $\sum_{t=1}^{T}Y_{t}$. For any DTR  $g=(g_1, \ldots, g_T)$, let $F_{g}(\cdot)$ denote the distribution of $\sum_{t=1}^{T}\widetilde{Y}_{t}(\underline{g}_{t})$.
We define the rank-dependent SWF of $g$ by
\begin{align}
    W_{\Lambda}(g) &\equiv W_{\Lambda} (F_g). \label{eq:rank-dependent_SWF}
\end{align}
Our goal is to estimate the optimal DTR that maximizes $W_{\Lambda}(g)$ over the pre-specified class of DTRs $\MG$. 

We estimate the optimal DTR by simultaneously maximizing the sample analogue of the population welfare function $W_{\Lambda}(g)$ over $g \in \MG$. Let 
\begin{align*}
    \widehat{F}_{g}(y) \equiv 1 - \frac{1}{n}\sum_{i=1}^{n} \left(\frac{\prod_{t=1}^{T}1\{D_{it} = g_{t}(H_{it})\}}{\prod_{t=1}^{T}e_{t}(D_{it},H_{it})}\cdot 1\left\{\sum_{t=1}^{T}Y_{it}>y\right\}\right),
\end{align*}
which is the inverse probability weighting estimator of the distribution of $\sum_{t=1}^{T}\widetilde{Y}_{t}(\underline{g}_{t})$. The sample analogue of the population welfare $W_{\Lambda}(g)$ is 
\begin{align*}
    \widehat{W}_{\Lambda}(g) \equiv \int_{0}^{\infty} \Lambda\left(\widehat{F}_{g}(y) \vee 0\right) dy, 
\end{align*}
where the maximum ($\vee$) of $\widehat{F}_{g}(y)$ and $0$ is taken because $\widehat{F}_{g}(y)$ may take values smaller than $0$, for which $\Lambda(\cdot)$ is not defined. The simultaneous DEWM approach estimates the optimal DTR by solving
\begin{align*}
    \hat{g}^{S} \in \argmax_{g\in \MG} \widehat{W}_{\Lambda}(g).
\end{align*}

Let $\MP$ be a class of distributions of $\left(\underline{A}_T,\{\underline{X}_{T}(\underline{d}_{T-1})\}_{\underline{d}_{T-1} \in \{0,1\}^{T-1}},\{\underline{Y}_{T}(\underline{d}_{T})\}_{\underline{d}_{T} \in \{0,1\}^{T}}\right)$. The following theorem derives a uniform upper bound of the average welfare loss of $\hat{g}^{S}$.

\bigskip

\begin{theorem}\label{thm:equality_mindied_welfare}
Suppose that Assumptions \ref{asm:sequential independence} and \ref{asm:overlap} hold for any distribution $P \in \MP$ and Assumption \ref{asm:vc-class} holds for $\MG$. Furthermore, suppose that the following hold:
\begin{itemize}
    \item $\Lambda(\cdot):[0,1] \rightarrow [0,1]$ is a non-increasing, convex function with $\Lambda(0)=1$, $\Lambda(1)=0$, and its right derivative at $0$ is finite;
    \item there exists $\Upsilon < \infty$ such that for all $P \in \MP$ and any $\underline{d}_{T} \in \{0,1\}^T$,
    \begin{align}
        \bigintsss_{0}^{\infty}\sqrt{P\left(\sum_{t=1}^{T}Y_{t}(\underline{d}_{t}) > y\right)} dy \leq \Upsilon. \label{eq:tail_bound}
    \end{align}
\end{itemize}
Then the average welfare loss of $\hat{g}^{S}$ satisfies
\begin{align}
    \sup_{P \in \MP} E_{P^n}\left[\sup_{g \in \MG} W_{\Lambda}(g) - W(\hat{g}^{S})\right] \leq 2C \mid \Lambda^{\prime}(0)\mid \frac{\Upsilon}{\prod_{t=1}^{T}\kappa_{t}}\sqrt{\frac{\sum_{t=1}^{T}v_{t}}{n}}
\end{align}
for all $n>1$, where $C$ is a universal constant.
\end{theorem}

\begin{proof}
See Appendix \ref{sec:proof_equality_minded_welfare}.
\end{proof}

\medskip

This theorem shows that for a large class of data-generating processes, the rank-dependent SWF of the simultaneous DEWM converges to the optimal welfare no slower than $n^{-1/2}$ rate. This uniform convergence rate of $n^{-1/2}$ coincides with that of the DEWM methods for the linear SWF shown in Theorem \ref{thm:upper bound}. The convergence rate of $n^{-1/2}$ also coincides with the minimax optimal convergence rate for the rank-dependent SWF in the static treatment case (Theorems 3.1 and 3.2 in \cite{Kitagawa_Tetenov_2021}).


\section{Multiple Treatment} \label{appendix:multiple}

In the main text, we consider the setting of binary treatment assignment for each stage. However, more than a few examples of DTRs involve multiple treatments in practice. In this section, we extend the DEWM to the case of multiple treatment.

Suppose that there are $K$ treatments in each stage. Let $\MD \equiv \{1,2,\ldots,K\}$ denote the treatment space in each stage, and $D_t \in \MD$ denote the observed treatment in stage $t$. Using the same notations as in Section \ref{sec:setup}, we define the potential outcomes and the potential covariates as $Y_t(\underline{d}_t)$ and $X_{t}(\underline{d}_{t-1})$, respectively. The observed outcomes and observed covariates are denoted as $Y_t \equiv Y_t(\underline{D}_t)$ and $X_t \equiv X_t(\underline{D}_{t-1})$, respectively. We define the i.i.d. sample $\{Z_i \equiv (D_{it},X_{it},Y_{it})_{t=1}^{T}:i=1,\ldots,n\}$.
Additionally, we define the history $H_t$ and $H_{it}$ in the same manner as described in the main text.

We suppose that the sequential independence assumption holds for multiple treatment.

\begin{assumption}(Sequential Independence Assumption) For any $t=1,\ldots,T$ and $\underline{d}_t \in \MD^t$,
\begin{align*}
    \left(Y_t(\underline{d}_t),\ldots,Y_T(\underline{d}_T),X_{t+1}(\underline{d}_t),\ldots,X_T(\underline{d}_{T-1})\right) \indep D_{t} \mid H_{t} \mbox{\ a.s.}
\end{align*}
\end{assumption}

In the multiple treatment setting, the treatment rule $g_t$ in each stage is a map from $\MH_t$ to $\MD$. The DTR denoted by $g$ is the sequence $g = (g_1,\ldots,g_T)$. We denote the class of feasible DTRs by $\MG \equiv \MG_1 \times \cdots \times \MG_T$, where $\MG_t$ is a class of feasible treatment rules at stage $t$.


In the following subsections, we describe the backward DEWM and simultaneous DEWM for multiple treatment in the experimental data setting.

\subsection{Backward DEWM for Multiple Treatments}

To guarantee the consistency of the backward estimation procedure, we suppose that the $\MG_t$ for $t \geq 2$ contains the first-best rule. 
As the same way as in the main text, for any $s<t$, we define 
\begin{align*}
\widetilde{Y}_{t}\left(\underline{d}_{s},\underline{g}_{(s+1):t}\right) \equiv \sum_{\underline{d}_{(s+1):t} \in \MD^{t-s}} Y_t(\underline{d}_{s},\underline{d}_{(s+1):t})\cdot \prod_{\ell=s+1}^{t}1\left\{g_{\ell}\left(H_{\ell}\left(\underline{d}_{\ell-1}\right)\right) = d_\ell\right\},
\end{align*}
which is the outcome in stage $t$ that is realized when the treatment assignments from stage $1$ to stage $s-1$ are fixed to $\underline{d}_{s}$ and the subsequent sequential treatment assignment follows $\underline{g}_{(s+1):t}$. We denote $\widetilde{Y}_{t}\left(\underline{d}_{t},\underline{g}_{(t+1):t}\right) = Y_{t}\left(\underline{d}_{t}\right)$ when $s=t$.

We suppose that the first-best treatment rule is available for a multiple treatment case in the following sense.

\smallskip

\begin{assumption}[First-Best Treatment Rule] 
\label{asm:first-best_multiple} 
For any $t=2,\ldots,T$, there exists $g_{t,FB}^{\ast}\in{\cal G}_{t}$ such that the following holds:
\begin{align*}
&E_{P}\left[\sum_{s=t}^{T}\gamma_{s}\widetilde{Y}_{s}\left(\text{\underline{D}}_{t-1},\underline{g}_{t:s,FB}^{\ast}\right) \middle| H_{t}\right] \geq \max_{d_t \in \MD^{t}}E_{P}\left[\sum_{s=t}^{T}\gamma_{s}\widetilde{Y}_{s}\left(\text{\underline{D}}_{t-1},d_t,\underline{g}_{(t+1):T,FB}^{\ast}\right) \middle| H_{t} \right] \mbox{\ a.s.}
\end{align*}
\end{assumption}

\smallskip

Then, we can consistently estimate the optimal DTR through the backward induction approach in the same manner as in Section \ref{sec:backward DEWM}. Given the propensity scores $\left\{ e_{t}\left(D_{t},H_{t}\right)\right\} _{t=1}^{T}$, let
\begin{align}
q_{t}\left(Z,g_{t};g_{t+1},\ldots,g_{T}\right)\equiv  \sum_{s=t}^{T}\left\{ \frac{\left(\prod_{\ell =t}^{s}1\left\{ D_{\ell} = g_{\ell}\left(H_{\ell}\right)\right\}\right) \gamma_{s}Y_{s}}{\prod_{\ell =t}^{s}e_{\ell}\left(D_{\ell},H_{\ell}\right)}\right\} .\nonumber
\end{align}
With the backward DEWM, the optimal DTR for multiple treatment is sequentially estimated as follows. In the first step, for the last stage $T$, the optimal treatment rule in the last stage is estimated as 
\begin{align*}
\hat{g}_{T}^{B} & \in\argmax_{g_{T}\in{\cal G}_{T}}\frac{1}{n}\sum_{i=1}^{n}q_{T}\left(Z_{i},g_{T}\right).
\end{align*}
Then, recursively, from $t=T-1$ to $1$, the method estimates $g_{t}^{\ast}$ by 
\begin{align*}
\hat{g}_{t}^{B} & \in\argmax_{g_{t}\in{\cal G}_{t}}\frac{1}{n}\sum_{i=1}^{n}q_{t}\left(Z_{i},g_{t};\hat{g}_{t+1}^{B},\ldots,\hat{g}_{T}^{B}\right). 
\end{align*}
Throughout this procedure, we obtain the DTR $\hat{g}^{B}\equiv \left(\hat{g}_{1}^{B},\ldots,\hat{g}_{T}^{B}\right)$.

\subsection{Simultaneous DEWM for Multiple Treatments}

The simultaneous DEWM can also be constructed in the same way as described in Section \ref{sec:simultaneous DEWM}. This approach estimates the optimal DTR through the following maximization problem: 
\begin{align*}
\left(\hat{g}_{1}^{S},\dots,\hat{g}_{T}^{S}\right) & \in\argmax_{g\in{\cal {\cal G}}}\sum_{t=1}^{T}\left[\frac{1}{n}\sum_{i=1}^{n}w_{t}^{S}(Z_{i},\text{\ensuremath{\underline{g}}}_{t})\right],
\end{align*}
with
\begin{align*}
w_{t}^{S}(Z_{i},\text{\ensuremath{\underline{g}}}_{t}) & \equiv \frac{\left(\prod_{s=1}^{t}1\left\{ g_{s}\left(H_{is}\right)=D_{is}\right\}\right) \gamma_{t}Y_{it}}{\prod_{s=1}^{t}e_{s}\left(D_{is},H_{is}\right)}.
\end{align*}


\section{Doubly Robust Estimation Using Observational Data}\label{appendix:DR_estimation}


We consider doubly robust estimation of the optimal DTR. In Section \ref{sec:DR_approach_1}, we discuss extension of the simultaneous maximization approach to doubly robust approach.\footnote{Doubly robust estimators for the optimal DTRs are also proposed by \cite{zhang_et_al_2013}, \cite{Wallace_et_al_2015}, and \cite{Ertefaie_et_al_2021}. \cite{Wallace_et_al_2015} and \cite{Ertefaie_et_al_2021} propose methods with backward induction, 
which requires the correct specification of the model for conditional treatment effects at each stage, referred to by \cite{Wallace_et_al_2015} as the blip function. In contrast, the doubly robust approach with simultaneous maximization proposed in my paper does not require the correct specification of the blip function or optimal treatment rules. 
\cite{zhang_et_al_2013} also consider a doubly robust estimation with simultaneous optimization but do not study the theoretical properties of their proposed method.}
However, as discussed in Remark \ref{rem:doubly_robust} in the main text, doubly robust approach with the simultaneous maximization estimation is computationally challenging unless the class of DTR $\MG$ is not small. In Section \ref{appendix:DR_approach_2}, we construct another doubly-robust simultaneous-maximization approach with computational feasibility under the setting that treatment choice at each stage depends only on the exogenous variables and past treatments.

\subsection{Simultaneous Maximization Method}\label{sec:DR_approach_1}


We here consider extending the simultaneous maximization approaches to doubly robust policy learning.
Following the doubly robust policy learning of \citeApp{Athey_Wager_2020} and \citeApp{Zhou_et_al_2023}, we employ cross-fitting to make the estimation of the welfare function and estimation of the optimal DTRs independent; whereby, to reduce the over-fitting. We randomly divide the data set $\{Z_i:i=1,\ldots,n\}$ into $K$ evenly-sized folds (e.g., $K=5$). Let $I_k$ be a set of indices of the data in the $k$-th fold and $I_{-k}$ be a set of indices of the data excluded from the $k$-th fold. Hereafter, for any statistics $\hat{f}$, we denote by $\hat{f}^{-k}$ the corresponding statistics calculated using data excluded from the $k$-th fold. We denote by $k(i)$ the number of the fold that containts the $i$-th observartion.

In the general dynamic setting, as proposed by \cite{Jiang_Li_2016} and \cite{Thomas_Brunskill_2016}, we can construct an AIPW estimator for the welfare function $W(g)$ of a fixed DTR $g$ as follows:\footnote{Note that \cite{Jiang_Li_2016} and \cite{Thomas_Brunskill_2016} consider the estimation of the value of a fixed DTR $g$, but not the estimation of the optimal DTR.}
    \begin{align}
    \widehat{W}^{AIPW}(g)=\frac{1}{n}\sum_{i=1}^{n}&\left(\sum_{t=1}^{T}\hat{\psi}_{it}^{-k(i)}\left(\underline{g}_{t}\right)\gamma_{t}Y_{it} \right. \nonumber\\
    &\left.-\sum_{t=1}^{T}\left(\hat{\psi}_{it}^{-k(i)}\left(\underline{g}_{t}\right)- \hat{\psi}_{i,t-1}^{-k(i)}\left(\underline{g}_{t-1}\right)\right)\cdot \widehat{Q}_{t}^{\underline{g}_{(t+1):T},-k(i)}\left(H_{it},D_{it}\right) \right), \label{eq:AIPW_simultaneous}
\end{align}
where $\hat{\psi}_{it}^{-k(i)}(\underline{g}_{t}) \equiv \left(\prod_{s=1}^{t}1\left\{ D_{is}=g_{s}(H_{is})\right\}\right)/\left(\prod_{s=1}^{t}\hat{e}_{t}^{-k(i)}\left(H_{is},g_{s}\right)\right)$ is an estimator of the sequential propensity weights, and 
$\widehat{Q}_{t}^{\underline{g}_{(t+1):T},-k(i)}(h_{t},d_{t})$ is an estimator of the action-value function (Q-function) for $\underline{g}_{(t+1):T}$:
$$Q_{t}^{\underline{g}_{(t+1):T}}(h_{t},d_{t}) \equiv E_{P}\left[\gamma_{t}Y_{t} + \sum_{s=t+1}^{T}\gamma_{s}\widetilde{Y}_s(\underline{D}_{t},\underline{g}_{(t+1):s})\middle|H_t = h_t,A_t = d_t\right].$$ 
We denote $\hat{\psi}_{i,0}^{-k(i)}(\underline{g}_0)=1$ when $t=1$ and $\widehat{Q}_{T}^{\underline{g}_{(T+1):T},-k(i)}(\cdot,\cdot)=\widehat{Q}_{T}^{-k(i)}(\cdot,\cdot)$ when $t=T$. The Q-functions $\left\{Q_{t}^{\underline{g}_{(t+1):T}}(h_{t},d_{t})\right\}_{t=1,\ldots,T}$ can be estimated by a sequential step-wise algorithm such as the fitted Q-evaluation \citep{Munos_2008,Le_et_al_2019} from the reinforcement learning literature.
The estimator (\ref{eq:AIPW_simultaneous}) generalizes the AIPW of \cite{Robins_et_al_1994} beyond the static case, and it is a consistent estimator of the population welfare $W(g)$ if either the the propensity weights $\{\psi_{t}(\cdot)\}_{t=1}^{T}$ or the Q-functions $\{Q_{t}^{\underline{g}_{t:T}}(\cdot)\}_{t=1}^{T}$ are consistently estimated.

Using the AIPW estimator (\ref{eq:AIPW_simultaneous}), we can estimate the optimal DTR as a solution of the estimated welfare maximization: $ \hat{g}^{AIPW} \in \argmax_{g \in \MG} \widehat{W}^{AIPW}(g)$. Remark \ref{rem:doubly_robust} in the main text describes the computational challenge of this method unless the class of DTR $\MG$ is not small.

In what follows, we show the statistical property of $\hat{g}^{AIPW}$. Specifically, we will show the convergence rate of the welfare regret $W_{\MG}^{\ast} - W(\hat{g}^{AIPW})$.  Without loss of generality, we suppose that $\gamma_1=\cdots=\gamma_T =1$. 

Let $\widehat{Q}_{t}^{\underline{g}_{(t+1):T},(n)}(\cdot,\cdot)$ and $\hat{e}_{t}^{(n)}(\cdot,\cdot)$, respectively, denote the estimators of the Q-function $Q_{t}^{\underline{g}_{(t+1):T}}(\cdot,\cdot)$ for $\underline{g}_{(t+1):T}$ and the propensity score $e_t(\cdot,\cdot)$ using size $n$ sample randomly drawn from the population $P$. We denote $\widehat{Q}_{T}^{\underline{g}_{(T+1):T},(n)}(\cdot,\cdot) = \widehat{Q}_{T}^{(n)}(\cdot,\cdot)$ when $t=T$.
We suppose that $\{\widehat{Q}_{t}^{\underline{g}_{(t+1):T},(n)}(\cdot,\cdot)\}_{t=1,\ldots,T}$ and $\{\hat{e}_{t}^{(n)}(\cdot,\cdot)\}_{t=1,\ldots,T}$ satisfy the following assumption.

\medskip

\begin{assumption}\label{asm:rate_of_convergence_AIPW}
(i) There exists $\tau >0$ such that the following holds: For all $t=1,\ldots,T$, $s=1,\ldots,t$, and $m \in \{0,1\}$,
\begin{align*}
    \sup_{\underline{d}_{s:t} \in \{0,1\}^{t-s+1}} & E\left[ \sup_{\underline{g}_{(t+1):T} \in \MG_{(t+1):T}}\left(\widehat{Q}_{t}^{\underline{g}_{(t+1):T},(n)}(H_{t},d_{t}) - Q_{t}^{\underline{g}_{(t+1):T}}(H_{t},d_{t})\right)^{2}\right] \\
    &\times E\left[\left(\frac{1}{\prod_{\ell=s}^{t-m}\hat{e}_{\ell}^{(n)}(H_{\ell},d_{\ell})} - \frac{1}{\prod_{\ell=s}^{t-m}e_{\ell}(H_{\ell},d_{\ell})} \right)^{2}\right] = \frac{o(1)}{n^{\tau}}.
\end{align*}
(ii) There exists $n_0 \in \Natural$ such that for any $n \geq n_0$ and $t=1,\ldots,T$, 
\begin{align*}
     \sup_{d_{t} \in \{0,1\},\ \underline{g}_{(t+1):T} \in \MG_{(t+1):T}}\widehat{Q}_{t}^{\underline{g}_{(t+1):T},(n)}(H_{t},d_{t}) < \infty \mbox{\ \ and\ \ }
     \sup_{d_t \in \{0,1\}}\hat{e}_{t}^{(n)}(H_{t},d_{t}) > 0.
\end{align*}
hold a.s.
\end{assumption}

\medskip

Assumption \ref{asm:rate_of_convergence_AIPW} (i) encompasses the property of double robustness; that is, Assumption \ref{asm:rate_of_convergence_AIPW} (i) is satisfied if either $\widehat{Q}_{t}^{\underline{g}_{(t+1):T},(n)}(\cdot,\cdot)$ is uniformly consistent or $\prod_{s=t}^{T}\hat{e}_{s}^{(n)}(\cdot,\cdot)$ is consistent.  As we will see later, the $\sqrt{n}$-consistency of the regret to zero can be achieved when Assumption \ref{asm:rate_of_convergence_AIPW} (i) holds with $\tau=1$. 
This condition is not very restrictive. For example, Assumption \ref{asm:rate_of_convergence_AIPW} (i) is satisfied when
\begin{align*}
&\sup_{d_{t} \in \{0,1\}}  E\left[\sup_{\underline{g}_{(t+1):T} \in \underline{g}_{(t+1):T}}\left(\widehat{Q}_{t}^{\underline{g}_{(t+1):T},(n)}(H_{t},d_{t}) - Q_{t}^{\underline{g}_{(t+1):T}}(H_{t},d_{t})\right)^{2}\right] = \frac{o(1)}{\sqrt{n}}\mbox{\ and\ }\\
    &\sup_{\underline{d}_{s:t} \in \{0,1\}^{t-s+1}} E\left[\left(\frac{1}{\prod_{\ell=s}^{t}\hat{e}_{\ell}^{(n)}(H_{\ell},d_{\ell})} - \frac{1}{\prod_{\ell=s}^{t}e_{\ell}(H_{\ell},d_{\ell})} \right)^{2}\right] = \frac{o(1)}{\sqrt{n}}
\end{align*}
hold for all $t=1,\ldots,T$ and $s=1,\ldots,t$.  

The following theorem shows the convergence rate of the welfare regret $W_{\MG}^{\ast} - W(\hat{g}^{AIPW})$.

\medskip

\begin{theorem}
\label{thm:main_theorem_AIPW}
Suppose that Assumptions \ref{asm:sequential independence}--\ref{asm:overlap} and \ref{asm:rate_of_convergence_AIPW} hold. Then
\begin{align}
    W_{\MG}^{\ast} - W(\hat{g}^{AIPW}) = O_{p}(n^{-\min\{1/2,\tau/2\}}). \label{eq:DML_result}
\end{align}
\end{theorem}

\begin{proof}
See Appendix \ref{sec:proof_DR_estimator}.
\end{proof}

\bigskip

When Assumption \ref{asm:rate_of_convergence_AIPW} holds with $\tau = 1$, the doubly robust estimator $\hat{g}^{AIPW}$ achieves the minimax optimal convergence rate $n^{-1/2}$ of welfare regret. This result is comparable with those of \cite{Athey_Wager_2020} and \cite{Zhou_et_al_2023} who study doubly robust policy learning in the static setting.

\subsection{Doubly Robust Estimation with Exogenous Variables \label{appendix:DR_approach_2}}

We say that time-varying variables are exogenous when they are not influenced by past treatment assignments. 
When treatment choice at each stage $t$ depends solely on exogenous variables and past treatment information, we can construct a doubly robust approach to estimate the optimal DTRs with less computational cost. 
This scenario is prevalent in various contexts. For example, in the context of sequential job training, variables representing exogenous economic conditions (e.g., the unemployment rate in a country where an individual resides) are not influenced by one's job training and are thus considered exogenous. 
The following assumption formalizes the exogeneity of the covariates.

\begin{assumption}\label{asm:exogenous_covariates}
    For any $t=2,\ldots,T$, $X_t(\underline{d}_{t-1}) = X_t(\underline{d}_{t-1}^{\prime})$ a.s. for any $\underline{d}_{t-1},\underline{d}_{t-1}^{\prime} \in \{0,1\}^{t-1}$.
\end{assumption}


Given our focus on using exogenous variables and past treatments exclusively for treatment choice, we redefine the observed and potential history as $H_t \equiv (\underline{D}_{t-1},\underline{X}_{t})$ and $H_t(\underline{d}_{t-1}) \equiv (\underline{d}_{t-1},\underline{X}_{t}(\underline{d}_{t-1}))$, respectively, where the history used in treatment choice does not include past outcomes.
We suppose that all underlying assumptions (Assumptions \ref{asm:sequential independence}--\ref{asm:overlap}) for the simultaneous maximization approach hold with this revised definition of $H_t$.

We denote the conditional expectation function of the weighted outcome $\gamma_t Y_t$ given $\underline{d}_{t}$ and $\underline{x}_{t}$ by $\mu_{t}(\underline{d}_{t},\underline{x}_{t}) \equiv E_{P}\left[\gamma_t Y_{t} \mid  \underline{D}_{t} = \underline{d}_{t},\underline{X}_{t} = \underline{x}_{t}\right]$.
Using this function, we can identify $W_t(\underline{g}_t)$ as follows.

\begin{lemma}\label{prop:mu_idntification}
    Under Assumptions \ref{asm:sequential independence} and \ref{asm:exogenous_covariates},
    \begin{align*}
        W_t\left(\underline{g}_t\right) = \sum_{\underline{d}_{t} \in \{0,1\}^t}E\left[\mu_t(\underline{d}_{t},\underline{X}_t)\cdot \prod_{s=1}^{t}1\{d_s = g_s(\underline{d}_{s-1},\underline{X}_s)\}\right].
    \end{align*}
\end{lemma}

\begin{proof}
    See Appendix \ref{sec:proof_DR_estimator}.
\end{proof}

For each cross-fitting fold $k$, we estimate $\mu_{t}(\underline{d}_{t},\underline{x}_{t})$ and $e_{t}(h_{t},\underline{a}_{t})$ by $\hat{\mu}_{t}^{-k}(\underline{a}_{t},s_{t})$ and $\hat{e}_{t}^{-k}(d_{t},h_{t})$, respectively, using the observations not included in the $k$-th fold. Any estimation methods, including semi/nonparametric estimators and machine learning methods, can be applied to estimate the nuisance functions $\mu_{t}(\underline{d}_{t},\underline{x}_{t})$ and $e_{t}(h_{t},\underline{a}_{t})$.

For a fixed DTR $g$, we construct an AIPW estimator of the welfare function $W(g)$ as
\begin{align*}
    \widehat{W}^{DR}(g) \equiv \frac{1}{n}\sum_{t=1}^{T}\sum_{i=1}^{n}&\left[  \frac{\gamma_{t}Y_{it}-\hat{\mu}_{t}^{-k(i)}(\underline{D}_{it},\underline{X}_{it})}{\prod_{s=1}^{t}\hat{e}_{s}^{-k(i)}(D_{is},H_{is})} \cdot
 \prod_{s=1}^{t}1\left\{D_{is}=g_{s}(H_{is})\right\}  \right. \\
 &\left. +\ \sum_{\underline{d}_{t} \in \{0,1\}^t}\left\{\hat{\mu}_{t}^{-k(i)}(\underline{d}_{t},\underline{X}_{it})\cdot \prod_{s=1}^{t}1\{d_s = g_s(\underline{d}_{s-1},\underline{X}_{is})\}\right\}\right] .
\end{align*}
We estimate the optimal DTR by maximizing $\widehat{W}^{DR}(g)$ simultaneously over $\MG$. We obtain the DTR estimator $\hat{g}^{DR}\equiv (\hat{g}_{1}^{DR},\ldots,\hat{g}_{T}^{DR})$ as the solution
\begin{align*}
\hat{g}^{DR} \in \arg\max_{g \in \MG} \widehat{W}^{DR}(g).
\end{align*}

In what follows, we show the statistical property of this approach. Let $\hat{\mu}_{t}^{(n)}(\cdot,\cdot)$ and $\hat{e}_{t}^{(n)}(\cdot,\cdot)$ denote estimators of the nuisance functions $\mu_t(\cdot,\cdot)$ and $e_{t}(\cdot,\cdot)$, respectively, using size $n$ sample randomly drawn from the underlying population.
We suppose that the following assumption holds.

\bigskip{}

\begin{assumption}\label{asm:rate_of_convergence_simultaneous}
(i) There exists $\tau >0$ such that the estimators $\left\{\hat{\mu}_{t}^{(n)}(\underline{d}_{t},\underline{x}_{t}):t=1,\ldots,T\right\}$ and $\left\{\hat{e}_{t}^{(n)}(d_{t},h_{t}):t=1,\ldots,T\right\}$ satisfy
\begin{align*}
    \sup_{\underline{d}_{t} \in \{0,1\}^t} &E_{P^n}\left[\left(\hat{\mu}_{t}^{(n)}(\underline{d}_{t},\underline{x}_{t}) - \mu_{t}(\underline{d}_{t},\underline{x}_{t})\right)^{2}\right] \\
    &\times E_{P^n}\left[\left(\frac{1}{\prod_{s=1}^{t}\hat{e}_{s}^{(n)}(d_{s},H_{is})} - \frac{1}{\prod_{s=1}^{t}e_{s}(d_{s},H_{is})} \right)^{2}\right] = \frac{o(1)}{n^{\tau^{\prime}}}
\end{align*}
for all $t=1,\ldots,T$.\\
(ii) There exists $n_0 \in \Natural$ such that for any $n \geq n_0$, $\hat{\mu}_{t}^{(n)}(\underbar{D}_{t}, \underbar{X}_{t}) < \infty$ and $\hat{e}_{t}^{(n)}(D_{t}, H_{t}) > 0$ hold a.s. for any $t$. 
\end{assumption}

\bigskip

As we will see later, the optimal $1/\sqrt{n}$ rate of convergence for the welfare regret of $\hat{g}^{DR}$ can be achieved when Assumption \ref{asm:rate_of_convergence_simultaneous} holds with $\tau^{\prime}=1$. This is not a strong or restrictive condition. For example, Assumption \ref{asm:rate_of_convergence_simultaneous} is satisfied when 
\begin{align*}
    &\sup_{\underline{d}_{t} \in \{0,1\}^t}E_{P^n}\left[\left(\hat{\mu}_{t}^{(n)}(\underline{d}_{t},\underline{X}_{it}) - \mu_{t}(\underline{d}_{t},\underline{X}_{it})\right)^{2}\right]  = \frac{o(1)}{\sqrt{n}} \mbox{ and } \\ 
    &\sup_{\underline{d}_{t} \in \{0,1\}^{t}}E_{P^n}\left[\left(\frac{1}{\prod_{s=1}^{t}\hat{e}_{s}^{(n)}(d_{s},H_{is})} - \frac{1}{\prod_{s=1}^{t}e_{s}(d_{s},H_{is})} \right)^{2}\right] = \frac{o(1)}{\sqrt{n}}
\end{align*}
hold for all $t=1,\ldots,T$. 
These conditions on the convergence rate of the mean squared errors can be satisfied even with nonparametric estimators with relatively mild conditions \citep[see, e.g.,][]{chernozhukov_et_al_2018}. Note also that Assumption \ref{asm:rate_of_convergence_simultaneous} encompasses double-robustness of the estimation of the nuisance components.

The following theorem shows the convergence rate of the welfare regret of the doubly robust estimation of the optimal DTR.

\begin{theorem}\label{thm:main_theorem_simultaneous}
Suppose that Assumptions \ref{asm:sequential independence}--\ref{asm:overlap}, \ref{asm:exogenous_covariates}, and \ref{asm:rate_of_convergence_simultaneous} hold with the redefinitions $H_t \equiv (\underline{D}_{t-1},\underline{X}_{t})$ and $H_t(\underline{d}_{t-1}) \equiv (\underline{d}_{t-1},\underline{X}_{t}(\underline{d}_{t-1}))$. Then
\begin{align}
    W_{\MG}^{\ast} - W(\hat{g}^{DR}) = O_{p}(n^{-\min\{1/2,\tau^{\prime}/2\}}). \label{eq:DML_result}
\end{align}
\end{theorem}

\begin{proof}
See Appendix \ref{sec:proof_DR_estimator}.
\end{proof}

\bigskip

When Assumption \ref{asm:rate_of_convergence_simultaneous} holds with $\tau^{\prime} = 1$, the doubly robust estimator $\hat{g}^{DR}$ achieves the minimax optimal convergence rate $1/\sqrt{n}$ of welfare regret. This result is comparable with those of \citeApp{Athey_Wager_2020} and \citeApp{Zhou_et_al_2023} who study doubly robust policy learning in static settings. 


\section{Proofs \label{appendix:proof}}

This appendix provides the proofs of Lemma \ref{lem:vc_subclass} and Theorems \ref{thm:lower bound}, \ref{thm:estimated propensity score}, \ref{thm:equality_mindied_welfare}, \ref{thm:main_theorem_AIPW}, and \ref{thm:main_theorem_simultaneous}.

\subsection{Proof of Lemma \ref{lem:vc_subclass}}

We first present the definitions of the VC-dimension of a class of indicator functions and relevant concepts as follows.

\medskip
\begin{definition}[VC-dimension of a Class of Indicator Functions]\label{def:VC-dimension of a class of binary functions}
Let $\MZ$ be an arbitrary space and $\MF$ be a class of indicator functions from $\MZ$ to $\{0,1\}$.
For a finite sample $S=\left(z_{1},\ldots,z_{m}\right)$ of $m\ge1$
points in $\MZ$, we define the set of dichotomies as $\Pi_{\mathcal{F}}\left(S\right)\equiv\left\{ \left(f\left(z_{1}\right),\ldots,f\left(z_{m}\right)\right):f\in\mathcal{F}\right\} $,
which is all possible assignments of $S$ by functions in $\MF$.
We say that $S$ is shattered by $\mathcal{F}$ when $\left|\Pi_{\mathcal{F}}\left(S\right)\right|=2^{m}$;
that is, $\mathcal{F}$ realizes all possible dichotomies of $S$.
Then the VC-dimension of $\mathcal{F}$, denoted by $VC(\MF)$, is defined to be the size of the largest sample $S$ shattered by $\mathcal{F}$, i.e., 
\begin{align*}
 VC(\MF) \equiv \max\left\{ m:\max_{S=\left(z_{1},\ldots,z_{m}\right)\subseteq\MZ}\left|\Pi_{\mathcal{F}}\left(S\right)\right|=2^{m}\right\}.    
\end{align*}
We say that $\mathcal{F}$ is a VC-class of indicator functions if $VC(\MF)< \infty$.\footnote{For example, the class of linear treatment rules
\begin{align*}
{\cal G}_{t}=  \left\{ 1\left\{ \beta_{1t}^{\prime}\text{x}_{t}+\beta_{2t}^{\prime}\text{\ensuremath{\underline{d}}}_{t-1}+\beta_{3t}^{\prime}\text{\ensuremath{\underline{y}}}_{t-1}\geq c_{t}\right\} :\left(\beta_{1t}^{\prime},\beta_{2t}^{\prime},\beta_{3t}^{\prime},c_{t}\right)^{\prime}\in\mathbb{R}^{k+2t-1}\right\}.
\end{align*}
has VC-dimension of at most $k+2t-1$.}
\end{definition}
\medskip

We next introduce the VC-dimension for a class of subsets. Let $\MZ$ be any space, and let $\mathbf{z}_{\ell}=\left(z_{1},\ldots,z_{\ell}\right)$
be a finite set of $\ell \geq1$ points in $\MZ$. Given a class
of subsets $\MG \subseteq 2^{\MZ}$ and a subset $\mathbf{\tilde{\Bz}}$
of $\Bz_{\ell}$, we say that $\MG$ picks out $\mathbf{\tilde{\Bz}}$
when $\tilde{\mathbf{z}}\cap G=\tilde{\mathbf{z}}$ holds for some $G\in\mathcal{G}$.
We say that $\MG$ shatters $\mathbf{z}_{\ell}$ when $\left|\left\{ \mathbf{z}_{\ell}\cap G:G\in\MG\right\} \right|=2^{\ell}$ holds,
that is all subsets of $\mathbf{z}_{\ell}$ are picked out by $\mathcal{G}$.
The VC-dimension of the class of subsets $\MG$, denoted by $VC(\MG)$, is defined
as the cardinality of the largest subset $\mathbf{z}_{\ell}$ contained in $\mathcal{Z}$
and shattered by $\MG$, i.e.,
\begin{align*}
   VC(\MG) \equiv \max\{\ell : \max_{\mathbf{z}_{\ell}\subseteq\mathcal{Z}}\left|\left\{ \mathbf{z}_{\ell}\cap G:G\in\mathcal{G}\right\} \right|=2^{\ell}\}.
\end{align*}
We say that a class of subsets $\MG$ is a VC-class of subsets if $VC(\MG)< \infty$.

We next introduce a concept of the subgraph of a real-valued function
$f:\mathcal{Z}\rightarrow\mathbb{R}$ that is the set 
\begin{align*}
\mbox{SG}\left(f\right)\equiv  \left\{ \left(z,t\right)\in\MZ \times\mathbb{R}:t\leq f\left(z\right)\right\} .
\end{align*}
Let $\mbox{SG}\left(\mathcal{F}\right)\equiv \{\mbox{SG}\left(f\right): f\in \mathcal{F} \}$ be a collection of subgraphs over a class of functions $\mathcal{F}$.
We here consider the VC-dimension of $\mbox{SG}(\MF)$ as a complexity measure of $\mathcal{F}$. Note that in the case of $\MF$ being a class of indicator functions, the VC-dimension of $\mbox{SG}\left(\mathcal{F}\right)$ corresponds to the VC-dimension of $\mathcal{F}$ in the sense of Definition \ref{def:VC-dimension of a class of binary functions}. 
We say that a class of functions $\mathcal{F}$ is a VC-subgraph class of functions if $VC(\mbox{SG}\left(\mathcal{F}\right)) < \infty$.

The following lemmas are auxiliary lemmas for Lemma \ref{lem:vc_subclass}.

\medskip

\begin{lemma} \label{lem:Sauer's lemma} (Sauer's lemma; see, for example, Theorem 3.6.2 of \cite{Gine_Nickl_2016}) 
Let $\MZ$ be any space, and $\left(z_{1},\ldots,z_{\ell}\right)$
be a finite set of $\ell \geq1$ points in $\MZ$. Let $\MG$ be a VC-class of subsets in $\MZ$ with $VC(\MG)=v < \infty$. Let $\Delta_{\ell}(\MG,(z_1,\ldots,z_\ell))$ denote the number of subsets of $(z_1,\ldots,z_\ell)$ that are picked out by $\MG$, i.e.,
\begin{align*}
    \Delta_{\ell}\left(\MG,(z_1,\ldots,z_\ell)\right) \equiv 
    \left| \left\{(z_1,\ldots,z_\ell) \cap G: G \in \MG \right\} \right|.
\end{align*}
Then the following holds:
\begin{align*}
    \max_{(z_1,\ldots,z_\ell) \subseteq \MZ}\Delta_{\ell}\left(\mathcal{G},\left(z_1,\ldots,z_\ell\right)\right) \leq \sum_{j=0}^{v}\left(\begin{array}{c}
\ell\\
j
\end{array}\right)
\leq \left(\frac{\ell e}{v}\right)^{v}.
\end{align*}
\end{lemma}
\medskip

\begin{lemma}\label{lem:lemma A1_KT2018}  Let $\mathcal{Z}=\mathcal{Z}_{1}\times\mathcal{Z}_{2}$
be any product space, and $\mathcal{G}$ be a class of indicator functions from $\MZ_{2}$ to $ \{0,1\}$. Suppose that $\mathcal{G}$ has VC-dimension $v\geq0$ in the sense of Definition \ref{def:VC-dimension of a class of binary functions}.
Fix a function $f$ on $\mathcal{Z}$, and define a class of functions
on $\mathcal{Z}$:
\begin{align*}
\mathcal{F}_{\mathcal{G}}= \left\{ f \cdot g:g\in\MG \right\} .
\end{align*}
Then $\mathcal{F}_{\mathcal{G}}$ is a VC-subgraph class of functions
with $VC(\mbox{SG}(\MF_\MG)) \leq v$. 
\end{lemma}

\begin{proof}
We prove the statement by contradiction. Suppose that there exist some $(v+1)$-points $\left\{ \left(z_{1},t_{1}\right),\ldots,\left(z_{v+1},t_{v+1}\right)\right\} \equiv \left\{ \left(z_{1,1},z_{2,1},t_{1}\right),\ldots,\left(z_{1,v+1},z_{2,v+1},t_{v+1}\right)\right\} \subset \MZ_{1}\times \MZ_{2}\times\mathbb{R}$ that are shattered by $\mbox{SG}\left(\mathcal{F}_{\mathcal{G}}\right)$.

When $t\leq f\left(z\right)\wedge0$ or $t>f\left(z\right)\vee0$
for some $\left(z,t\right)\in\left\{ \left(z_{1},t_{1}\right),\ldots, \left(z_{v+1},t_{v+1}\right)\right\} $,
$\mbox{SG}\left(\mathcal{F}_{\mathcal{G}}\right)$ cannot pick out
$\left\{ \left(z_{1},t_{1}\right),\ldots\left(z_{v+1},t_{v+1}\right)\right\} \backslash\{\left(z,t\right)\}$
or $\left\{ \left(z,t\right)\right\} $. Thus, we need to consider
only the case that $\left(f\left(z\right)\wedge0\right)<t\leq\left(f\left(z\right)\vee0\right)$
for all $\left(z,t\right)\in\left\{ \left(z_{1},t_{1}\right),\ldots\left(z_{v+1},t_{v+1}\right)\right\} $.
In the remaining case, we indicate $\delta_{j}=1$ if $t_{j}\leq f(z_{j})$ and $\delta_{j}=0$ otherwise.
Since the VC-dimension of $\MG$ is at most $v$ in the sense of Definition \ref{def:VC-dimension of a class of binary functions}, there exists a subset $S\equiv (\tilde{z}_{2,1},\ldots,\tilde{z}_{2,m})$ (for some $m>0$) of $\left\{ z_{2,1},\ldots,z_{2,v+1}\right\}$ such that $(g(\tilde{z}_{2,1}),\ldots,g(\tilde{z}_{2,m}))\neq (1,\ldots,1)$ and $(g(z_{2,1}),\ldots,g(z_{2,v+1})\backslash(g(\tilde{z}_{2,1}),\ldots,g(\tilde{z}_{2,m}))\neq (0,\ldots,0)$ for any $g \in \MG$.
Then $\mbox{SG}\left(\mathcal{F}_{\mathcal{G}}\right)$
cannot pick out the following subset: 
\begin{align*}
\left\{ \left(z_{j},t_{j}\right):\left(z_{2,j}\in S\mbox{ and }\delta_{j}=1\right)\mbox{ or }\left(z_{2,j}\notin S\mbox{ and }\delta_{j}=0\right)\right\} ,
\end{align*}
because this set of points could be contained in $\mbox{SG}\left(f\cdot g\right)$
only when $\mbox{sign}(t_{j})=\mbox{sign}(g\left(z_{2,j}\right)-1/2)$ for all $j=1,\ldots,v+1$. 
This contradicts the assumption that $\left\{ \left(z_{1},t_{1}\right),\ldots\left(z_{v+1},t_{v+1}\right)\right\} \subset\mathcal{Z}\times\mathbb{R}$ is shattered by $\mbox{SG}\left(\mathcal{F}_{\mathcal{G}}\right)$.
\end{proof}
\medskip
\begin{proof}[Proof of Lemma \ref{lem:vc_subclass}]
We prove for the case that $s=1$ and $t=T$. The result follows for the remaining cases by a similar argument. Let $m$ be an arbitrary integer and $\left(z_{1},\ldots,z_{m}\right)$ be $m$ arbitrary points on
$\mathcal{Z}$. For each $t$, fixing $g_{s}\in\MG_{s}$ for all $s\neq t$, define a class of functions 
\begin{align*}
\widetilde{{\cal F}}_{t}\equiv  \left\{ f\left(z\right)=1\left\{ g_{1}\left(h_{1}\right)=d_{1},\ldots,g_{T}\left(h_{T}\right)=d_{T}\right\} \cdot r(z):g_{t}\in\mathcal{G}_{t}\right\}, 
\end{align*}
 and, fixing $g_{s}\in\mathcal{G}_{s}$ for all $s> t$, define
\begin{align*}
\widetilde{{\cal F}}_{1:t} \equiv & \left\{ f\left(z\right)=1\left\{ g_{1}\left(h_{1}\right)=d_{1},\ldots,g_{T}\left(h_{T}\right)=d_{T}\right\} \cdot r(z):\left(g_{1},\ldots,g_{t}\right)\in\mathcal{G}_{1}\times\cdots\times\mathcal{G}_{t}\right\} .
\end{align*}

We first consider $\widetilde{\MF}_{1}$, or equivalently $\widetilde{\MF}_{1:1}$.
Applying Lemma \ref{lem:lemma A1_KT2018} to $\widetilde{\MF}_{1}$ shows that $\widetilde{\MF}_{1}$ is a VC-subgraph of functions with $VC(\mbox{SG}(\widetilde{\MF}_{1})) \leq v_{1}$.
Therefore, from Lemma \ref{lem:Sauer's lemma}, SG$(\widetilde{\MF_{1}})$
can pick out at most $O(m^{v_{1}})$ subsets from $(z_{1},\ldots,z_{m})$.

Next we study $\widetilde{\MF}_{2}$ and then $\widetilde{\MF}_{1:2}$.
Let $(z_{1},\ldots,z_{m^{\prime}})$ be an arbitrary subset
picked out by $\mbox{SG}(\tilde{f_{1}})$ where $\tilde{f}_{1}\in\widetilde{\MF}_{1:1}$
has a fixed $g_{1}\in{\cal G}_{1}$. Lemmas \ref{lem:Sauer's lemma} and \ref{lem:lemma A1_KT2018} show that SG$(\widetilde{\mathcal{F}}_{2})$ can pick out at most $O(m^{v_{2}})$ subsets from $(z_{1},\ldots,z_{m^{\prime}})$. Because SG$(\widetilde{\mathcal{F}}_{2})$ can pick out at most $O(m^{v_{2}})$
subsets from each subset of $(z_{1},\ldots,z_{m})$ and
SG$(\widetilde{{\cal F}}_{1:1})$ can pick out at most $O(m^{v_{1}})$
subsets from $(z_{1},\ldots,z_{m})$, by varying $(g_{1},g_{2})$
over $\mathcal{G}_{1}\times\mathcal{G}_{2}$, SG$(\widetilde{\mathcal{F}}_{1:2})$ picks out at most $O(m^{v_{1}+v_{2}})$ subsets from $(z_{1},\ldots,z_{m})$.

For $s\geq2$, suppose that $\widetilde{\MF}_{1:s-1}$
can pick out at most $O(m^{\sum_{t=1}^{s-1}v_{t}})$ subsets from $(z_{1},\ldots,z_{m})$.
Let $(z_{1},\ldots,z_{m^{\prime}})$ be an arbitrary subset
picked out by SG$(\tilde{f}_{s-1})$ where $\tilde{f}_{s-1}\in\widetilde{\MF}_{1:s-1}$ has fixed $g_{1},\ldots,g_{s-1}$. From $(z_{1},\ldots,z_{m^{\prime}})$,
SG$(\tilde{\MF}_{s})$ can pick out at most $O(m^{v_{s}})$
subsets. Combining this result with the fact that $\widetilde{\MF}_{1:s-1}$
can pick out at most $O(m^{\sum_{t=1}^{s-1}v_{t}})$ subsets from
$(z_{1},\ldots,z_{m})$ leads to the conclusion that $\widetilde{\MF}_{1:s}$
picks out at most $O(m^{\sum_{t=1}^{s}v_{t}})$ subsets from
$(z_{1},\ldots,z_{m})$.

Recursively, we can prove that $\widetilde{\mathcal{F}}_{1:T}$ picks
out at most $O(m^{\sum_{t=1}^{T}v_{t}})$ subsets from $(z_{1},\ldots,z_{m})$. Hence, SG$(\widetilde{\mathcal{F}}_{1:T})$ is a VC-subgraph class of functions with VC-dimension less than or equal to $\sum_{t=1}^{T}v_{t}$.
\end{proof}
\medskip

\subsection{Proof of Theorem \ref{thm:lower bound}.}

This appendix present the proof of Theorem \ref{thm:lower bound}. The following is its auxiliary lemma, where we use the same strategy as the proofs of Theorem 2 of \citet{massart2006risk} and Theorem 2.2 of \cite{Kitagawa_Tetenov_2018a}, but extend it to the dynamic treatment setting.

\medskip

\begin{lemma}\label{lem:lower bound}
Suppose that Assumptions \ref{asm:sequential independence}, \ref{asm:bounded outcome}, and \ref{asm:overlap} hold for any distribution $P\in{\cal P}\left(M, \kappa, \MG\right)$ and Assumption \ref{asm:vc-class} holds for ${\cal G}$. Fix $t \in \{1,\ldots,T\}$, and let $\gamma_{t}=1$ and $\gamma_{s}=0$ for all $s\neq t$. Then, for
any DTR $\hat{g}\in\mathcal{G}$ as a function of $\left(Z_{1},\ldots,Z_{n}\right)$,
\begin{align*}
\sup_{P\in{\cal P}\left(M, \kappa, \MG\right)}E_{P^{n}}\left[W_{{\cal G}}^{\ast}-W\left(\hat{g}\right)\right] & \geq2^{-1}\exp\left(-2\right)M_{t}\sqrt{\frac{v_{1:t}}{n}}
\end{align*}
holds for all $n\geq16v_{1:t}$. This result holds irrespective of whether Assumption \ref{asm:first-best} additionally holds for a pair of $\MG$ and any $P \in \MP(M, \kappa, \MG)$ or not.
\end{lemma}

\begin{proof}
The proof follows by constructing a specific subclass of ${\cal P}\left(M, \kappa, \MG\right)$, for which the worst-case average welfare regret can be bounded from below.
We here prove the statement for the lemma in the case that $t=T$ (i.e., $\gamma_{T}=1$ and $\gamma_{s}=0$ for $s\neq T$). 
The proof follows for the remaining cases by a similar argument.
For simplicity, we normalize the support of the potential outcomes to $Y_{t}\left(\text{\ensuremath{\underline{d}}}_{t}\right)\in \left[-1/2,1/2\right]$
for all $\underline{d}_t \in \{0,1\}^{t}$ and $t=1,\ldots,T$. We also suppose that $X_{t}(\underline{d}_{t-1}) = X_{t}(\underline{d}_{t-1}^{\prime})$ for any $t \geq 2$ and any $\underline{d}_{t-1},\underline{d}_{t-1}^{\prime} \in \{0,1\}^{t-1}$; that is, the covariates do not depend on the past treatments. Let $\mathbf{1}_{T}$ denote a $T$-dimensional vector of ones.

We construct a specific subclass $\MP^{\ast}\subset{\cal P}(\mathbf{1}_{T},\kappa)$ as follows. Let $\tilde{Z} \equiv ((D_{t},X_{t},Y_{t})_{t=1}^{T-1},D_{T},X_{T})$, which is a vector of all the observed variables excluding $Y_{T}$, and denote its space by $\widetilde{\MZ}$.  Let $\tilde{z}_{1},\ldots,\tilde{z}_{v_{1:T}}$
be $v_{1:T}$ points in $\widetilde{\MZ}$ such that a 
set $\left\{(\tilde{z}_{1},1/2),\ldots,(\tilde{z}_{v_{1:T}},1/2)\right\}$
is shattered by a collection of indicator functions
\begin{align*}
\left\{ f\left(z\right)=1\left\{ g_{1}\left(h_{1}\right)=d_{1},\cdots,g_{T}\left(h_{T}\right)=d_{T}\right\} :\left(g_{1},\ldots,g_{T}\right)\in\mathcal{G}_{1}\times\cdots\times\mathcal{G}_{T}\right\}
\end{align*}
in the sense of Definition \ref{def:VC-dimension of a class of binary functions}.
For $j=1,\ldots,v_{1:T}$, denote
    $\tilde{z}_{j}=\left(\left(d_{tj},x_{tj},y_{tj}\right)_{t=1}^{T-1},d_{Tj},x_{Tj}\right)\in \widetilde{\MZ}$.
For any $P\in{\cal P}^{\ast}$, we suppose for the marginal distributions of $\tilde{Z} $ on $\widetilde{\MZ}$ that
$P\left(\tilde{Z}=\tilde{z}_{j}\right)=1/v_{1:T}$ for each $j=1,\ldots,v_{1:T}$.
Let $\mathbf{b}=\left(b_{1},\ldots,b_{v_{1:T}}\right)\in\left\{ 0,1\right\} ^{v_{1:T}}$
be a bit vector that indexes a member of ${\cal P}^{\ast}$. Hence
${\cal P}^{\ast}$ consists of $2^{v_{1:T}}$ distinct DGPs. For
each $j=1,\ldots,v_{1:T}$, depending on $\mathbf{b}$, we construct
the following conditional distribution of $Y_{T}\left(\text{\ensuremath{\underline{d}}}_{T}\right)$
given $\tilde{Z}=\tilde{z}_{j}$: if $b_{j}=1$, 
\begin{align*}
Y_{T}\left(\text{\ensuremath{\underline{d}}}_{Tj}\right)= & \begin{cases}
\begin{array}{c}
1/2\\
-1/2
\end{array} & \begin{array}{c}
\mbox{w.p. }1/2+\delta\\
\mbox{w.p. }1/2-\delta
\end{array}\end{cases},
\end{align*}
otherwise 
\begin{align*}
Y_{T}\left(\text{\ensuremath{\underline{d}}}_{Tj}\right)= & \begin{cases}
\begin{array}{c}
1/2\\
-1/2
\end{array} & \begin{array}{c}
\mbox{w.p. }1/2-\delta\\
\mbox{w.p. }1/2+\delta
\end{array}\end{cases},
\end{align*}
where $\underline{d}_{Tj}$ is the history of the realized treatments from stage $1$ to $T$ when $\tilde{Z} = \tilde{z}_{j}$, and $\delta\in\left[0,1/2\right]$ is chosen properly in a later
step of the proof. When $b_{j}=1$, $E_P\left[Y_{T}\left(\text{\ensuremath{\underline{d}}}_{Tj}\right)\mid\tilde{Z}=\tilde{z}_{j}\right]=\delta$;
otherwise, $E_P\left[Y_{T}\left(\text{\ensuremath{\underline{d}}}_{Tj}\right)\mid\tilde{Z}=\tilde{z}_{j}\right]=-\delta$.
For conditional distributions of the other potential outcomes $Y_{T}\left(\text{\ensuremath{\underline{d}}}_{T}\right)$
given $\tilde{Z}=\tilde{z}_{j}$, we set $Y_{T}\left(\text{\ensuremath{\underline{d}}}_{T}\right)=0$
with probability 1 if $\text{\ensuremath{\underline{d}}}_{T}\neq\text{\ensuremath{\underline{d}}}_{Tj}$.

When $\mathbf{b}$ is known, an optimal DTR, denoted by $g_{\mathbf{b}}^{\ast}=\left(g_{1,\mathbf{b}}^{\ast},\ldots,g_{T,\mathbf{b}}^{\ast}\right)$, is such that
\begin{align*}
(g_{1,\mathbf{b}}^{\ast}(h_{1j}),\ldots,g_{T,\mathbf{b}}^{\ast}(h_{Tj}))= & \begin{cases}
\begin{array}{c}
\text{\ensuremath{\underline{d}}}_{Tj}\\
\left(1-d_{1j},\cdots,1-d_{Tj}\right)
\end{array} & \begin{array}{c}
\mbox{if }b_{j}=1\\
\mbox{otherwise}
\end{array}\end{cases}
\end{align*}
for $j=1,\ldots,v_{1:T}$, where $h_{tj}$ is the history information in $\tilde{z}_{j}$ up to stage $t$. Such a DTR is feasible in ${\cal G}$. Then, the optimized social
welfare given $\mathbf{b}$ is 
\begin{align*}
W\left(g_{\mathbf{b}}^{\ast}\right)=  \frac{1}{v_{1:T}}\delta\sum_{j=1}^{v_{1:T}}b_{j}.
\end{align*}
Let $\hat{g}=\left(\hat{g}_{1},\ldots,\hat{g}_{T}\right):{\cal H}_{1}\times\cdots\times{\cal H}_{T}\mapsto\left\{ 0,1\right\} ^{T}$
be an arbitrary DTR depending on the sample $\left(Z_{1},\ldots,Z_{n}\right)$,
and let $\hat{\mathbf{b}}\in\left\{ 0,1\right\} ^{v_{1:T}}$ be a
binary vector such that its $j$-th element is given by
\begin{align*}
    \hat{b}_{j}=1\left\{ \hat{g}_{1}\left(h_{1j}\right)=d_{1j},\ldots,\hat{g}_{T}\left(h_{Tj}\right)=d_{Tj}\right\}.
\end{align*}
We define by $g\left(\mathbf{b}\right)$ a prior of $\mathbf{b}$
such that $b_{1},\ldots,b_{v_{1:T}}$ are i.i.d and $b_{1}\sim\mbox{Ber}\left(1/2\right)$.

Then the maximum average welfare regret on $\MP\left(\mathbf{1}_{T},\kappa\right)$ satisfies the following:
\begin{align}
 & \sup_{P\in \MP\left(\mathbf{1}_{T},\kappa\right)}E_{P^{n}}\left[W_{{\cal G}}^{\ast}-W\left(\hat{g}\right)\right]\nonumber \\
\geq &\ \sup_{P_{\mathbf{b}}\in{\cal P}^{\ast}}E_{P_{\Bb}^{n}}\left[W\left(g_{\mathbf{b}}^{\ast}\right)-W\left(\hat{g}\right)\right]\geq\int_{\mathbf{b}}E_{P_{\Bb}^{n}}\left[W\left(g_{\mathbf{b}}^{\ast}\right)-W\left(\hat{g}\right)\right]dg(\mathbf{b})\nonumber \\
\geq &\ \delta\int_{\mathbf{b}}\int_{Z_{1},\ldots,Z_{n}}P_{\tilde{Z}}\left(\left\{ b(\tilde{Z})\neq\hat{b}(\tilde{Z})\right\} \right)dP_{\mathbf{b}}^{n}\left(Z_{1},\ldots,Z_{n}\right)dg(\mathbf{b})\nonumber \\
\geq &\ \inf_{\hat{g}\in\mathcal{G}}\delta\int_{\mathbf{b}}\int_{Z_{1},\ldots,Z_{n}}P_{\tilde{Z}}\left(\left\{ b(\tilde{Z})\neq\hat{b}(\tilde{Z})\right\} \right)dP_{\mathbf{b}}^{n}\left(Z_{1},\ldots,Z_{n}\right)dg(\mathbf{b}),\nonumber
\end{align}
where $P_{\tilde{Z}}$ is a probability measure of $\tilde{Z}$, and
$b(\tilde{Z})$ and $\hat{b}(\tilde{Z})$ are
elements of $\mathbf{b}$ and $\hat{\mathbf{b}}$ such
that $b\left(\tilde{z}_{j}\right)=b_{j}$ and $\hat{b}\left(\tilde{z}_{j}\right)=\hat{b}_{j}$, respectively.
Note that the above minimization problem can be seen as the minimization of the Bayes risk when the loss function corresponds to the classification error for predicting the binary random variable $b(\tilde{Z})$.
Hence, the risk is minimized by the Bayes classifier such that for
each $j=1,\ldots,J$, 
\begin{align*}
\hat{b}^{\ast}\left(\tilde{z}_{j}\right)= & \begin{cases}
\begin{array}{c}
1\\
0
\end{array} & \begin{array}{c}
\mbox{if }g\left(b_{j}=1\mid Z_{1},\ldots,Z_{n}\right)\geq1/2\\
\mbox{otherwise}
\end{array}\end{cases},
\end{align*}
where $g\left(b_{j}=1\mid Z_{1},\ldots,Z_{n}\right)$ is the posterior distribution for $b_{j}=1$. This Bayes
classifier is achieved by a DTR $\hat{g}^{\ast}\equiv\left(\hat{g}_{1}^{\ast},\ldots,\hat{g}_{T}^{\ast}\right)$
that satisfies for $j=1,\ldots,j$, 
\begin{align*}
\left(\hat{g}_{1}^{\ast}\left(h_{1j}\right),\ldots,\hat{g}_{T}^{\ast}\left(h_{Tj}\right)\right)=  \begin{cases}
\begin{array}{c}
\text{\ensuremath{\underline{d}}}_{Tj}\\
\left(1-d_{1j},\ldots,1-d_{Tj}\right)
\end{array} & \begin{array}{c}
\mbox{if }g\left(b_{j}=1\mid Z_{1},\ldots,Z_{n}\right)\geq1/2\\
\mbox{otherwise}
\end{array}\end{cases}.
\end{align*}
Note that $\hat{g}^{\ast}$ is feasible in $\mathcal{G}$.

Then, using $\hat{g}^{\ast}$, the minimized risk is given by 
\begin{align}
 & \delta\int_{Z_{1},\ldots,Z_{n}}E_{\tilde{Z}}\left[\min\left\{ g\left(b\left(\tilde{Z}\right)=1\mid Z_{1},\ldots,Z_{n}\right),1-g\left(b\left(\tilde{Z}\right)=1\mid Z_{1},\ldots,Z_{n}\right)\right\} \right]d\tilde{P}^{n} \notag \\
= & \frac{1}{v_{1:T}}\delta\int_{Z_{1},\ldots,Z_{n}}\sum_{j=1}^{v_{1:T}}\left[\min\left\{ g\left(b_{j}=1\mid Z_{1},\ldots,Z_{n}\right),1-g\left(b_{j}=1\mid Z_{1},\ldots,Z_{n}\right)\right\} \right]d\tilde{P}^{n},
\label{eq:proof-lower bound-1}
\end{align}
where $\tilde{P}$ is the marginal likelihood of $\left\{ \left\{ Y_{it}\left(\text{\ensuremath{\underline{d}}}_{t}\right)\right\} _{\text{\ensuremath{\underline{d}}}_{t}\in\left\{ 0,1\right\} ^{t}},\left\{ D_{it},X_{it}\right\} \right\} _{i=1,\ldots,n;t=1,\ldots,T}$
with prior $g(\mathbf{b})$.

For each $j=1,\ldots,v_{1:T}$, let 
\begin{align*}
k_{j}^{+}&=  \#\left\{ i:\tilde{Z}_{i}=\tilde{z}_{j},Y_{iT}=\frac{1}{2}\right\} ,\\
k_{j}^{-}&=  \#\left\{ i:\tilde{Z}_{i}=\tilde{z}_{j},Y_{iT}=-\frac{1}{2}\right\} .
\end{align*}
Then the posteriors for $b_{j}=1$ can be written as 
\begin{align*}
g\left(b_{j}=1\mid Z_{1},\ldots,Z_{n}\right)=  \begin{cases}
\begin{array}{c}
\frac{1}{2}\\
\frac{\left(\frac{1}{2}+\delta\right)^{k_{j}^{+}}\left(\frac{1}{2}-\delta\right)^{k_{j}^{-}}}{\left(\frac{1}{2}+\delta\right)^{k_{j}^{+}}\left(\frac{1}{2}-\delta\right)^{k_{j}^{-}}+\left(\frac{1}{2}+\delta\right)^{k_{j}^{-}}\left(\frac{1}{2}-\delta\right)^{k_{j}^{+}}}
\end{array} & \begin{array}{c}
\mbox{if }k_{j}^{+}+k_{j}^{-}=0\\
\mbox{otherwise}.
\end{array}\end{cases}
\end{align*}
Hence, the following holds: 
\begin{align}
 & \min\left\{ g\left(b_{j}=1\mid Z_{1},\ldots,Z_{n}\right),1-g\left(b_{j}=1\mid Z_{1},\ldots,Z_{n}\right)\right\} \notag \\
=\ & \frac{\min\left\{ \left(\frac{1}{2}+\delta\right)^{k_{j}^{+}}\left(\frac{1}{2}-\delta\right)^{k_{j}^{-}},\left(\frac{1}{2}+\delta\right)^{k_{j}^{-}}\left(\frac{1}{2}-\delta\right)^{k_{j}^{+}}\right\} }{\left(\frac{1}{2}+\delta\right)^{k_{j}^{+}}\left(\frac{1}{2}-\delta\right)^{k_{j}^{-}}+\left(\frac{1}{2}+\delta\right)^{k_{j}^{-}}\left(\frac{1}{2}-\delta\right)^{k_{j}^{+}}}\notag \\
=\ & \frac{\min\left\{ 1,\left(\frac{\frac{1}{2}+\delta}{\frac{1}{2}-\delta}\right)^{k_{j}^{+}-k_{j}^{-}}\right\} }{1+\left(\frac{\frac{1}{2}+\delta}{\frac{1}{2}-\delta}\right)^{k_{j}^{+}-k_{j}^{-}}} \notag \\
=\ & \frac{1}{1+a^{\left|k_{j}^{+}-k_{j}^{-}\right|}},\mbox{ where }a=\frac{1+2\delta}{1-2\delta}>1.
\label{eq:proof-lower bound-2}
\end{align}
Since 
\begin{align*}
k_{j}^{+}-k_{j}^{-} &=  \sum_{i:\tilde{Z}_{i}=\tilde{z}_{j}}2Y_{Ti},
\end{align*}
plugging (\ref{eq:proof-lower bound-2}) into (\ref{eq:proof-lower bound-1})
yields 
\begin{align*}
(\ref{eq:proof-lower bound-1}) =  \frac{1}{v_{1:T}}\delta\sum_{j=1}^{v_{1:T}}E_{\tilde{P}^{n}}\left[\frac{1}{1+a^{\left|\sum_{i:\tilde{Z}_{i}=\tilde{z}_{j}}2Y_{iT}\right|}}\right]& \geq  \frac{\delta}{2v_{1:T}}\sum_{j=1}^{v_{1:T}}E_{\tilde{P}^{n}}\left[\frac{1}{a^{\left|\sum_{i:\tilde{Z}_{i}=\tilde{z}_{j}}2Y_{iT}\right|}}\right]\\
&\geq  \frac{\delta}{2v_{1:T}}\sum_{j=1}^{v_{1:T}}a^{-E_{\tilde{P}^{n}}\left|\sum_{i:\tilde{Z}_{i}=\tilde{z}_{j}}2Y_{iT}\right|},
\end{align*}
where $E_{\tilde{P}^{n}}\left[\cdot\right]$ is the expectation with
respect to the marginal likelihood of 
\begin{align*}
\left\{ \left\{ Y_{it}\left(\text{\ensuremath{\underline{d}}}_{t}\right)\right\} _{\text{\ensuremath{\underline{d}}}_{t}\in\left\{ 0,1\right\} ^{t}},D_{it},X_{it}\right\} _{i=1,\ldots,n;t=1,\ldots,T}.    
\end{align*}
The first inequality follows by $a>1$, and the second inequality
follows by Jensen's inequality. Given our prior distribution for $\mathbf{b}$,
for each $\text{\ensuremath{\underline{d}}}_{T}\in\left\{ 0,1\right\} ^{T}$,
the marginal distribution of $Y_{iT}\left(\underline{d}_{T}\right)$ is
$P\left(Y_{iT}\left(\text{\ensuremath{\underline{d}}}_{T}\right)=1/2\right)=P\left(Y_{iT}\left(\text{\ensuremath{\underline{d}}}_{T}\right)=-1/2\right)=1/2$
if there exist $\text{\ensuremath{\underline{d}}}_{Tj}$ among $\text{\ensuremath{\underline{d}}}_{T1},\ldots,\text{\ensuremath{\underline{d}}}_{Tv_{1:T}}$
such that $\text{\ensuremath{\underline{d}}}_{Tj}=\text{\ensuremath{\underline{d}}}_{T}$;
otherwise, $P\left(Y_{iT}\left(\text{\ensuremath{\underline{d}}}_{T}\right)=0\right)=1$.
Thus, we have 
\begin{align*}
E_{\tilde{P}^{n}}\left|\sum_{i:\tilde{Z}_{i}=\tilde{z}_{j}}2Y_{iT}\right|
&= E_{\tilde{P}^{n}}\left|\sum_{i:\tilde{Z}_{i}=\tilde{z}_{j}}2Y_{iT}\left(\text{\ensuremath{\underline{d}}}_{Tj}\right)\right|\\
&=  \sum_{k=0}^{n}\left(\begin{array}{c}
n\\
k
\end{array}\right)\left(\frac{1}{v_{1:T}}\right)^{k}\left(1-\frac{1}{v_{1:T}}\right)^{n-k}E\left|B\left(k,\frac{1}{2}\right)-\frac{k}{2}\right|,
\end{align*}
where $B\left(k,1/2\right)$ is the binomial random variable with parameters $k$ and $1/2$. By the Cauchy-Schwarz inequality, it follows
that 
\begin{align*}
E\left|B\left(k,\frac{1}{2}\right)-\frac{1}{2}\right|\leq  \sqrt{E\left(B\left(k,\frac{1}{2}\right)-\frac{k}{2}\right)^{2}}=\sqrt{\frac{k}{4}}.
\end{align*}
Thus, we obtain 
\begin{align*}
E_{\tilde{P}^{n}}\left|\sum_{i:\tilde{Z}_{i}=\tilde{z}_{j}}2Y_{iT}\right| & \leq\sum_{k=0}^{n}\left(\begin{array}{c}
n\\
k
\end{array}\right)\left(\frac{1}{v_{1:T}}\right)^{k}\left(1-\frac{1}{v_{1:T}}\right)^{n-k}\sqrt{\frac{k}{4}}\\
 & =E\sqrt{\frac{B\left(n,\frac{1}{v_{1:T}}\right)}{4}}\\
 & \leq\sqrt{\frac{n}{4v_{1:T}}},
\end{align*}
where the last inequality follows by Jensen's inequality. Hence, the
Bayes risk is bounded from below by 
\begin{align*}
\frac{\delta}{2}a^{-\sqrt{\frac{n}{4v_{1:T}}}}& \geq  \frac{\delta}{2}\exp\left\{ -\left(a-1\right)\sqrt{\frac{n}{4v_{1:T}}}\right\} \\
&=  \frac{\delta}{2}\exp\left\{ -\frac{4\delta}{1-2\delta}\sqrt{\frac{n}{4v_{1:T}}}\right\} ,
\end{align*}
where the inequality follows from the fact that $1+x\leq e^{x}$ for any $x$.
This lower bound on the Bayes risk has the slowest convergence rate
when $\delta$ is set to be proportional to $n^{-1/2}$. Specifically,
letting $\delta=\sqrt{v_{1:T}/n}$, we have 
\begin{align*}
\frac{\delta}{2}\exp\left\{ -\frac{4\delta}{1-2\delta}\sqrt{\frac{n}{4v_{1:T}}}\right\}  & =\frac{1}{2}\sqrt{\frac{v_{1:T}}{n}}\exp\left\{ -\frac{2}{1-2\delta}\right\} \geq\frac{1}{2}\sqrt{\frac{v_{1:T}}{n}}\exp\left(-4\right)\mbox{ if }1-2\delta\geq\frac{1}{2}.
\end{align*}
The condition $1-2\delta \geq 1/2$ is equivalent to $n\geq16v_{1:T}$.
Multiplying the lower bound by $M_{T}$ gives 
\begin{align*}
\sup_{P\in{\cal P}\left(M, \kappa, \MG\right)}E_{P^{n}}\left[W_{{\cal G}}^{\ast}-W\left(\hat{g}\right)\right] & \geq \frac{1}{2}\exp\left(-4\right)M_{T}\sqrt{\frac{v_{1:T}}{n}}
\end{align*}
for all $n\geq16v_{1:T}$. 

The proof is valid irrespective of whether Assumption \ref{asm:first-best} holds for a pair $(P,\MG)$ with any $P \in \MP(M, \kappa, \MG)$ or not.
\end{proof}
\medskip

\begin{proof}[Proof of Theorem \ref{thm:lower bound}]
The result immediately follows by setting 
\begin{align*}
    t = \underset{s\in \left\{1,\ldots,T\right\}}{\argmax}\gamma_{s}M_{s} \sqrt{\frac{v_{1:s}}{n}}
\end{align*}
in the statement of Lemma \ref{lem:lower bound}.
\end{proof}

\subsection{Proof of Theorem \ref{thm:estimated propensity score}}

We derive uniform upper bounds on the worst-case average welfare regrets of the two DEWM methods in the case where estimated propensity scores are used instead of true ones.

\medskip

\begin{proof}[Proof of Theorem \ref{thm:estimated propensity score} (i).]

Let $P \in \MP_{e} \cap \MP(M, \kappa, \MG)$ be fixed. Define $\hat{W}_{nt}(\text{\ensuremath{\underline{g}}}_{t})\equiv n^{-1}\sum_{i=1}^{n}\hat{w}_{t}^{S}(Z_{i},\text{\ensuremath{\underline{g}}}_{t})$ and $\hat{W}_{n}\left(g\right)\equiv\sum_{t=1}^{T}\hat{W}_{nt}(\text{\ensuremath{\underline{g}}}_{t})$, which are estimators of $W_{t}(\text{\ensuremath{\underline{g}}}_{t})$
and $W\left(g\right)$, respectively. It follows for any $g\in{\cal G}$
that 
\begin{align*}
W\left(g\right)-W\left(\hat{g}_{e}^{S}\right) & \leq  W_{n}\left(g\right)-\hat{W}_{n}\left(g\right)-W_{n}\left(\hat{g}_{e}^{S}\right)+\hat{W}_{n}\left(\hat{g}_{e}^{S}\right)\\
 & +W\left(g\right)-W\left(\hat{g}_{e}^{S}\right)+W_{n}\left(\hat{g}_{e}^{S}\right)-W_{n}\left(g\right)\\
& =  \frac{1}{n}\sum_{i=1}^{n}\sum_{t=1}^{T}\sum_{\text{\ensuremath{\underline{d}}}_{t}\in\left\{ 0,1\right\} ^{t}}\left[\left(\frac{\gamma_{t}Y_{it}\cdot1\left\{ \text{\ensuremath{\underline{D}}}_{it}=\text{\ensuremath{\underline{d}}}_{t}\right\} }{\prod_{s=1}^{t}e_{s}\left(d_{s},H_{is}\right)}-\frac{\gamma_{t}Y_{it}\cdot1\left\{ \text{\ensuremath{\underline{D}}}_{it}=\text{\ensuremath{\underline{d}}}_{t}\right\} }{\prod_{s=1}^{t}\hat{e}_{s}\left(d_{s},H_{is}\right)}\right)\right.\\
 & \left.\times\left(\prod_{s=1}^{t}1\left\{ g_{s}\left(H_{is}\right)=d_{s}\right\} -\prod_{s=1}^{t}1\left\{ \hat{g}_{e,s}^{S}\left(H_{is}\right)=d_{s}\right\} \right)\right]\\
 & +W\left(g\right)-W_{n}\left(g\right)+W_{n}\left(\hat{g}_{e}^{S}\right)-W\left(\hat{g}_{e}^{S}\right)\\
& \leq  \sum_{t=1}^{T}\sum_{\text{\ensuremath{\underline{d}}}_{t}\in\left\{ 0,1\right\} ^{t}}\left(\frac{1}{n}\sum_{i=1}^{n}\left|\tau_{t}\left(\underline{d}_{t},H_{it}\right)-\hat{\tau}_{t}\left(\underline{d}_{t},H_{it}\right)\right|\right)+2\sup_{g\in{\cal G}}\left|W_{n}\left(g\right)-W\left(g\right)\right|\\
& \leq  \sum_{t=1}^{T}\sum_{\text{\ensuremath{\underline{d}}}_{t}\in\left\{ 0,1\right\} ^{t}}\left(\frac{1}{n}\sum_{i=1}^{n}\left|\tau_{t}\left(\underline{d}_{t},H_{it}\right)-\hat{\tau}_{t}\left(\underline{d}_{t},H_{it}\right)\right|\right)\\
&+2\sum_{t=1}^{T}\sup_{ \underline{g}_{t}\in{\MG_{1}\times \cdots \times \MG_{t}}}\left|W_{nt}(\text{\ensuremath{\underline{g}}}_{t})-W_{t}(\text{\ensuremath{\underline{g}}}_{t})\right|.
\end{align*}
The first inequality follows from the fact that $\hat{g}_{e}^{S}$
maximizes $\hat{W}_{n}\left(\cdot\right)$ over ${\cal G}$. The second inequality follows from the fact that
\begin{align*}
 \left|\prod_{s=1}^{t}1\left\{ g_{s}\left(H_{is}\right)=d_{s}\right\} -\prod_{s=1}^{t}1\left\{ \hat{g}_{e,s}^{S}\left(H_{is}\right)=d_{s}\right\} \right|
 \leq1 .   
\end{align*}

Thus, the average welfare regret can be bounded from above by 
\begin{align*}
E_{P^{n}}\left[W_{{\cal G}}^{\ast}-W\left(\hat{g}_{e}^{S}\right)\right]&
\leq  \sum_{t=1}^{T}\sum_{\text{\ensuremath{\underline{d}}}_{t}\in\left\{ 0,1\right\} ^{t}}E_{P^{n}}\left[\frac{1}{n}\sum_{i=1}^{n}\left|\tau_{t}\left(\text{\ensuremath{\underline{d}}}_{t},H_{it}\right)-\hat{\tau}_{t}\left(\text{\ensuremath{\underline{d}}}_{t},H_{it}\right)\right|\right]\\
 & +2\sum_{t=1}^{T}\sup_{ \underline{g}_{t}\in{\MG_{1}\times \cdots \times \MG_{t}}}E_{P^{n}}\left[\left|W_{nt}(\text{\ensuremath{\underline{g}}}_{t})-W_{t}(\text{\ensuremath{\underline{g}}}_{t})\right|\right].
\end{align*}
Therefore, by the same argument as in the proof of Theorem \ref{thm:upper bound} (i) and from Assumption \ref{asm:estimated propensity score} (i), the
average welfare regret is bounded from above as 
\begin{align*}
E_{P^{n}}\left[W_{{\cal G}}^{\ast}-W\left(\hat{g}_{e}^{S}\right)\right] & \leq C\sum_{t=1}^{T}\left\{ \frac{\gamma_{t}M_{t}}{\prod_{s=1}^{t}\kappa_{s}}\sqrt{\frac{\sum_{s=1}^{t}v_{s}}{n}}\right\} +O\left(\phi_{n}^{-1}\right).
\end{align*}
Since this upper bound does not depend on $P\in{\cal P}_{e}\bigcap{\cal P}\left(M, \kappa, \MG\right)$,
the upper bound is uniform over ${\cal P}_{e}\bigcap{\cal P}\left(M, \kappa, \MG\right)$. 
\end{proof}
\medskip

Before proceeding to the proof of Theorem \ref{thm:estimated propensity score} (ii), we define 

\begin{align*}
 \Delta\tilde{Q}_{t,e}&\equiv\tilde{Q}_{t}\left(g_{t}^{\ast},\ldots,g_{T}^{\ast}\right)-\tilde{Q}_{t}\left(\hat{g}_{t,e}^{B},\ldots,\hat{g}_{T,e}^{B}\right), \\
 \Delta\tilde{Q}_{t,e}^{\dagger}&\equiv\tilde{Q}_{t}\left(g_{t}^{\ast},\hat{g}_{t+1,e}^{B},\ldots,\hat{g}_{T,e}^{B}\right)-\tilde{Q}_{t}\left(\hat{g}_{t,e}^{B},\ldots,\hat{g}_{T,e}^{B}\right),\\
 \check{Q}_{nt}\left(g_{t},\ldots,g_{T}\right) & \equiv E_{n}\left[\hat{q}_{t}\left(Z,g_{t},\ldots,g_{T}\right)\right]\\
 & =\sum_{s=t}^{T}E_{n}\left[ \frac{(\prod_{\ell =t}^{s}1\left\{ g_{\ell}\left(H_{\ell}\right)=D_{\ell}\right\}) \gamma_{s}Y_{s}}{\prod_{\ell =t}^{s}\hat{e}_{\ell}\left(D_{\ell},H_{\ell}\right)} \right], 
\end{align*}
where $\tilde{Q}$ is defined in (\ref{eq:Q-function with propensity scofre}).
The following lemmas will be used in the proof of Theorem \ref{thm:estimated propensity score} (ii).

\medskip

\begin{lemma} \label{lem:backward DEWM-2} Suppose that Assumptions
\ref{asm:sequential independence}, \ref{asm:bounded outcome}, and \ref{asm:overlap} hold for any $P \in \MP(M, \kappa, \MG)$, that Assumption \ref{asm:vc-class} holds for $\MG$,  that Assumption \ref{asm:first-best} holds for a pair $(P,\MG)$ with any $P \in \MP(M,\kappa,\MG)$, and that Assumption \ref{asm:estimated propensity score} (ii) holds for any $P \in \MP_{e}$. Then, for any $P \in \MP(M, \kappa, \MG) \cap \MP_{e}$, the following hold:
\begin{itemize}
\item[(i)] for $t=1,\ldots,T$, 
\begin{align*}
  E_{P^{n}}\left[\Delta\tilde{Q}_{t,e}^{\dagger}\right]
& \leq  C\left(\sum_{s=t}^{T}\frac{\gamma_{s}M_{s}}{\prod_{\ell =t}^{s}\kappa_{\ell}}\right)\sqrt{\frac{\sum_{s=t}^{T}v_{s}}{n}} + O\left(\xi_{n}^{-1}\right),
\end{align*}
where $C$ is the same constant term as introduced in Lemma \ref{lem:KT};
\item[(ii)] for $t=1,\ldots,T-1$ and $s=t+1,\ldots,T$, 
\begin{align*}
\tilde{Q}_{t}\left(g_{t}^{\ast},\ldots,g_{T}^{\ast}\right)-\tilde{Q}_{t}\left(g_{t}^{\ast},\ldots,g_{s}^{\ast},\hat{g}_{s+1,e}^{B},\ldots,\hat{g}_{T,e}^{B}\right) & \leq\frac{1}{\prod_{\ell =t}^{s}\kappa_{\ell}}\Delta\tilde{Q}_{s+1,e};
\end{align*}
\item[(iii)] 
\begin{align*}
\Delta\tilde{Q}_{1,e} & \leq\Delta \tilde{Q}_{1,e}^{\dagger}+\sum_{s=1}^{T-1}\frac{2^{s-1}}{\prod_{t=1}^{s}\kappa_{t}}\Delta \tilde{Q}_{s+1,e}^{\dagger}.
\end{align*}
\end{itemize}
\end{lemma}
\medskip

\begin{proof}

(i) It follows for any $\tilde{g}_{t}\in{\cal G}_{t}$ that 
\begin{align*}
 & \tilde{Q}_{t}\left(\tilde{g}_{t},\hat{g}_{t+1,e}^{B},\ldots,\hat{g}_{T,e}^{B}\right)-\tilde{Q}_{t}\left(\hat{g}_{t,e}^{B},\ldots,\hat{g}_{T,e}^{B}\right)\\
& \leq  \tilde{Q}_{nt}\left(\tilde{g}_{t},\hat{g}_{t+1,e}^{B},\ldots,\hat{g}_{T,e}^{B}\right)-\check{Q}_{nt}\left(\tilde{g}_{t},\hat{g}_{t+1,e}^{B},\ldots,\hat{g}_{T,e}^{B}\right) -\tilde{Q}_{nt}\left(\hat{g}_{t,e}^{B},\ldots,\hat{g}_{T,e}^{B}\right)+\check{Q}_{nt}\left(\hat{g}_{t,e}^{B},\ldots,\hat{g}_{T,e}^{B}\right)\\
 & +\tilde{Q}_{t}\left(\tilde{g}_{t},\hat{g}_{t+1,e}^{B},\ldots,\hat{g}_{T,e}^{B}\right)-\tilde{Q}_{t}\left(\hat{g}_{t,e}^{B},\ldots,\hat{g}_{T,e}^{B}\right)+\tilde{Q}_{nt}\left(\hat{g}_{t,e}^{B},\ldots,\hat{g}_{T,e}^{B}\right) - \tilde{Q}_{nt}\left(\tilde{g}_{t},\hat{g}_{t+1,e}^{B},\ldots,\hat{g}_{T,e}^{B}\right)\\
& =  \frac{1}{n}\sum_{i=1}^{n}\sum_{s=t}^{T}\sum_{\text{\ensuremath{\underline{d}}}_{t:s}\in\left\{ 0,1\right\} ^{t-s+1}}\left[\left(\frac{1\left\{ \underline{D}_{t:s}=\underline{d}_{t:s}\right\} \gamma_{s}Y_{s}}{\prod_{\ell =t}^{s}e_{\ell}\left(d_{\ell},H_{\ell}\right)}-\frac{1\left\{ \underline{D}_{t:s}=\underline{d}_{t:s}\right\} \gamma_{s}Y_{s}}{\prod_{\ell =t}^{s}\hat{e}_{\ell}\left(d_{\ell},H_{\ell}\right)}\right)\right.\\
 & \left.\times\left(1\left\{ \tilde{g}_{t}(H_{it})=d_{t}\right\} -1\left\{ \hat{g}_{t,e}^{B}\left(H_{it}\right)=d_{t}\right\} \right)\right]\\
 & +\tilde{Q}_{t}\left(\tilde{g}_{t},\hat{g}_{t+1,e}^{B},\ldots,\hat{g}_{T,e}^{B}\right)-\tilde{Q}_{t}\left(\hat{g}_{t,e}^{B},\ldots,\hat{g}_{T,e}^{B}\right) +\tilde{Q}_{nt}\left(\tilde{g}_{t},\hat{g}_{t+1,e}^{B},\ldots,\hat{g}_{T,e}^{B}\right)-\tilde{Q}_{nt}\left(\hat{g}_{t,e}^{B},\ldots,\hat{g}_{T,e}^{B}\right)\\
& \leq  \sum_{\text{\ensuremath{\underline{d}}}_{t:T}\in\left\{ 0,1\right\} ^{T-t+1}}\left(\frac{1}{n}\sum_{i=1}^{n}\left|\hat{\eta}_{t}^{-k}\left(\text{\ensuremath{\underline{d}}}_{t:T},H_{iT}\right)-\eta_{t}\left(\text{\ensuremath{\underline{d}}}_{t:T},H_{iT}\right)\right|\right)\\
 & +2\sup_{\left(g_{t},\ldots,g_{T}\right)\in{\cal G}_{t}\times\cdots\times{\cal G}_{T}}\left|\tilde{Q}_{nt}\left(g_{t},\ldots,g_{T}\right)-\tilde{Q}_{t}\left(g_{t},\ldots,g_{T}\right)\right|.
\end{align*}
The first inequality follows from the fact that $\hat{g}_{t,e}^{B}$ maximizes
$\check{Q}_{nt}\left(\cdot,\hat{g}_{t+1,e}^{B},\ldots,\hat{g}_{T,e}^{B}\right)$ over ${\cal G}_{t}$.

Then we have
\begin{align*}
    E_{P^n}\left[\Delta \tilde{Q}_{t,e}^{\dagger}\right] &\leq
    2E_{P^n}\left[\sup_{\left(g_{t},\ldots,g_{T}\right)\in{\cal G}_{t}\times\cdots\times{\cal G}_{T}}\left|\tilde{Q}_{nt}\left(g_{t},\ldots,g_{T}\right)-\tilde{Q}_{t}\left(g_{t},\ldots,g_{T}\right)\right|\right] \\
    & + \sum_{\text{\ensuremath{\underline{d}}}_{t:T}\in\left\{ 0,1\right\} ^{T-t+1}}E_{P^n}\left[\frac{1}{n}\sum_{i=1}^{n}\left|\hat{\eta}_{t}^{-k}\left(\text{\ensuremath{\underline{d}}}_{t:T},H_{iT}\right)-\eta_{t}\left(\text{\ensuremath{\underline{d}}}_{t:T},H_{iT}\right)\right|\right]
\end{align*}
Therefore, applying Lemma \ref{lem:KT} to the first term in the right hand side (as in the proof of Lemma \ref{lem:backward DEWM} (i)) and Assumption \ref{asm:estimated propensity score} (ii) to the second term in the right hand side leads to the result.

\medskip
\noindent (ii) The proof of Lemma \ref{lem:backward DEWM-2} follows from the same argument with the proof of Lemma \ref{lem:backward DEWM} (ii).

\medskip
\noindent (iii) We follow the same strategy as in Lemma \ref{lem:backward DEWM} (iii). First, note that
\begin{align*}
\Delta\tilde{Q}_{T,e}  =\tilde{Q}_{T}\left(g_{T}^{\ast}\right)-\tilde{Q}_{T}\left(\hat{g}_{T,e}^{B}\right)=\Delta \tilde{Q}_{T,e}^{\dagger}.
\end{align*}
Then, for $t=T-1$, we have
\begin{align*}
\Delta\tilde{Q}_{T-1,e} & =\tilde{Q}_{T-1}\left(g_{T-1}^{\ast},g_{T}^{\ast}\right)-\tilde{Q}_{T-1}\left(\hat{g}_{T-1,e}^{B},\hat{g}_{T,e}^{B}\right)\\
 & =\tilde{Q}_{T-1}\left(g_{T-1}^{\ast},g_{T}^{\ast}\right)-\tilde{Q}_{T-1}\left(g_{T-1}^{\ast},\hat{g}_{T,e}^{B}\right)+\tilde{Q}_{T-1}\left(g_{T-1}^{\ast},\hat{g}_{T,e}^{B}\right)-\tilde{Q}_{T-1}\left(\hat{g}_{T-1,e}^{B},\hat{g}_{T,e}^{B}\right)\\
 & \leq\frac{1}{\kappa_{T-1}}\Delta \tilde{Q}_{T,e}^{\dagger}+\Delta \tilde{Q}_{T-1,e}^{\dagger},
\end{align*}
where the inequality follows from Lemma \ref{lem:backward DEWM-2} (ii).

 Generally, for any $k=1,\ldots,T-1$, it follows that 
\begin{align*}
\Delta\tilde{Q}_{T-k,e} & =\tilde{Q}_{T-k}\left(g_{T-k}^{\ast},\ldots,g_{T}^{\ast}\right)-\tilde{Q}_{T-k}\left(\hat{g}_{T-k,e}^{B},\ldots,\hat{g}_{T,e}^{B}\right)\\
 & =\sum_{s=T-k}^{T}\left[\tilde{Q}_{T-k}\left(g_{T-k}^{\ast},\ldots,g_{s}^{\ast},\hat{g}_{s+1,e}^{B},\ldots,\hat{g}_{T,e}^{B}\right)-\tilde{Q}_{T-k}\left(g_{T-k}^{\ast},\ldots,g_{s-1}^{\ast},\hat{g}_{s,e}^{B},\ldots,\hat{g}_{T,e}^{B}\right)\right]\\
 & \leq\sum_{s=T-k}^{T}\left[\tilde{Q}_{T-k}\left(g_{T-k}^{\ast},\ldots,g_{T}^{\ast}\right)-\tilde{Q}_{T-k}\left(g_{T-k}^{\ast},\ldots,g_{s-1}^{\ast},\hat{g}_{s,e}^{B},\ldots,\hat{g}_{T,e}^{B}\right)\right]\\
 & =\sum_{s=T-k+1}^{T}\left[\tilde{Q}_{T-k}\left(g_{T-k}^{\ast},\ldots,g_{T}^{\ast}\right)-\tilde{Q}_{T-k}\left(g_{T-k}^{\ast},\ldots,g_{s-1}^{\ast},\hat{g}_{s,e}^{B},\ldots,\hat{g}_{T,e}^{B}\right)\right]+\Delta \tilde{Q}_{T-k,e}^{\dagger}\\
 & \leq\sum_{s=T-k+1}^{T}\frac{1}{\prod_{\ell=T-k}^{s-1}\kappa_{\ell}}\Delta\tilde{Q}_{s,e}+\Delta \tilde{Q}_{T-k,e}^{\dagger},
\end{align*}
where the second line follows by taking a telescope sum;
the third line follows from the fact that $\left(g_{s+1}^{\ast},\ldots,g_{T}^{\ast}\right)$
maximizes $\tilde{Q}_{T-k}\left(g_{T-k}^{\ast},\ldots,g_{s}^{\ast},\cdot,\ldots,\cdot\right)$
over ${\cal G}_{s+1}\times\cdots\times{\cal G}_{T}$ under Assumption
\ref{asm:first-best}; the last line follows from Lemma \ref{lem:backward DEWM-2}(ii). 

Then, recursively, the following hold:
\begin{align*}
\Delta\tilde{Q}_{T-1,e} & \leq\frac{1}{\kappa_{T-1}}\Delta\tilde{Q}_{T,e}+\Delta \tilde{Q}_{T-1,e}^{\dagger}=\frac{1}{\kappa_{T-1}}\Delta \tilde{Q}_{T,e}^{\dagger} + \Delta \tilde{Q}_{T-1,e}^{\dagger},\\
\Delta\tilde{Q}_{T-2,e} & \leq\frac{1}{\kappa_{T-2}}\Delta\tilde{Q}_{T-1,e}+\frac{1}{\kappa_{T-2}\kappa_{T-1}}\Delta\tilde{Q}_{T,e}+ \Delta \tilde{Q}_{T-2,e}^{\dagger}\\
 & \leq\frac{2}{\kappa_{T-2}\kappa_{T-1}}\Delta \tilde{Q}_{T,e}^{\dagger}+\frac{1}{\kappa_{T-2}}\Delta \tilde{Q}_{T-1,e}^{\dagger}+ \Delta \tilde{Q}_{T-2,e}^{\dagger},\\
 & \vdots\\
\Delta\tilde{Q}_{T-k,e} & \leq\sum_{s=1}^{k}\frac{2^{k-s}}{\prod_{t=T-k}^{T-s}\kappa_{t}}\Delta \tilde{Q}_{T-s+1,e}^{\dagger}+\Delta \tilde{Q}_{T-k,e}^{\dagger}.
\end{align*}
Therefore, when $k=T-1$, we have 
\begin{align*}
\Delta\tilde{Q}_{1,e} & \leq\Delta \tilde{Q}_{1,e}^{\dagger}+\sum_{s=1}^{T-1}\frac{2^{T-1-s}}{\prod_{t=1}^{T-s}\kappa_{t}}\Delta \tilde{Q}_{T-s+1,e}^{\dagger}\\
 & =\Delta \tilde{Q}_{1,e}^{\dagger}+\sum_{s=1}^{T-1}\frac{2^{s-1}}{\prod_{t=1}^{s}\kappa_{t}}\Delta \tilde{Q}_{s+1,e}^{\dagger}.
\end{align*}

\end{proof}
\medskip

\begin{proof}[Proof of Theorem \ref{thm:estimated propensity score} (ii).]

Let $P\in \widetilde{\MP}_{e}\bigcap{\cal P}\left(M, \kappa, \MG\right)$ be fixed. By the same argument as in the proof of Theorem \ref{thm:upper bound}
(ii), it follows for $g^* \in \argmax_{g \in \MG}W(g)$ that
\begin{align*}
    W_{\MG}^{\ast}-W(\hat{g}_{e}^{B}) &= \tilde{Q}_{1}(g^{\ast}) - \tilde{Q}_{1}(\hat{g}_{e}^{B}) \leq \Delta \tilde{Q}_{1,e} \\
    &\leq \Delta\tilde{Q}_{1,e}^{\dagger}+\sum_{s=1}^{T-1}\frac{2^{s-1}}{\prod_{t=1}^{s}\kappa_{t}}\Delta\tilde{Q}_{s+1,e}^{\dagger},
\end{align*}
where the second inequality follows from Lemma \ref{lem:backward DEWM-2} (iii). Thus, since $W_{\MG}^{\ast}-W(\hat{g}_{e}^{B}) \geq 0$, we have
\begin{align*}
    E_{P^{n}}\left[W_{{\cal G}}^{\ast}-W\left(\hat{g}_{e}^{B}\right)\right]\leq E_{P^n}\left[\Delta \tilde{Q}_{1,e}^{\dagger}\right]
    + \sum_{s=1}^{T-1}\frac{2^{s-1}}{\prod_{t=1}^{s}\kappa_{t}}E_{P^n}\left[\Delta\tilde{Q}_{s+1,e}^{\dagger}\right].
\end{align*}
Applying Lemma \ref{lem:backward DEWM-2} (i) to each term in the right hand side gives 
\begin{align*}
E_{P^{n}}\left[W_{{\cal G}}^{\ast}-W\left(\hat{g}_{e}^{B}\right)\right] & \leq  C\sum_{t=1}^{T}\left\{ \frac{\gamma_{t}M_{t}}{\prod_{s=1}^{t}\kappa_{s}}\sqrt{\frac{\sum_{s=1}^{t}v_{s}}{n}}\right\}  \\
&+\sum_{t=2}^{T}\frac{2^{t-2}}{\prod_{s=1}^{t-1}\kappa_{s}}\left(C\sum_{s=t}^{T}\left\{ \frac{\gamma_{s}M_{s}}{\prod_{\ell =t}^{s}\kappa_{\ell}}\sqrt{\frac{\sum_{\ell =t}^{s}v_{\ell}}{n}}\right\} \right) + O\left(\xi_{n}^{-1}\right).
\end{align*}
Since this upper bound does not depend on $P\in \widetilde{\MP}_{e}\bigcap{\cal P}\left(M, \kappa, \MG\right)$, 
the upper bound is uniform over $\widetilde{\MP}_{e}\bigcap{\cal P}\left(M, \kappa, \MG\right)$. 
\end{proof}
\medskip

\subsection{Proof of Theorem \ref{thm:equality_mindied_welfare}}\label{sec:proof_equality_minded_welfare}

This section provides the proof of Theorem \ref{thm:equality_mindied_welfare} for the rank-dependent SWF (\ref{eq:rank-dependent_SWF}). The following is a preliminary lemma.

\medskip

\begin{lemma} \label{lem:tail_concentration_inequality}(Lemma A.5 in \cite{Kitagawa_Tetenov_2021})
Let $\MF$ be a class of uniformly bounded functions, that is, there exists $\bar{F} < \infty$ such that $||f||_{\infty} \leq \bar{F}$ for all $f \in \MF$. Assume that $\MF$ is a VC-subgraph class with VC-dimension $v < \infty$. Let $(Y,Z)\sim P$, where $Y\geq 0$ is a scalar ($Y$ and $Z$ may be dependent). Let $\{(Y_i,Z_i)\}_{i=1}^{n} \sim P^n$ be an iid sample from $P$. Assume that
\begin{align*}
    \int_{0}^{\infty} \sqrt{P(Y>y)}dy \leq M.
\end{align*}
Then, there is a universal constant $C$ such that 
\begin{align*}
    \int_{0}^{\infty} E_{P^n}\left[ \sup_{f \in \MF} \left| \frac{1}{n}\sum_{i=1}^{n}f(Z_i)1\{Y_i > y\} -  E_{P}[f(Z)1\{Y>y\}] \right|\right] dy \leq C \bar{F}M\sqrt{\frac{v}{n}}.
\end{align*}
holds for all $n \geq 1$.
\end{lemma}

\medskip

\begin{proof}[Proof of Theorem \ref{thm:equality_mindied_welfare}.]
By the same argument as in the proof of Theorem \ref{thm:upper bound} (i), it follows that 
\begin{align}
    E_{P^n}\left[\sup_{g \in \MG} W_{\Lambda}(g) - W_{\Lambda}(\hat{g}^{S}) \right] \leq 2 E_{P^n}\left[\sup_{g \in \MG} \left| \widehat{W}_{\Lambda}(g) - W_{\Lambda}(g)\right|\right]. \label{eq:lambda_welfare_bound_1}
\end{align}
Since $\Lambda(\cdot)$ is convex and nonincreasing,
\begin{align}
    \sup_{g \in \MG} \left|\widehat{W}_{\Lambda}(G) - W_{\Lambda}(g)\right| &= \sup_{g \in \MG} \left|\int_{0}^{\infty} \Lambda\left(\widehat{F}_{g}(y) \vee 0 \right)dy - \int_{0}^{\infty} \Lambda\left(F_{g}(y) \right)dy\right| \nonumber \\
    & \leq \sup_{g \in \MG} \int_{0}^{\infty} \left|\Lambda\left(\widehat{F}_{g}(y) \vee 0 \right) -  \Lambda\left(F_{g}(y) \right)\right| dy \nonumber \\
    & \leq  \int_{0}^{\infty} \sup_{g \in \MG} \left|\Lambda\left(\widehat{F}_{g}(y) \vee 0 \right) -  \Lambda\left(F_{g}(y) \right)\right| dy \nonumber \\
    & \leq  \left|\Lambda^{\prime}(0)\right|\int_{0}^{\infty} \sup_{g \in \MG} \left|\widehat{F}_{g}(y) - F_{g}(y)\right| dy. \label{eq:lambda_welfare_bound_2} 
\end{align}
Combining (\ref{eq:lambda_welfare_bound_1}) and (\ref{eq:lambda_welfare_bound_2}) yields
\begin{align*}
     E_{P^n}\left[\sup_{g \in \MG} W_{\Lambda}(g) - W_{\Lambda}(\hat{g}^{S}) \right] \leq  2\left|\Lambda^{\prime}(0)\right|\int_{0}^{\infty} E_{P^n}\left[\sup_{g \in \MG} \left|\widehat{F}_{g}(y) - F_{g}(y)\right| \right] dy.
\end{align*}
For $g \in \MG$, let
\begin{align*}
    w_{g}(Z_i) \equiv \frac{\prod_{t=1}^{T}1\{g_{t}(H_{it})=D_{it}\}}{\prod_{t=1}^{T}e_{t}\left(D_{it},H_{it}\right)}.
\end{align*}
By Lemma \ref{lem:vc_subclass}, the class of functions $\mathcal{W}=\{w_{g}(\cdot): g \in \MG\}$ is a VC-subclass with VC-dimension of at most $\sum_{t=1}^{T}v_{t}$. Assumption \ref{asm:overlap} implies that $w_{g}(Z_i) \in \left[0,1/\left(\prod_{t=1}^{T}\kappa_{t}\right)\right]$; hence, functions in $\mathcal{W}$ are uniformly bounded by $1/\left(\prod_{t=1}^{T}\kappa_{t}\right)$.

 Since $F_{g}(y) = 1 - E_{P}\left[w_{g}(Z)\cdot 1\left\{\sum_{t=1}^{T}Y_t > y\right\}\right]$, under Assumption \ref{asm:sequential independence}, and $\widehat{F}_{g}(y) = 1 - \frac{1}{n}\sum_{i=1}^{n}w_{g}(Z_i)\cdot 1\{\sum_{t=1}^{T}Y_t > y\}$,
 \begin{align*}
     \left|\widehat{F}_{g}(y) - F_{g}(y)\right| = \left| \frac{1}{n}\sum_{i=1}^{n}w_{g}(Z_i)\cdot 1\left\{\sum_{t=1}^{T}Y_t > y \right\} -  E_{P}\left[w_{g}(Z)\cdot 1\left\{\sum_{t=1}^{T}Y_t > y\right\}\right] \right|.
 \end{align*}
It follows from Assumption \ref{eq:tail_bound} that 
\begin{align*}
    \bigintsss_{0}^{\infty} \sqrt{P\left(\sum_{t=1}^{T}Y_t > y\right)}dy \leq \max_{\underline{d}_{T} \in \{0,1\}^T} \bigintsss_{0}^{\infty} \sqrt{P\left(\sum_{t=1}^{T}Y_t\left(\underline{d}_{t}\right) > y\right)}dy \leq \Upsilon.
\end{align*}
Applying Lemma \ref{lem:tail_concentration_inequality} to (\ref{eq:lambda_welfare_bound_2}) yields
\begin{align*}
    E_{P^n}\left[\sup_{g \in \MG} W_{\Lambda}(g) - W_{\Lambda}(\hat{g}^{S}) \right] \leq 2 C\left|\Lambda^{\prime}(0) \right| \frac{\Upsilon}{\prod_{t=1}^{T}\kappa_t} \sqrt{\frac{\sum_{t=1}^{T}v_t}{n}}.
\end{align*}
\end{proof}


\subsection{Proofs of Theorems \ref{thm:main_theorem_AIPW} and \ref{thm:main_theorem_simultaneous}}\label{sec:proof_DR_estimator}

The following lemma, which directly follows from Lemma 2 in \cite{Zhou_et_al_2023} and its proof, plays important roles in the proofs of Theorems \ref{thm:main_theorem_AIPW} and \ref{thm:main_theorem_simultaneous}.

\medskip

\begin{lemma} \label{lem:concentration inequality_influence difference function}
Fix $t \in \{ 1,\ldots,T\}$. For any $\underline{d}_{t}\in \{0,1\}^t$, let $\{\Gamma_{i}(\underline{d}_{t})\}_{i=1}^{n}$ be i.i.d. random variables with bounded supports. For any $\underline{g}_{t:T} \in \MG_{t:T}$, let $\widetilde{Q}(\underline{g}_{t:T})  \equiv  \frac{1}{n} \sum_{i=1}^{n}  \Gamma_{i}^{\dag}(\underline{g}_{t:T})$, where $\Gamma_{i}^{\dag}(\underline{g}_{t:T}) \equiv \Gamma_{i}^{\dag}((g_{t}(H_{it}),\ldots,g_{T}(H_{iT})))$, and $Q(\underline{g}_{t:T}) \equiv  E[\widetilde{Q}(\underline{g}_{t:T})]$. For any $\underline{g}_{t:T}^{a},\underline{g}_{t:T}^{b} \in \MG_{t:T}$, denote $\widetilde{\Delta}(\underline{g}_{t:T}^{a},\underline{g}_{t:T}^{b}) \equiv \widetilde{Q}(\underline{g}_{t:T}^{a}) - \widetilde{Q}(\underline{g}_{t:T}^{b})$ and $\Delta(\underline{g}_{t:T}^{a},\underline{g}_{t:T}^{b}) \equiv Q(\underline{g}_{t:T}^{a}) - Q(\underline{g}_{t:T}^{b})$. Let $\kappa(\cdot)$ denote the entorpy integral defined in \cite{Zhou_et_al_2023}. Then, under Assumption \ref{asm:vc-class}, the following holds: For any $\delta \in (0,1)$, with probability at least $1-2\delta$,
\begin{align*}
    \sup_{\underline{g}_{t:T}^{a},\underline{g}_{t:T}^{b} \in \MG_{t:T}} \left|\widetilde{\Delta}(\underline{g}_{t:T}^{a},\underline{g}_{t:T}^{b})-\Delta(\underline{g}_{t:T}^{a},\underline{g}_{t:T}^{b})\right| &\leq \left(54.4 \sqrt{2}\kappa(\underline{g}_{t:T}) + 435.2 + \sqrt{2 \log \frac{1}{\delta}}\right)\sqrt{\frac{V_{t:T}^{\ast}}{n}} \\
    &+ o\left(\frac{1}{\sqrt{n}}\right),
\end{align*}
where $V_{t:T}^{\ast}  \equiv  \sup_{\underline{g}_{t:T}^{a},\underline{g}_{t:T}^{b} \in \MG_{t:T}}E
    \left[\left(\Gamma_{i}^{\dag}(\underline{g}_{t:T}^{a}) - \Gamma_{i}^{\dag}(\underline{g}_{t:T}^{b}) \right)^2\right].$
\end{lemma}

\medskip

\begin{proof}[Proof of Theorem \ref{thm:main_theorem_AIPW}]
Let us define
\begin{align*}
 \widetilde{V}_{i}(g)=&\sum_{t=1}^{T}\psi_{it}\left(\underline{g}_{t}\right)Y_{it}-\sum_{t=1}^{T}\psi_{it}\left(g\right)Q_{t}^{\underline{g}_{(t+1):T}}\left(H_{it},D_{i,}\right)  \nonumber\\
    & + \sum_{t=1}^{T}\psi_{i,t-1}\left(\underline{g}_{t-1}\right)Q_{t}^{\underline{g}_{(t+1):T}}\left(H_{it},g_{t}(H_{it})\right) , \\
     \widehat{V}_{i}(g)=&\sum_{t=1}^{T}\hat{\psi}_{it}^{-k(i)}\left(\underline{g}_{t}\right)Y_{it}-\sum_{t=1}^{T}\hat{\psi}_{it}^{-k(i)}\left(g\right)\widehat{Q}_{t}^{\underline{g}_{(t+1):T},-k(i)}\left(H_{it},D_{i,}\right)  \nonumber\\
    & + \sum_{t=1}^{T}\hat{\psi}_{i,t-1}^{-k(i)}\left(\underline{g}_{t-1}\right)\widehat{Q}_{t}^{\underline{g}_{(t+1):T},-k(i)}\left(H_{it},g_{t}(H_{it})\right) , 
\end{align*}
with $$\psi_{it}\left(\underline{g}_{t}\right) \equiv \frac{\prod_{s=1}^{t}1\{D_{is}=g_{s}(H_{is})\}}{\prod_{s=1}^{t}e_{t}(H_{is},g_s(H_{is}))}.$$
We also define $V(g) = E\left[\widetilde{V}_{i}(g)\right]$. Lemma B.2 in \cite{Sakaguchi_2024} shows that $V(g) = W(g)$ under Assumptions \ref{asm:sequential independence}.

For any $g^{a},g^{b} \in \MG$, we define
\begin{align*}
    \Delta(g^{a},g^{b}) &\equiv V(g^{a}) - V(g^{b}),\\
    \widetilde{\Delta}(g^{a},g^{b}) &\equiv \frac{1}{n}\sum_{i=1}^{n}\widetilde{V}_{i}(g^{a}) - \frac{1}{n}\sum_{i=1}^{n}\widetilde{V}_{i}(g^{b}),\\
    \widehat{\Delta}(g^{a},g^{b}) &\equiv \frac{1}{n}\sum_{i=1}^{n}\widehat{V}_{i}(g^{a}) - \frac{1}{n}\sum_{i=1}^{n}\widehat{V}_{i}(g^{b}).
\end{align*}

Let $g_{opt}^{\ast} \in \argmax_{g \in \MG} W(g)$. A standard argument of the statistical learning theory gives
\begin{align}
    W_{\MG}^{\ast} - W\left(\hat{g}^{AIPW}\right) &= \Delta(g_{opt}^{\ast},\hat{g}^{AIPW}) \nonumber \\
    & \leq \Delta(g_{opt}^{\ast},\hat{g}^{AIPW}) - \widehat{\Delta}(g_{opt}^{\ast},\hat{g}^{AIPW}) \nonumber\\
    & \leq \sup_{g^{a},g^{b}\in \MG} \left|\Delta(g^{a},g^{b}) - \widehat{\Delta}(g^{a},g^{b})\right| \nonumber \\
    & \leq \sup_{g^{a},g^{b}\in \MG} \left|\Delta(g^{a},g^{b}) - \widetilde{\Delta}(g^{a},g^{b})\right| 
    + \sup_{g^{a},g^{b}\in \MG} \left|\widehat{\Delta}(g^{a},g^{b}) - \widetilde{\Delta}(g^{a},g^{b})\right|, \label{eq:former_latter_decomposition}
\end{align}
where the first inequality follows because $\hat{g}^{AIPW}$ maximizes $\widehat{V}(g)$ over $\MG$; hence, $\widehat{\Delta}(g^{\ast},\hat{g}^{AIPW}) \leq 0$.

We can now evaluate $ W_{\MG}^{\ast} - W(\hat{g}^{AIPW})$ through evaluating $\sup_{g^{a},g^{b}\in \MG} \left|\Delta(g^{a},g^{b}) - \widetilde{\Delta}(g^{a},g^{b})\right|$ and $\sup_{g^{a},g^{b}\in \MG} \left|\widehat{\Delta}(g^{a},g^{b}) - \widetilde{\Delta}(g^{a},g^{b})\right|$. As for the former, under Assumptions \ref{asm:sequential independence}--\ref{asm:overlap}, we can apply Lemma \ref{lem:concentration inequality_influence difference function} to obtain the following result: For any stage $t$ and $\delta \in (0,1)$, with probability at least $1-2\delta$, 
\begin{align}
    \sup_{g^{a},g^{b}\in \MG} \left|\Delta(g^{a},g^{b}) - \widetilde{\Delta}(g^{a},g^{b})\right| &\leq \left(54.4 \sqrt{2}\kappa(\MG) + 435.2 + \sqrt{2 \log \frac{1}{\delta}}\right)\sqrt{\frac{V^{\ast}}{n}} + o\left(\frac{1}{\sqrt{n}}\right), \label{eq:bound_influence_difference_function_v2}
\end{align}
where 
    $V^{\ast}  \equiv  \sup_{g^{a},g^{b} \in \MG}E\left[\left(\widetilde{V}_{i}(g^{a}) - \widetilde{V}_{i}(g^{b}) \right)^2\right] < \infty.$
Note that $\kappa(\MG) \leq 2.5 \sqrt{VC(\MG)} < \infty$ from Remark 8 in \cite{Zhou_et_al_2023} and Assumption \ref{asm:vc-class}. 

As for the latter $\sup_{g^{a},g^{b}\in \MG} \left|\widehat{\Delta}(g^{a},g^{b}) - \widetilde{\Delta}(g^{a},g^{b})\right|$, under Assumptions \ref{asm:sequential independence}--\ref{asm:overlap} and \ref{asm:rate_of_convergence_AIPW}, we can obtain from Lemma A.3 in \cite{Sakaguchi_2024} that
\begin{align}
    \sup_{g^{a},g^{b}\in \MG} \left|\widehat{\Delta}(g^{a},g^{b}) - \widetilde{\Delta}(g^{a},g^{b})\right| = O_{p}\left(n^{-\min\{1/2,\tau\}}\right). \label{eq:latter_convergence}
\end{align}
Combining (\ref{eq:former_latter_decomposition}), (\ref{eq:bound_influence_difference_function_v2}), and (\ref{eq:latter_convergence}) leads to the result.
\end{proof}

\medskip

\begin{proof}[Proof of Lemma \ref{prop:mu_idntification}]

Under Assumption \ref{asm:exogenous_covariates} and the redefinition $H_t(\underline{d}_t) = (\underline{D}_{t-1},\underline{X}_t)$,
    \begin{align*}
        \widetilde{Y}_t\left(\underline{g}_t \right) &= \sum_{\underline{d}_{t} \in \{0,1\}^t} Y_t(\underline{d}_t)\cdot \prod_{s=1}^{t}1\left\{g_s\left(\underline{d}_{s-1},\underline{X}_s\right) = d_s\right\},
    \end{align*}
    where we use the fact that $\underline{X}_t = \underline{X}_t(\underline{d}_t)$ for any $\underline{d}_t$. Then we have 
    \begin{align*}
        W_t\left(\underline{g}_t\right) = \sum_{\underline{d}_{t} \in \{0,1\}^t}E_{P}\left[ \gamma_{t}Y_t(\underline{d}_t)\cdot \prod_{s=1}^{t}1\left\{g_s\left(\underline{d}_{s-1},\underline{X}_s\right) = d_s\right\}\right].
    \end{align*}

    It follows that 
    \begin{align*}
        &\sum_{\underline{d}_{t} \in \{0,1\}^t}E_{P}\left[\mu_t(\underline{d}_{t},\underline{X}_t)\cdot \prod_{s=1}^{t}1\{d_s = g_s(\underline{d}_{s-1},\underline{X}_s)\}\right] \\
        &= \sum_{\underline{d}_{t} \in \{0,1\}^t}E_{P}\left[E_{P}\left[\gamma_t Y_t(\underline{d}_{t})\cdot \prod_{s=1}^{t}1\{d_s = g_s(\underline{d}_{s-1},\underline{X}_s)\} \middle|\underline{D}_t = \underline{d}_t,\underline{X}_t \right]\right]\\
        &= \sum_{\underline{d}_{t} \in \{0,1\}^t}E_{P}\left[E_{P}\left[\gamma_t Y_t(\underline{d}_{t})\cdot \prod_{s=1}^{t}1\{d_s = g_s(\underline{d}_{s-1},\underline{X}_s)\} \middle|\underline{D}_{t-1} = \underline{d}_{t-1},\underline{X}_t \right]\right]\\
        &= \sum_{\underline{d}_{t} \in \{0,1\}^t}E_{P}\left[E_{P}\left[\gamma_t Y_t(\underline{d}_{t})\cdot \prod_{s=1}^{t}1\{d_s = g_s(\underline{d}_{s-1},\underline{X}_s)\} \middle|\underline{D}_{t-1} = \underline{d}_{t-1},\underline{X}_{t-1} \right]\right].
    \end{align*}
    where the second equality follows from Assumption \ref{asm:sequential independence} and the third follows by the law of iterated expectations. Applying the same argument recursively,  
   \begin{align*}
        &\sum_{\underline{d}_{t} \in \{0,1\}^t}E_{P}\left[\mu_t(\underline{d}_{t},\underline{X}_t)\cdot \prod_{s=1}^{t}1\{d_s = g_s(\underline{d}_{s-1},\underline{X}_s)\}\right] \\
        &= \sum_{\underline{d}_{t} \in \{0,1\}^t}E_{P}\left[\gamma_t Y_t(\underline{d}_{t})\cdot \prod_{s=1}^{t}1\{d_s = g_s(\underline{d}_{s-1},\underline{X}_s)\} \right]  = W_t\left(\underline{g}_t\right).
    \end{align*}
\end{proof}
\medskip

\begin{proof}[Proof of Theorem \ref{thm:main_theorem_simultaneous}]
    For any $g \in \MG$, let $\tilde{g} = (\tilde{g}_1,\ldots,\tilde{g}_T)$, where $\tilde{g}_t:\underline{\MX}_t \rightarrow \{0,1\}$, be recursively defined as 
\begin{align*}
    \tilde{g}_1(x_1) &= g_1(x_1),\\
    \tilde{g}_2(\underline{x}_2) &= g_2(\tilde{g}_1(x_1),\underline{x}_2),\\
    \tilde{g}_3(\underline{x}_3) &= g_2(\tilde{g}_1(x_1),\tilde{g}_2(\underline{x}_2),\underline{x}_3),\\
    & \vdots\\
    \tilde{g}_t(\underline{x}_t) &= g_t(\tilde{g}_1(x_1),\tilde{g}_2(\underline{x}_2),\ldots,\tilde{g}_{t-1}(\underline{x}_{t-1}),\underline{x}_t),\\
    & \vdots\\
    \tilde{g}_T(\underline{x}_T) &= g_T(\tilde{g}_1(x_1),\tilde{g}_2(\underline{x}_2),\ldots,\tilde{g}_{T-1}(\underline{x}_{T-1}),\underline{x}_T).
\end{align*}
For the class $\MG$ of DTRs, we define the class of $\tilde{g}$ as 
\begin{align*}
    \widetilde{\MG} \equiv \left\{\tilde{g}=(\tilde{g}_{1},\ldots,\tilde{g}_{T}): g \in \MG\right\}.
\end{align*}
Let $\widetilde{\MG}_{1:t} := \widetilde{\MG}_{1} \times \cdots \times \widetilde{\MG}_{t}$. We denote by $\kappa(\cdot)$ the entropy integral defined in \cite{Zhou_et_al_2023}. Note that $\kappa(\widetilde{\MG}) \leq \kappa(\MG)$.

Without loss of generality, we suppose that $\gamma_{1}=\gamma_2 = \cdots = \gamma_{T}$. Let us define, for $t=1,\ldots,T$,
\begin{align*}
{\Gamma}_{it}(\underline{d}_{t}) &\equiv \frac{ Y_{it} - {\mu}_{t}(\underline{d}_{t},\underline{X}_{it}) }{\eta_t(\underline{d}_{t},\underline{X}_{it})}\cdot 1\{\underline{D}_{it}=\underline{d}_t\} + {\mu}_{t}(\underline{d}_{t},\underline{X}_{it}),\\
\widehat{\Gamma}_{it}(\underline{d}_{t}) &\equiv \frac{ Y_{it} - \hat{\mu}_{t}^{-k(i)}(\underline{d}_{t},\underline{X}_{it}) }{\hat{\eta}_{t}^{-k(i)}(\underline{d}_{t},\underline{X}_{it})}\cdot 1\{\underline{D}_{it}=\underline{d}_t\} + \hat{\mu}_{t}^{-k(i)}(\underline{d}_{t},\underline{X}_{it}),
\end{align*}
where $\eta_t(\underline{d}_t,\underline{x}_t) \equiv \prod_{s=1}^{t}e_{s}(d_{s},H_{is})$ and $\hat{\eta}_{t}^{-k(i)}(\underline{d}_t,\underline{x}_t) \equiv \prod_{s=1}^{t}\hat{e}_{s}^{-k(i)}(d_{s},H_{is})$. 

Let $\underline{\tilde{g}}_t(\underline{x}_t) = (\tilde{g}_{1}(x_1),\tilde{g}_{2}(\underline{x}_2),\ldots,\tilde{g}_{t}(\underline{x}_t))$. Given a fixed DTR $\tilde{g} \in \widetilde{\MG}$ constructed from $g \in \MG$, with some abuse of the notation, we define for $t=1,\ldots,T$,
\begin{align*}
\Gamma_{it}(\underline{\tilde{g}}_{t}) &\equiv \sum_{\underline{d}_t \in \{0,1\}^t} {\Gamma}_{it}(\underline{d}_{t}) \cdot 1\{\underline{\tilde{g}}_t(\underline{x}_t)=\underline{d}_t\},\\
\widehat{\Gamma}_{it}(\underline{\tilde{g}}_{t}) &\equiv \sum_{\underline{d}_t \in \{0,1\}^t} \widehat{\Gamma}_{it}(\underline{d}_{t}) \cdot 1\{\underline{\tilde{g}}_t(\underline{x}_t)=\underline{d}_t\}.
\end{align*}
Note that $(1/n)\sum_{t=1}^{t}\sum_{i=1}^{n}\widehat{\Gamma}_{it}(\underline{\tilde{g}}_{t}) = \widehat{W}^{DR}(g)$. Hence,  $$\max_{\tilde{g} \in \widetilde{\MG}} \frac{1}{n}\sum_{t=1}^{t}\sum_{i=1}^{n}\widehat{\Gamma}_{it}(\underline{\tilde{g}}_{t}) = \max_{g\in \MG} \widehat{W}^{DR}(g).$$
Note also that $(1/n)\sum_{i=1}^{n}\Gamma_{it}\left(\underline{\tilde{g}}_t\right)$ is an oracle estimate of $W_t(\underline{g}_t)$ with oracle access to $\{\mu_s(\cdot,\cdot):s=1,\ldots,t\}$ and $\{e_{s}(\cdot,\cdot):s=1,\ldots,t\}$.
For $\underline{g}_t \in \MG_{1:t}$, we define $\widetilde{W}_t(\underline{\tilde{g}}_t) \equiv E\left[\Gamma_{it}(\underline{\tilde{g}}_t)\right]$. Note that $W_t(\underline{g}_t) = \widetilde{W}_t(\underline{\tilde{g}}_t)$ by Lemma \ref{prop:mu_idntification} under Assumptions \ref{asm:sequential independence} and \ref{asm:exogenous_covariates}.

Following the analysis of \cite{Zhou_et_al_2023}, we define the policy value difference function $\Delta_{t}(\cdot;\cdot):\widetilde{\MG}_{1:t} \times \widetilde{\MG}_{1:t} \rightarrow \Real$, the oracle influence difference function $\widetilde{\Delta}_{t}(\cdot;\cdot):\widetilde{\MG}_{1:t} \times \widetilde{\MG}_{1:t} \rightarrow \Real$, and the estimated policy value difference function $\widehat{\Delta}_{t}(\cdot;\cdot):\widetilde{\MG}_{1:t} \times \widetilde{\MG}_{1:t} \rightarrow \Real$, respectively, as follows: For $\underline{\tilde{g}}_{t}^{a}=(\tilde{g}_{1}^{a},\ldots,\tilde{g}_{t}^{a}) \in \widetilde{\MG}_{1:t}$ and $\underline{\tilde{g}}_{t}^{b}=(\tilde{g}_{1}^{b},\ldots,\tilde{g}_{t}^{b}) \in \widetilde{\MG}_{1:t}$,
\begin{align*}
&\Delta_{t}(\underline{\tilde{g}}_{t}^{a};\underline{\tilde{g}}_{t}^{b})\equiv \widetilde{W}_{t}(\underline{\tilde{g}}_{t}^{a}) - \widetilde{W}_{t}(\underline{\tilde{g}}_{t}^{b}), \\ 
 &\widetilde{\Delta}_{t}(\underline{\tilde{g}}_{t}^{a};\underline{\tilde{g}}_{t}^{b}) \equiv \frac{1}{n}\sum_{i=1}^{n} \Gamma_{it}\left(\underline{\tilde{g}}_{t}^{a}\right) - \frac{1}{n}\sum_{i=1}^{n} \Gamma_{it}\left(\underline{\tilde{g}}_{t}^{b}\right), \\
 & \widehat{\Delta}_{t}(\underline{\tilde{g}}_{t}^{a};\underline{\tilde{g}}_{t}^{b}) \equiv  \frac{1}{n}\sum_{i=1}^{n} \widehat{\Gamma}_{it}\left(\underline{\tilde{g}}_{t}^{a}\right) -  \frac{1}{n}\sum_{i=1}^{n} \widehat{\Gamma}_{it}\left(\underline{\tilde{g}}_{t}^{b}\right). \nonumber
\end{align*}
Note that $\widetilde{\Delta}_{t}(\underline{\tilde{g}}_{t}^{a};\underline{\tilde{g}}_{t}^{b})$ is an unbiased estimator of the policy value difference function $\Delta_{t}(\underline{\tilde{g}}_{t}^{a};\underline{\tilde{g}}_{t}^{b})$. From the definitions, 
\begin{align*}
    W_{\MG}^{\ast} - W(\hat{g}^{DR}) = \sum_{t=1}^{T}\Delta_{t}\left(\underline{\tilde{g}}_{1:t}^{\ast};\tilde{\hat{\underline{g}}}_{1:t}^{DR}\right).
\end{align*}

A standard argument of the statistical learning theory gives
\begin{align}
    W_{\MG}^{\ast} - W(\hat{g}^{DR})
&=\sum_{t=1}^{T}\Delta_{t}\left(\underline{\tilde{g}}_{1:t}^{\ast};\tilde{\hat{\underline{g}}}_{1:t}^{DR}\right) \nonumber \\
    &\leq \sum_{t=1}^{T}\Delta_{t}\left(\underline{\tilde{g}}_{1:t}^{\ast};\tilde{\hat{\underline{g}}}_{1:t}^{DR}\right)- \sum_{t=1}^{T}\widehat{\Delta}_{t}\left(\underline{\tilde{g}}_{1:t}^{\ast};\tilde{\hat{\underline{g}}}_{1:t}^{DR}\right) \nonumber \\
    &\leq \sum_{t=1}^{T}\sup_{\underline{\tilde{g}}_{t}^{a},\underline{\tilde{g}}_{t}^{b} \in \widetilde{\MG}_{1:t}} |\Delta_{t}(\underline{\tilde{g}}_{t}^{a};\underline{\tilde{g}}_{t}^{b}) - \widehat{\Delta}_{t}(\underline{\tilde{g}}_{t}^{a};\underline{\tilde{g}}_{t}^{b})| \nonumber \\
    & \leq  \sum_{t=1}^{T}\sup_{\underline{\tilde{g}}_{t}^{a},\underline{\tilde{g}}_{t}^{b} \in \widetilde{\MG}_{1:t}} |\Delta_{t}(\underline{\tilde{g}}_{t}^{a};\underline{\tilde{g}}_{t}^{b}) - \widetilde{\Delta}_{t}(\underline{\tilde{g}}_{t}^{a};\underline{\tilde{g}}_{t}^{b})| \nonumber \\ 
    &+ \sum_{t=1}^{T} \sup_{\underline{\tilde{g}}_{t}^{a},\underline{\tilde{g}}_{t}^{b} \in \widetilde{\MG}_{1:t}} |\widehat{\Delta}_{t}(\underline{\tilde{g}}_{t}^{a};\underline{\tilde{g}}_{t}^{b}) - \widetilde{\Delta}_{t}(\underline{\tilde{g}}_{t}^{a};\underline{\tilde{g}}_{t}^{b})|, \label{eq:standard_inequality}
\end{align}
where the first inequality follows because $\tilde{\hat{g}}^{DR}$ maximizes  $(1/n)\sum_{t=1}^{T}\sum_{i=1}^{n} \widehat{\Gamma}_{it}\left(\underline{\tilde{g}}_{1:t}\right)$ over $\widetilde{\MG}$; hence, $\sum_{t=1}^{T}\widehat{\Delta}_{t}\left(\underline{\tilde{g}}_{1:t}^{\ast};\tilde{\hat{\underline{g}}}_{1:t}^{DR}\right) \leq 0$. 

We can now evaluate $ W_{\MG}^{\ast} - W(\hat{g}^{DR})$ through evaluating $\sup_{\underline{\tilde{g}}_{t}^{a},\underline{\tilde{g}}_{t}^{b} \in \widetilde{\MG}_{1:t}} |\Delta_{t}(\underline{\tilde{g}}_{t}^{a};\underline{\tilde{g}}_{t}^{b}) - \widetilde{\Delta}_{t}(\underline{\tilde{g}}_{t}^{a};\underline{\tilde{g}}_{t}^{b})|$ and $\sup_{\underline{\tilde{g}}_{t}^{a},\underline{\tilde{g}}_{t}^{b} \in \widetilde{\MG}_{1:t}} |\widehat{\Delta}_{t}(\underline{\tilde{g}}_{t}^{a};\underline{\tilde{g}}_{t}^{b}) - \widetilde{\Delta}_{t}(\underline{\tilde{g}}_{t}^{a};\underline{\tilde{g}}_{t}^{b})|$ for each $t$. As for the former, under Assumptions \ref{asm:sequential independence}--\ref{asm:overlap}, we can apply Lemma \ref{lem:concentration inequality_influence difference function} for the oracle influence difference function with some modifications to the notations to obtain the following result: For any stage $t$ and $\delta \in (0,1)$, with probability at least $1-2\delta$, 
\begin{align}
    \sup_{\underline{\tilde{g}}_{t}^{a},\underline{\tilde{g}}_{t}^{b} \in \widetilde{\MG}_{1:t}} \left|\widetilde{\Delta}_{t}(\underline{\tilde{g}}_{t}^{a};\underline{\tilde{g}}_{t}^{b})-\Delta_{t}(\underline{\tilde{g}}_{t}^{a};\underline{\tilde{g}}_{t}^{b})\right| &\leq \left(54.4 \sqrt{2}\kappa(\widetilde{\MG}_{1:t}) + 435.2 + \sqrt{2 \log \frac{1}{\delta}}\right)\sqrt{\frac{V_{t}^{\ast}}{n}} \nonumber \\
    &+ o\left(\frac{1}{\sqrt{n}}\right), \label{eq:bound_influence_difference_function}
\end{align}
where 
    $$V_{t}^{\ast}  \equiv  \sup_{\underline{\tilde{g}}_{t}^{a},\underline{\tilde{g}}_{t}^{b} \in \widetilde{\MG}_{1:t}}E\left[\left(\Gamma_{it}(\underline{\tilde{g}}_{t}^{a}) - \Gamma_{it}(\underline{\tilde{g}}_{t}^{b}) \right)^2\right] < \infty.$$

Note that from Remark 8 in \cite{Zhou_et_al_2023}, $\kappa(\widetilde{\MG}_{1:t}) \leq 2.5 \sqrt{VC(\widetilde{\MG}_{1:t})}$. Using Lemma \ref{lem:vc_subclass} and the fact that $VC(\widetilde{\MG}_{1:t}) \leq VC(\MG_{1:t})$, it follwos that 
\begin{align}
  \kappa(\widetilde{\MG}_{1:t}) \leq  2.5 \sqrt{VC(\MG_{1:t})} \leq 2.5 \sum_{s=1}^{t}v_{s} < \infty, \label{eq:bounded_entoropy}
\end{align}
where the last inequality follows from Assumption \ref{asm:vc-class}.

We next consider evaluating $\sup_{\underline{\tilde{g}}_{t}^{a},\underline{\tilde{g}}_{t}^{b} \in \widetilde{\MG}_{1:t}} |\widehat{\Delta}_{t}(\underline{\tilde{g}}_{t}^{a};\underline{\tilde{g}}_{t}^{b}) - \widetilde{\Delta}_{t}(\underline{\tilde{g}}_{t}^{a};\underline{\tilde{g}}_{t}^{b})|$. We employ the general strategy of the poof of \citeauthor{Zhou_et_al_2023} (\citeyear{Zhou_et_al_2023}, Lemma 3).
Fix $t$. For any $\underline{d}_{t} \in \{0,1\}^t$, let 
\begin{align*}
 &\widetilde{\Delta}_{t}^{\underline{d}_{t}}(\underline{\tilde{g}}_{t}^{a};\underline{\tilde{g}}_{t}^{b}) \equiv \frac{1}{n}\sum_{i=1}^{n} \Gamma_{it}^{\underline{d}_{t}}\left(\underline{\tilde{g}}_{t}^{a}\right) - \frac{1}{n}\sum_{i=1}^{n} \Gamma_{it}^{\underline{d}_{t}}\left(\underline{\tilde{g}}_{t}^{b}\right), \\
 & \widehat{\Delta}_{t}^{\underline{d}_{t}}(\underline{\tilde{g}}_{t}^{a};\underline{\tilde{g}}_{t}^{b}) \equiv  \frac{1}{n}\sum_{i=1}^{n} \widehat{\Gamma}_{it}^{\underline{d}_{t}}\left(\underline{\tilde{g}}_{t}^{a}\right) -  \frac{1}{n}\sum_{i=1}^{n} \widehat{\Gamma}_{it}^{\underline{d}_{t}}\left(\underline{\tilde{g}}_{t}^{b}\right), \nonumber
\end{align*}
with $\Gamma_{it}^{\underline{d}_{t}}(\underline{\tilde{g}}_{t}) \equiv {\Gamma}_{it}(\underline{d}_{t}) \cdot 1\{\underline{\tilde{g}}_t(\underline{x}_t)=\underline{d}_t\}$ and $
\widehat{\Gamma}_{it}^{\underline{d}_{t}}(\underline{\tilde{g}}_{t}) \equiv \sum_{\underline{d}_t \in \{0,1\}^t} \widehat{\Gamma}_{it}(\underline{d}_{t}) \cdot 1\{\underline{\tilde{g}}_t(\underline{x}_t)=\underline{d}_t\}$. 
Noting that $\widehat{\Delta}_{t}(\underline{\tilde
{g}}_{t}^{a};\underline{\tilde
{g}}_{t}^{b}) - \widetilde{\Delta}_{t}(\underline{\tilde
{g}}_{t}^{a};\underline{\tilde
{g}}_{t}^{b}) = \sum_{\underline{d}_{t} \in \{0,1\}^t}\left(\widehat{\Delta}_{t}^{\underline{d}_{t}}(\underline{\tilde
{g}}_{t}^{a};\underline{\tilde
{g}}_{t}^{b}) - \widetilde{\Delta}_{t}^{\underline{d}_{t}}(\underline{\tilde
{g}}_{t}^{a};\underline{\tilde
{g}}_{t}^{b})\right)$, we will provide an upper bound for each $\widehat{\Delta}_{t}^{\underline{d}_{t}}(\underline{\tilde
{g}}_{t}^{a};\underline{\tilde
{g}}_{t}^{b}) - \widetilde{\Delta}_{t}^{\underline{d}_{t}}(\underline{\tilde
{g}}_{t}^{a};\underline{\tilde
{g}}_{t}^{b})$. To do so, we make the following decomposition:
\begin{align*}
    &\widehat{\Delta}_{t}^{\underline{d}_{t}}(\underline{\tilde{g}}_{t}^{a};\underline{\tilde{g}}_{t}^{b}) - \widetilde{\Delta}_{t}^{\underline{d}_{t}}(\underline{\tilde{g}}_{t}^{a};\underline{\tilde{g}}_{t}^{b})
    = S_{1,t}^{\underline{d}_{t}}(\underline{\tilde{g}}_{t}^{a};\underline{\tilde{g}}_{t}^{b})
    + S_{2,t}^{\underline{d}_{t}}(\underline{\tilde{g}}_{t}^{a};\underline{\tilde{g}}_{t}^{b})
    + S_{3,t}^{\underline{d}_{t}}
    (\underline{\tilde{g}}_{t}^{a};\underline{\tilde{g}}_{t}^{b}),
\end{align*}
where 
\begin{align*}
        S_{(A1),t}^{\underline{d}_{t}}(\underline{\tilde{g}}_{t}^{a};\underline{\tilde{g}}_{t}^{b})&\equiv \frac{1}{n}\sum_{i=1}^{n}G_{i,\underline{\tilde{g}}_{t}^{a},\underline{\tilde{g}}_{t}^{b}}^{\underline{d}_{t}}\left(\hat{\mu}_{t}^{-k(i)}(\underline{d}_{t},\underline{X}_{it}) - {\mu}_{t}^{-k(i)}(\underline{d}_{t},\underline{X}_{it})\right) \left(1 - \frac{1\{\underline{D}_{it} = \underline{d}_{t}\}}{{\eta}_{t}(\underline{d}_{t},\underline{X}_{it})}\right),
    \\
      S_{(A2),t}^{\underline{d}_{t}}(\underline{\tilde{g}}_{t}^{a};\underline{\tilde{g}}_{t}^{b}) &\equiv \frac{1}{n}\sum_{i=1}^{n}G_{i,\underline{\tilde{g}}_{t}^{a},\underline{\tilde{g}}_{t}^{b}}^{\underline{d}_{t}}\left(Y_{it} - \mu_{t}^{-k(i)}(\underline{d}_{t},\underline{X}_{it})\right)\left(\frac{1\{\underline{D}_{it} = \underline{d}_{t}\}}{\hat{\eta}_{t}^{-k}(\underline{d}_{t},\underline{X}_{it})} - \frac{1\{\underline{D}_{it} = \underline{d}_{t}\}}{{\eta}_{t}(\underline{d}_{t},\underline{X}_{it})}\right), \\
   S_{(A3),t}^{\underline{d}_{t}}(\underline{\tilde{g}}_{t}^{a};\underline{\tilde{g}}_{t}^{b})& \equiv \frac{1}{n}\sum_{i=1}^{n}G_{i,\underline{\tilde{g}}_{t}^{a},\underline{\tilde{g}}_{t}^{b}}^{\underline{d}_{t}}\left({\mu}_{t}^{-k(i)}(\underline{d}_{t},\underline{X}_{it}) - \hat{\mu}_{t}^{-k(i)}(\underline{d}_{t},\underline{X}_{it})\right)\left(\frac{1\{\underline{D}_{it} = \underline{d}_{t}\}}{\hat{\eta}_{t}^{-k}(\underline{d}_{t},\underline{X}_{it})} - \frac{1\{\underline{D}_{it} = \underline{d}_{t}\}}{{\eta}_{t}(\underline{d}_{t},\underline{X}_{it})}\right),
\end{align*}
with $G_{i,\underline{\tilde{g}}_{t}^{a},\underline{\tilde{g}}_{t}^{b}}^{\underline{d}_{t}} := 1\{\underline{\tilde{g}}_{t}^{a}(\underline{X}_{it})=\underline{d}_t\} - 1\{\underline{\tilde{g}}_{t}^{b}(\underline{X}_{it})=\underline{d}_t\}$.

For each fold $k$, define
\begin{align*}
S_{(A1),t}^{\underline{d}_{t},k}(\underline{\tilde{g}}_{t}^{a};\underline{\tilde{g}}_{t}^{b})& \equiv \frac{1}{n}\sum_{\{i|k(i)=k\}}G_{i,\underline{\tilde{g}}_{t}^{a},\underline{\tilde{g}}_{t}^{b}}^{\underline{d}_{t}}\left(\hat{\mu}_{t}^{-k}(\underline{d}_{t},\underline{X}_{it}) - {\mu}_{t}^{-k}(\underline{d}_{t},\underline{X}_{it})\right) \left(1 - \frac{1\{\underline{D}_{it} = \underline{d}_{t}\}}{{\eta}_{t}(\underline{d}_{t},\underline{X}_{it})}\right); \\
S_{(A2),t}^{\underline{d}_{t},k}(\underline{\tilde{g}}_{t}^{a};\underline{\tilde{g}}_{t}^{b}) &\equiv \frac{1}{n}\sum_{\{i|k(i)=k\}}G_{i,\underline{\tilde{g}}_{t}^{a},\underline{\tilde{g}}_{t}^{b}}^{\underline{d}_{t}}\left(Y_{it} - \mu_{t}^{-k}(\underline{d}_{t},\underline{X}_{it})\right)\left(\frac{1\{\underline{D}_{it} = \underline{d}_{t}\}}{\hat{\eta}_{t}^{-k}(\underline{d}_{t},\underline{X}_{it})} - \frac{1\{\underline{D}_{it} = \underline{d}_{t}\}}{{\eta}_{t}(\underline{d}_{t},\underline{X}_{it})}\right).
\end{align*}
Note that $S_{(A1),t}^{\underline{d}_{t}}(\underline{\tilde{g}}_{t}^{a};\underline{\tilde{g}}_{t}^{b}) = \sum_{k=1}^{K}S_{(A1),t}^{\underline{d}_{t},k}(\underline{\tilde{g}}_{t}^{a};\underline{\tilde{g}}_{t}^{b})$ and $S_{(A2),t}^{\underline{d}_{t}}(\underline{\tilde{g}}_{t}^{a};\underline{\tilde{g}}_{t}^{b}) = \sum_{k=1}^{K}S_{(A2),t}^{\underline{d}_{t},k}(\underline{\tilde{g}}_{t}^{a};\underline{\tilde{g}}_{t}^{b})$.

Fix $k \in \{1,\ldots,K\}$. We first consider $S_{(A1),t}^{\underline{d}_{t}}(\underline{\tilde{g}}_{t}^{a};\underline{\tilde{g}}_{t}^{b})$. Since $\hat{\mu}_{t}^{-k}(\underline{d}_{t},\cdot)$ is computed using the data in the rest $K-1$ folds, when the data $\{Z_i : k(i) \neq k\}$ in the rest $K-1$ folds is conditioned,  $\hat{\mu}_{t}^{-k}(\underline{d}_{t},\cdot)$ is fixed; hence, $\widetilde{S}_{(A1),t}^{\underline{d}_{t},k}\left(\underline{\tilde{g}}_{t}^{a};\underline{\tilde{g}}_{t}^{b}\right)$ is a sum of i.i.d. bounded random variables under Assumptions \ref{asm:bounded outcome}, \ref{asm:overlap}, and \ref{asm:rate_of_convergence_simultaneous} (ii).

It follows that 
\begin{align*}	&E\left[G_{i,\underline{\tilde{g}}_{t}^{a};\underline{\tilde{g}}_{t}^{b}}^{\underline{d}_{t}}  \left(\hat{\mu}_{t}^{-k}(\underline{d}_{t},\underline{X}_{it}) - {\mu}_{t}^{-k}(\underline{d}_{t},\underline{X}_{it})\right) \left(1 - \frac{1\{\underline{D}_{it} = \underline{d}_{t}\}}{{\eta}_{t}(\underline{d}_{t},\underline{X}_{it})}\right)\middle|\hat{\mu}_{t}^{-k}(\underline{d}_{t},\cdot)\right]\\
=&E\left[G_{i,\underline{\tilde{g}}_{t}^{a};\underline{\tilde{g}}_{t}^{b}}^{\underline{d}_{t}}  \left(\hat{\mu}_{t}^{-k}(\underline{d}_{t},\underline{X}_{it}) - {\mu}_{t}^{-k}(\underline{d}_{t},\underline{X}_{it})\right) E\left[\left(1 - \frac{1\{\underline{D}_{it} = \underline{d}_{t}\}}{{\eta}_{t}(\underline{d}_{t},\underline{X}_{it})}\right)\middle|\underline{X}_{it} \right]\middle|\hat{\mu}_{t}^{-k}(\underline{d}_{t},\cdot)\right]\\
=&	0,
\end{align*}
where the last line follows from the sequential independence assumption (Assumption \ref{asm:sequential independence}).
Hence, $\sup_{\underline{\tilde{g}}_{t}^a , \underline{\tilde{g}}_{t}^b \in  \widetilde{\MG}_{1:t}} \left| ,k\widetilde{S}_{(A1),t}^{\underline{d}_{t}}\left(\underline{\tilde{g}}_{t}^{a};\underline{\tilde{g}}_{t}^{b}\right)\right|$ can be written as
\begin{align*}
	&\sup_{\underline{\tilde{g}}_{t}^{a},\underline{\tilde{g}}_{t}^{b}\in \widetilde{\MG}_{1:t}}\left|\widetilde{S}_{(A1),t}^{\underline{d}_{t},k}\left(\underline{\tilde{g}}_{t}^{a};\underline{\tilde{g}}_{t}^{b}\right)\right|\\
&=	\frac{1}{K}\sup_{\underline{\tilde{g}}_{t}^{a},\underline{\tilde{g}}_{t}^{b}\in \widetilde{\MG}_{1:t}}\left|\frac{1}{n/K}\sum_{i \in I_k }\left\{ 
G_{i,\underline{\tilde{g}}_{t}^{a};\underline{\tilde{g}}_{t}^{b}}^{\underline{d}_{t}}  \left(\hat{\mu}_{t}^{-k}(\underline{d}_{t},\underline{X}_{it}) - {\mu}_{t}^{-k}(\underline{d}_{t},\underline{X}_{it})\right) \left(1 - \frac{1\{\underline{D}_{it} = \underline{d}_{t}\}}{{\eta}_{t}(\underline{d}_{t},\underline{X}_{it})}\right)
\right.\right.\\
&\left. \left. - E\left[G_{i,\underline{\tilde{g}}_{t}^{a};\underline{\tilde{g}}_{t}^{b}}^{\underline{d}_{t}}  \left(\hat{\mu}_{t}^{-k}(\underline{d}_{t},\underline{X}_{it}) - {\mu}_{t}^{-k}(\underline{d}_{t},\underline{X}_{it})\right) \left(1 - \frac{1\{\underline{D}_{it} = \underline{d}_{t}\}}{{\eta}_{t}(\underline{d}_{t},\underline{X}_{it})}\right)\middle|\hat{\mu}_{t}^{-k}(\underline{d}_{t},\cdot)\right]\right\} \right|.
\end{align*}

By applying Lemma \ref{lem:concentration inequality_influence difference function} with setting $i \in I_k$ and 
\begin{align*}
    \Gamma_{i}(d_t)=G_{i,\underline{\tilde{g}}_{t}^{a};\underline{\tilde{g}}_{t}^{b}}^{\underline{d}_{t}}  \left(\hat{\mu}_{t}^{-k}(\underline{d}_{t},\underline{X}_{it}) - {\mu}_{t}^{-k}(\underline{d}_{t},\underline{X}_{it})\right) \left(1 - \frac{1\{\underline{D}_{it} = \underline{d}_{t}\}}{{\eta}_{t}(\underline{d}_{t},\underline{X}_{it})}\right),
\end{align*}
the following holds: $\forall \delta > 0$, with probability at least $1-2\delta$,
\begin{align*}
    	&\sup_{\underline{\tilde{g}}_{t}^{a},\underline{\tilde{g}}_{t}^{b}\in \widetilde{\MG}_{1:t}}\left|\widetilde{S}_{(A1),t}^{\underline{d}_{t},k}\left(\underline{\tilde{g}}_{t}^{a};\underline{\tilde{g}}_{t}^{b}\right)\right|\\
&\leq	o\left(n^{-1/2}\right)+\left(54.4\kappa\left(\widetilde{\MG}_{1:t}\right)+435.2+\sqrt{2\log(1/\delta)}\right) \\
	&\times\left[\sup_{\underline{\tilde{g}}_{t}^{a},\underline{\tilde{g}}_{t}^{b}\in \widetilde{\MG}_{1:t}}E\left[\left(G_{i,\underline{\tilde{g}}_{t}^{a};\underline{\tilde{g}}_{t}^{b}}^{\underline{d}_{t}}\right)^{2}\left(\hat{\mu}_{t}^{-k}(\underline{d}_{t},\underline{X}_{it}) - {\mu}_{t}^{-k}(\underline{d}_{t},\underline{X}_{it})\right)^{2} \right.\right.\\
&	\left.\left.\left.\times  \left(1 - \frac{1\{\underline{D}_{it} = \underline{d}_{t}\}}{{\eta}_{t}(\underline{d}_{t},\underline{X}_{it})}\right)\right|\hat{\mu}_{t}^{-k}\left(\underline{d}_{t},\cdot\right) \right] \middle/ \left(\frac{n}{K}\right) \right]^{1/2} \\
&\leq	o\left(n^{-1/2}\right)+\sqrt{K} \cdot \left(54.4\kappa\left(\widetilde{\MG}_{1:t}\right)+435.2+\sqrt{2\log(1/\delta)}\right)
	\cdot \left(1-\frac{1}{\eta}\right)^{t}\\
	&\times \sqrt{\frac{E\left[\left.\left(\hat{\mu}_{t}^{-k}(\underline{d}_{t},\underline{X}_{it}) - {\mu}_{t}^{-k}(\underline{d}_{t},\underline{X}_{it})\right)^2\right|\hat{\mu}_{t}^{-k}\left(\underline{d}_{t},\cdot\right) \right]}{n}},
\end{align*}
where the last inequality follows from $\left(G_{i,\underline{\tilde{g}}_{t}^{a};\underline{\tilde{g}}_{t}^{b}}^{\underline{d}_{t}}\right)^{2}\leq 1$ a.s. and Assumption \ref{asm:overlap} (overlap condition). From Assumptions \ref{asm:bounded outcome} and \ref{asm:rate_of_convergence_simultaneous} (ii), we have $E\left[\left(\hat{\mu}_{t}^{-k}(\underline{d}_{t},\underline{X}_{it}) - {\mu}_{t}^{-k}(\underline{d}_{t},\underline{X}_{it})\right)^2\right]<\infty$. Hence, Markov's inequality leads to
\begin{align*}
    E\left[\left(\hat{\mu}_{t}^{-k}(\underline{d}_{t},\underline{X}_{it}) - {\mu}_{t}^{-k}(\underline{d}_{t},\underline{X}_{it})\right)^2\middle|\hat{\mu}_{t}^{-k}\left(\underline{d}_{t},\cdot\right) \right] = O_p(1).
\end{align*}
Note also that $\kappa(\widetilde{\MG}_{1:t})<\infty$ from (\ref{eq:bounded_entoropy}). Combining these results, we have
\begin{align*}
    \sup_{\underline{\tilde{g}}_{t}^{a},\underline{\tilde{g}}_{t}^{b}\in \widetilde{\MG}_{1:t}}\left|\widetilde{S}_{(A1),t}^{\underline{d}_{t},k}\left(\underline{\tilde{g}}_{t}^{a};\underline{\tilde{g}}_{t}^{b}\right)\right| = O_p \left(\frac{1}{\sqrt{n}}\right). 
\end{align*}
Consequently, 
\begin{align}
    \sup_{\underline{\tilde{g}}_{t}^{a},\underline{\tilde{g}}_{t}^{b}\in \widetilde{\MG}_{1:t}}\left|\widetilde{S}_{(A1),t}^{\underline{d}_{t}}\left(\underline{\tilde{g}}_{t}^{a};\underline{\tilde{g}}_{t}^{b}\right)\right| \leq  \sum_{k=1}^{K}\sup_{\underline{\tilde{g}}_{t}^{a},\underline{\tilde{g}}_{t}^{b}\in \widetilde{\MG}_{1:t}}\left|\widetilde{S}_{(A1),t}^{\underline{d}_{t},k}\left(\underline{\tilde{g}}_{t}^{a};\underline{\tilde{g}}_{t}^{b}\right)\right| = O_p \left(\frac{1}{\sqrt{n}}\right). \label{eq:S_A1_bound}
\end{align}
By the same argument, we have $\sup_{\underline{\tilde{g}}_{t}^{a},\underline{\tilde{g}}_{t}^{b}\in \widetilde{\MG}_{1:t}}\left|\widetilde{S}_{(A2),t}^{\underline{d}_{t},k}\left(\underline{\tilde{g}}_{t}^{a};\underline{\tilde{g}}_{t}^{b}\right)\right| = O_p \left(1/\sqrt{n}\right)$. Hence 
\begin{align}
    \sup_{\underline{\tilde{g}}_{t}^{a},\underline{\tilde{g}}_{t}^{b}\in \widetilde{\MG}_{1:t}}\left|\widetilde{S}_{(A2),t}^{\underline{d}_{t}}\left(\underline{\tilde{g}}_{t}^{a};\underline{\tilde{g}}_{t}^{b}\right)\right| \leq \sum_{k=1}^{K} \sup_{\underline{\tilde{g}}_{t}^{a},\underline{\tilde{g}}_{t}^{b}\in \widetilde{\MG}_{1:t}}\left|\widetilde{S}_{(A2),t}^{\underline{d}_{t},k}\left(\underline{\tilde{g}}_{t}^{a};\underline{\tilde{g}}_{t}^{b}\right)\right| = O_p \left(\frac{1}{\sqrt{n}}\right). \label{eq:S_A2_bound}
\end{align}

We next consider to bound $S_{(A3),t}^{\underline{d}_{t}}(\cdot,\cdot)$ from above. It follows that 
\begin{align*}
    &\sup_{\underline{\tilde{g}}_{t}^{a},\underline{\tilde{g}}_{t}^{b}\in \widetilde{\MG}_{1:t}}\left|\widetilde{S}_{(A3),t}^{\underline{d}_{t}}\left(\underline{\tilde{g}}_{t}^{a};\underline{\tilde{g}}_{t}^{b}\right)\right|\\
&= \frac{1}{n}\sup_{\underline{\tilde{g}}_{t}^{a},\underline{\tilde{g}}_{t}^{b}\in \widetilde{\MG}_{1:t}}\left|\sum_{i=1}^{n}G_{i,\underline{\tilde{g}}_{t}^{a},\underline{\tilde{g}}_{t}^{b}}^{\underline{d}_{t}}\left({\mu}_{t}(\underline{d}_{t},\underline{X}_{it}) - \hat{\mu}_{t}^{-k(i)}(\underline{d}_{t},\underline{X}_{it})\right)\left(\frac{1}{\hat{\eta}_{t}^{-k(i)}(\underline{d}_{t},\underline{X}_{it})} - \frac{1}{{\eta}_{t}(\underline{d}_{t},\underline{X}_{it})}\right)\right|\\
& \leq \frac{1}{n}\sum_{\{i|\underline{D}_{it} = \underline{d}_{t}\}}\left|\left(\mu_{t}(\underline{d}_{t},\underline{X}_{it}) - \hat{\mu}_{t}^{-k(i)}(\underline{d}_{t},\underline{X}_{it})\right)\right| \cdot \left|\left(\frac{1}{\hat{\eta}_{t}^{-k(i)}(\underline{d}_{t},\underline{X}_{it})} - \frac{1}{{\eta}_{t}(\underline{d}_{t},\underline{X}_{it})}\right)\right| \\
& \leq \sqrt{\frac{1}{n}\sum_{\{i|\underline{D}_{it} = \underline{d}_{t}\}}\left(\mu_{t}(\underline{d}_{t},\underline{X}_{it}) - \hat{\mu}_{t}^{-k(i)}(\underline{d}_{t},\underline{X}_{it})\right)^2} \sqrt{\frac{1}{n}\sum_{\{i|\underline{D}_{it} = \underline{d}_{t}\}}\left(\frac{1}{\hat{\eta}_{t}^{-k(i)}(\underline{d}_{t},\underline{X}_{it})} - \frac{1}{{\eta}_{t}(\underline{d}_{t},\underline{X}_{it})}\right)^2},
\end{align*}
where the last inequality follows from Cauchy-Schwartz inequality. Taking expectation of both sides yields:

\begin{align*}E\left[\sup_{\underline{\tilde{g}}_{t}^{a},\underline{\tilde{g}}_{t}^{b}\in \widetilde{\MG}_{1:t}}\left|\widetilde{S}_{(A3),t}^{\underline{d}_{t}}\left(\underline{\tilde{g}}_{t}^{a};\underline{\tilde{g}}_{t}^{b}\right)\right|\right] &\leq E\left[\sqrt{\frac{1}{n}\sum_{\{i|\underline{D}_{it} = \underline{d}_{t}\}}\left(\mu_{t}(\underline{d}_{t},\underline{X}_{it}) - \hat{\mu}_{t}^{-k(i)}(\underline{d}_{t},\underline{X}_{it})\right)^2}\right]\\
    &\times E\left[\sqrt{\frac{1}{n}\sum_{\{i|\underline{D}_{it} = \underline{d}_{t}\}}\left(\frac{1}{\hat{\eta}_{t}^{-k(i)}(\underline{d}_{t},\underline{X}_{it})} - \frac{1}{{\eta}_{t}(\underline{d}_{t},\underline{X}_{it})}\right)^2}\right]\\
    &\leq \sqrt{\frac{1}{n}\sum_{\{i|\underline{D}_{it} = \underline{d}_{t}\}}E\left[\left(\mu_{t}(\underline{d}_{t},\underline{X}_{it}) - \hat{\mu}_{t}^{-k(i)}(\underline{d}_{t},\underline{X}_{it})\right)^2\right]}\\
    &\times \sqrt{\frac{1}{n}\sum_{\{i|\underline{D}_{it} = \underline{d}_{t}\}}E\left[\left(\frac{1}{\hat{\eta}_{t}^{-k(i)}(\underline{d}_{t},\underline{X}_{it})} - \frac{1}{{\eta}_{t}(\underline{d}_{t},\underline{X}_{it})}\right)^2\right]}\\
    &= o(n^{-\tau^{\prime} /2}),
\end{align*}
where the second inequality follows from Cauchy-Schwartz inequality and the last line follows from Assumption \ref{asm:rate_of_convergence_simultaneous} (i). Then applying Markov's inequality leads to
\begin{align}
    \sup_{\underline{\tilde{g}}_{t}^{a},\underline{\tilde{g}}_{t}^{b}\in \widetilde{\MG}_{1:t}}\left|\widetilde{S}_{(A3),t}^{\underline{d}_{t}}\left(\underline{\tilde{g}}_{t}^{a};\underline{\tilde{g}}_{t}^{b}\right)\right| &= O_{P}\left(n^{-\tau^{\prime}/2}\right). \label{eq:S_A3_bound}
\end{align}

We therefore obtain
\begin{align}
    &\sum_{t=1}^{T} \sup_{\underline{\tilde{g}}_{t}^{a},\underline{\tilde{g}}_{t}^{b} \in \widetilde{\MG}_{1:t}} |\widehat{\Delta}_{t}(\underline{\tilde{g}}_{t}^{a};\underline{\tilde{g}}_{t}^{b}) - \widetilde{\Delta}_{t}(\underline{\tilde{g}}_{t}^{a};\underline{\tilde{g}}_{t}^{b})| \nonumber \\
    &\leq \sum_{t=1}^{T} \sum_{\underline{d}_{t} \in \{0,1\}^{t}}\sup_{\underline{\tilde{g}}_{t}^{a},\underline{\tilde{g}}_{t}^{b} \in \widetilde{\MG}_{1:t}} |\widehat{\Delta}_{t}^{\underline{d}_{t}}(\underline{\tilde{g}}_{t}^{a};\underline{\tilde{g}}_{t}^{b}) - \widetilde{\Delta}_{t}^{\underline{d}_{t}}(\underline{\tilde{g}}_{t}^{a};\underline{\tilde{g}}_{t}^{b})| \nonumber\\
    & \leq \sum_{t=1}^{T} \sum_{\underline{d}_{t} \in \{0,1\}^{t}} \sup_{\underline{\tilde{g}}_{t}^{a},\underline{\tilde{g}}_{t}^{b}\in \widetilde{\MG}_{1:t}}\left|\widetilde{S}_{(A1),t}^{\underline{d}_{t}}\left(\underline{\tilde{g}}_{t}^{a};\underline{\tilde{g}}_{t}^{b}\right)\right| + \sum_{t=1}^{T} \sum_{\underline{d}_{t} \in \{0,1\}^{t}}\sup_{\underline{\tilde{g}}_{t}^{a},\underline{\tilde{g}}_{t}^{b}\in \widetilde{\MG}_{1:t}}\left|\widetilde{S}_{(A2),t}^{\underline{d}_{t}}\left(\underline{\tilde{g}}_{t}^{a};\underline{\tilde{g}}_{t}^{b}\right)\right| \nonumber \\
    & + \sum_{t=1}^{T} \sum_{\underline{d}_{t} \in \{0,1\}^{t}}\sup_{\underline{\tilde{g}}_{t}^{a},\underline{\tilde{g}}_{t}^{b}\in \widetilde{\MG}_{1:t}}\left|\widetilde{S}_{(A3),t}^{\underline{d}_{t}}\left(\underline{\tilde{g}}_{t}^{a};\underline{\tilde{g}}_{t}^{b}\right)\right| \nonumber\\
    &= O_{p}\left(n^{-\min\{1/2,\tau^{\prime}/2\}}\right), \label{eq:latter_bound}
\end{align}
where the last line follows from (\ref{eq:S_A1_bound}), (\ref{eq:S_A2_bound}), and (\ref{eq:S_A3_bound}).

Combining (\ref{eq:standard_inequality}), (\ref{eq:bound_influence_difference_function}), (\ref{eq:bounded_entoropy}), and (\ref{eq:latter_bound}) leads to the result (\ref{eq:DML_result}). 
\end{proof}


\section{Additional Simulation Results \label{appendix:simulation_results}}

We conduct an additional simulation study to examine the finite sample performance for the estimation methods proposed in Section \ref{sec:DEWM} under the circumstance that the sequential independence assumption does not hold due to the presence of unobserved heterogeneity. We consider the same DGPs as those used in Section \ref{sec:simulation}, except that the treatment assignments $D_1$ and $D_2$ are distributed as
\begin{align}
    D_1 \sim 1\{N(0,1) + \rho\cdot U_1 \geq 0\}\mbox{ and }D_2 \sim 1\{N(0,1) + \rho\cdot U_2 \geq 0\}. \label{eq:dependent_treatment}
\end{align}
Recall that the potential outcomes $Y_1(d_1)$ and $Y_2(d_1,d_2)$ depend on $U_1$ and $U_2$, respectively. Hence, unless $\rho \neq 0$, the sequential independence assumption (Assumption \ref{asm:sequential independence}) is not satisfied. We consider two values of $\rho$: $\rho=-1$ and $1$. For each $j=1,2,3$, we label the DGP that is the same as DGP $j$ used in Section \ref{sec:simulation} except for that $D_1$ and $D_2$ follow equation (\ref{eq:dependent_treatment}) with $\rho=-1$ and $1$ as DGPs $j^{\prime}$ and $j^{\prime\prime}$, respectively.

Table \ref{table:additional_simulation_results} presents the results of 500 simulations with sample sizes $n=200$, $500$, and $800$, where we compare Q-learning, backward DEWM, and simultaneous DEWM and calculate the mean and median welfare achieved by each estimated DTR. 
Panel (A) of Table \ref{table:additional_simulation_results} shows that in the case of $\rho=-1$, simultaneous DEWM leads to the lower mean welfare in DGP $3^\prime$ than Q-learning and backward DEWM. This result differs from the simulation results in Section \ref{sec:simulation}, where the DGPs satisfy the sequential independence assumption, and simultaneous DEWM leads to the highest mean welfare in DGP3. Panel (B) of Table \ref{table:additional_simulation_results} shows that the DGPs $1^{\prime \prime}$, $2^{\prime\prime}$, and $3^{\prime\prime}$ lead to similar results to those with DGPs 1-3 in terms of the order of mean/median welfare among the three methods.


\begin{table}[h!]
\centering \caption{Additional Monte Carlo Simulation Results}
\label{table:additional_simulation_results}
\bigskip
\centering {Panel(A) DGPs $1^\prime$-$3^\prime$}
\scalebox{0.9}{
\begin{tabular}{cccccccccccccc}
\hline 
 &  &  & n=200  &  &  &  & n=500  &  &  &  & n=800  &  & \tabularnewline
\hline 
 & DGP  & Mean  & Median  & SD  &  & Mean  & Median  & SD  &  & Mean  & Median  & SD  & \tabularnewline
\hline 
\text{Q-learning} & \text{$1^{\prime}$} & 1.846 & 1.85 & 0.042 & \text{} & 1.853 & 1.853 & 0.042 & \text{} & 1.858 & 1.856 & 0.038 & \text{} \\
\text{B-DEWM} & \text{$1^{\prime}$} & 1.608 & 1.625 & 0.196 & \text{} & 1.637 & 1.662 & 0.168 & \text{} & 1.663 & 1.695 & 0.163 & \text{} \\
\text{S-DEWM} & \text{$1^{\prime}$} & 1.456 & 1.523 & 0.263 & \text{} & 1.442 & 1.514 & 0.273 & \text{} & 1.435 & 1.503 & 0.27 & \text{} \\
\hdashline 
\text{Q-learning} & \text{$2^{\prime}$} & 3.125 & 3.125 & 0.058 & \text{} & 3.128 & 3.13 & 0.059 & \text{} & 3.125 & 3.128 & 0.059 & \text{} \\
\text{B-DEWM} & \text{$2^{\prime}$} & 2.54 & 2.621 & 0.487 & \text{} & 2.649 & 2.754 & 0.419 & \text{} & 2.693 & 2.816 & 0.42 & \text{} \\
\text{S-DEWM} & \text{$2^{\prime}$} & 2.585 & 2.806 & 0.569 & \text{} & 2.761 & 2.985 & 0.527 & \text{} & 2.905 & 3.041 & 0.421 & \text{} \\
\hdashline 
\text{Q-learning} & \text{$3^{\prime}$} & 1.574 & 1.566 & 0.22 & \text{} & 1.545 & 1.545 & 0.207 & \text{} & 1.543 & 1.538 & 0.178 & \text{} \\
\text{B-DEWM} & \text{$3^{\prime}$} & 1.678 & 1.734 & 0.186 & \text{} & 1.721 & 1.753 & 0.134 & \text{} & 1.721 & 1.745 & 0.129 & \text{} \\
\text{S-DEWM} & \text{$3^{\prime}$} & 1.361 & 1.346 & 0.144 & \text{} & 1.36 & 1.343 & 0.12 & \text{} & 1.351 & 1.333 & 0.11 & \text{} 
 \tabularnewline
\hline 
\end{tabular}
}

\bigskip
\par
\bigskip

\centering {Panel(B) DGPs $1^{\prime\prime}$-$3^{\prime\prime}$}
\scalebox{0.9}{
\begin{tabular}{cccccccccccccc}
\hline 
 &  &  & n=200  &  &  &  & n=500  &  &  &  & n=800  &  & \tabularnewline
\hline 
 & DGP  & Mean  & Median  & SD  &  & Mean  & Median  & SD  &  & Mean  & Median  & SD  & \tabularnewline
\hline 
\text{Q-learning} & \text{$1^{\prime\prime}$} & 3.099 & 3.098 & 0.037 & \text{} & 3.1 & 3.101 & 0.036 & \text{} & 3.102 & 3.103 & 0.035 & \text{} \\
\text{B-DEWM} & \text{$1^{\prime\prime}$} & 2.875 & 2.983 & 0.298 & \text{} & 2.976 & 3.046 & 0.212 & \text{} & 2.979 & 3.063 & 0.235 & \text{} \\
\text{S-DEWM} & \text{$1^{\prime\prime}$} & 2.901 & 3.008 & 0.327 & \text{} & 2.994 & 3.053 & 0.225 & \text{} & 2.994 & 3.07 & 0.309 & \text{} \\
\hdashline 
\text{Q-learning} & \text{$2^{\prime\prime}$} & 5.206 & 5.205 & 0.061 & \text{} & 5.206 & 5.21 & 0.059 & \text{} & 5.209 & 5.208 & 0.063 & \text{} \\
\text{B-DEWM} & \text{$2^{\prime\prime}$} & 4.639 & 4.995 & 0.757 & \text{} & 4.854 & 5.075 & 0.563 & \text{} & 4.882 & 5.104 & 0.657 & \text{} \\
\text{S-DEWM} & \text{$2^{\prime\prime}$} & 4.858 & 5.044 & 0.565 & \text{} & 4.999 & 5.097 & 0.442 & \text{} & 5.037 & 5.118 & 0.388 & \text{} \\
\hdashline 
\text{Q-learning} & \text{$3^{\prime\prime}$} & 2.19 & 2.186 & 0.123 & \text{} & 2.186 & 2.192 & 0.115 & \text{} & 2.191 & 2.192 & 0.112 & \text{} \\
\text{B-DEWM} & \text{$3^{\prime\prime}$} & 2.02 & 1.906 & 0.272 & \text{} & 2.151 & 2.249 & 0.27 & \text{} & 2.224 & 2.272 & 0.218 & \text{} \\
\text{S-DEWM} & \text{$3^{\prime\prime}$} & 2.315 & 2.334 & 0.15 & \text{} & 2.311 & 2.317 & 0.128 & \text{} & 2.3 & 2.312 & 0.142 & \text{} \tabularnewline
\hline 
\end{tabular}
}

\bigskip

\begin{tablenotes} \footnotesize
\item Note: Mean and Median represent the mean
and median of the population mean welfares achieved by the estimated
DTRs across the simulations; SD is the standard deviation of the population mean welfares across the simulations.
The population mean welfare is calculated using 3,000 observations
randomly drawn from the corresponding DGP. B-DEWM and S-DEWM mean the Backward and Simultaneous DEWM methods, respectively. 
\end{tablenotes} 
\end{table}

\section{Computation \label{appendix:computation}}

In this appendix, we explain computation of the backward and simultaneous DEWM with $\MG_t$ ($t=1,\ldots,T$) being classes of the linear treatment rules. The non-convexity of the objective functions make these computations challenging. However, the optimization problems can be formulated as Mixed Integer Linear Programming (MILP) problems, for which some efficient softwares (e.g., CPLEX; Gurobi) are available.
In the following subsections, we illustrate the MILP formalization for each of the backward and simultaneous DEWM in the case of $T=2$. We suppose that the class of feasible treatment rules for each stage $t=1,2$ takes the form of $\mathcal{G}_{t}=\left\{ 1\left\{(1,H_{t}^{\prime})\beta_{t} \geq 0\right\}:\beta_{t}\in \MB_{t} \subset \mathbb{R}^{(k+2)t-1}\right\} $ where $\MB_t$ is a compact set.

\subsection{Backward DEWM}

Using slightly different notation from Section \ref{sec:backward DEWM}, the first step of the backward DEWM method is
\begin{align*}
\max_{g_{2}\in{\cal G}_{2}} & \sum_{i=1}^{n}m_{i2}^{B}g_{2},
\end{align*}
where 
\begin{align*}
m_{i2}^{B}= & \left(\frac{D_{i2}}{e_{2}\left(1,H_{i2}\right)}-\frac{1-D_{i2}}{e_{2}\left(0,H_{i2}\right)}\right)\gamma_{2}Y_{i2}.
\end{align*}
Let $\hat{g}_{2}^{B}$ be a maximizer of the above problem.
Then, the second step of the backward DEWM method is 
\begin{align*}
\max_{g_{1}\in{\cal G}_{1}} & \sum_{i=1}^{n}m_{i1}^{B}g_{1},
\end{align*}
where 
\begin{align*}
m_{i1}^{B}&=  \left(\frac{D_{i1}}{e_{1}\left(1,H_{i1}\right)}-\frac{1-D_{i1}}{e_{2}\left(0,H_{i1}\right)}\right)\\
& \times
\left(\frac{D_{i2}\hat{g}_{2}^{B}\left(H_{i2}\right)}{e_{2}\left(1,H_{i2}\right)}-\frac{\left(1-D_{i2}\right)\left(1-\hat{g}_{2}^{B}\left(H_{i2}\right)\right)}{e_{2}\left(0,H_{i2}\right)}\right)   \left(\gamma_{1}Y_{1i}+\gamma_{2}Y_{i2}\right).
\end{align*}

When the class of DTRs is constrained to the class of linear eligibility
rules, each step of the backward DEWM method described in Section \ref{sec:backward DEWM} can be formulated as MILP problem. 
The optimization problem in the first step is equivalent to
the following MILP problem:

\begin{description}

\item [{$\mbox{(First step)}$}] 
\begin{align*}
\underset{\left(z_{12},\ldots,z_{n2}\right)\in\left\{ 0,1\right\} ^{n}}{\max_{\beta_{2}\in\MB_2}} & \sum_{i=1}^{n}m_{i2}^{B}z_{i2}\\
\mbox{s.t. } & \frac{(1,H_{i2}^{\prime})\beta_{2}}{C_{i2}}<z_{i2}\leq1+\frac{(1,H_{i2}^{\prime})\beta_{2}}{C_{i2}}\ \mbox{for }i=1,\ldots,n,
\end{align*}
where $C_{i2}$ are constants that should satisfy $C_{i2}>\sup_{\beta_{2}\in \MB_2}\left|(1,H_{i2}^{\prime})\beta_{2}\right|$. 
\end{description}

Subsequently, the optimization problem in the second step is equivalent to
the following MILP problem: 

\begin{description}
\item [{$\mbox{(Second step)}$}] 
\begin{align*}
\underset{\left(z_{11},\ldots,z_{n1}\right)\in\left\{ 0,1\right\} ^{n}}{\max_{\beta_{1}\in\MB_1}} & \sum_{i=1}^{n}m_{i1}^{B}z_{i1}\\
\mbox{s.t. } & \frac{(1,H_{i1}^{\prime})\beta_{1}}{C_{i1}}<z_{i1}\leq1+\frac{(1,H_{i1}^{\prime})\beta_{1}}{C_{i1}}\ \mbox{for }i=1,\ldots,n,
\end{align*}
where $C_{i1}$ are constants that should satisfy $C_{i1}>\sup_{\beta_{1}\in \MB_1}\left|(1,H_{i1}^{\prime})\beta_{1}\right|$. 
\end{description}

\medskip

When we specify the dynamic treatment choice problem as the start (stop) time decision problem discussed in Section \ref{sec:dynamic treatment choice problem}, the linear constraints $z_{i2}\geq D_{i1}$ and $D_{i2}\geq z_{i1}$
($z_{i2}\leq D_{i1}$ and $D_{i2}\leq z_{i1}$) should be added into the MILP problems for the first and second steps, respectively. When we specify the problem as the one-shot treatment decision problem discussed in Section \ref{sec:dynamic treatment choice problem}, the linear constraints $z_{i2} + D_{i1}\leq1$ and $D_{i2} + z_{i1}\leq1$ should be added into the
MILP problems for the first and second steps, respectively.

\subsection{Simultaneous DEWM}

In the case of $T=2$, the optimization problem of the simultaneous DEWM method is equivalent to 
\begin{align*}
\max_{\left(g_{1},g_{2}\right)\in{\cal G}} & \sum_{i=1}^{n}\left[m_{i1}^{S}g_{1}+m_{i2}^{S}g_{2}+m_{i3}^{S}g_{1}g_{2}\right],
\end{align*}
where $m_{is}^{S}$ for $s=1,2,3$ are defined as 
\begin{align*}
m_{i1}^{S}&=  \frac{D_{i1}}{e_{1}\left(1,H_{i1}\right)}\left(\gamma_{1}Y_{i1}+\frac{\left(1-D_{i2}\right)\gamma_{2}Y_{i2}}{e_{2}\left(0,H_{i2}\right)}\right)-\frac{1-D_{i1}}{e_{1}\left(0,H_{i1}\right)}\left(\gamma_{1}Y_{i1}+\frac{\left(1-D_{i2}\right)\gamma_{2}Y_{i2}}{e_{2}\left(0,H_{i2}\right)}\right),\\
m_{i2}^{S}&=  \left(\frac{\left(1-D_{i1}\right)D_{i2}}{e_{1}\left(0,H_{i1}\right)e_{2}\left(1,H_{i2}\right)}-\frac{\left(1-D_{i1}\right)\left(1-D_{i2}\right)}{e_{1}\left(0,H_{i1}\right)e_{2}\left(0,H_{i2}\right)}\right)\gamma_{2}Y_{i2},\\
m_{i3}^{S}&= \sum_{\left(d_{1},d_{2}\right)\in\left\{ 0,1\right\} ^{2}}\frac{1\left\{ D_{i1}=d_{1},D_{i2}=d_{2}\right\} \gamma_{2}Y_{i2}}{e_{1}\left(d_{1},H_{i1}\right)e_{2}\left(d_{2},H_{i2}\right)}.
\end{align*}
When the class of DTRs is constrained to the class of linear eligibility
rules, the above optimization problem is equivalent to the following
MILP problem: 
\begin{align*}
\underset{\left(z_{1t},\ldots,z_{nt}\right)_{t=1}^{3}\in \{0,1\}^{3n}}{\max_{\left(\beta_{1},\beta_{2}\right)\in\MB_{1} \times \MB_{2}}} & \sum_{i=1}^{n}\left[m_{i1}^{S}z_{i1}+m_{i2}^{S}z_{i2}+m_{i3}^{S}z_{i3}\right]\\
\mbox{s.t. } & \frac{(1,H_{it}^{\prime})\beta_{t}}{C_{it}}<z_{it}\leq1+\frac{(1,H_{it}^{\prime})\beta_{t}}{C_{it}}\ \mbox{for }i=1,\ldots,n\mbox{ and }t=1,2,\\
 & z_{i3}=z_{i1} z_{i2}\mbox{ for }i=1,\ldots,n,
\end{align*}
where $C_{it}$ are constants that should satisfy $C_{it}>\sup_{\beta_{t}\in \MB_{t}}\left|(1,H_{it}^{\prime})\beta_{t}\right|$.

When we specify the dynamic treatment choice problem as the start (stop) time decision problem discussed in Section \ref{sec:dynamic treatment choice problem}, the linear constraints $z_{i2}\geq z_{i1}$
($z_{i2}\leq z_{i1}$)  should be added into the MILP problem. When we specify the problem as the one-shot treatment decision problem discussed in Section \ref{sec:dynamic treatment choice problem}, the linear constraint $z_{i1} + z_{i2}\leq1$ should be added into the MILP problem.

\subsection{Budget/Capacity Constraint}

The budget/capacity constraints studied in Section \ref{sec:budget constraint} can be incorporated into the MILP problem for the simultaneous DEWM. The optimization
problem (\ref{eq:budget constrained SDEWM}) with the class of linear eligibility score rules is formulated as the following MILP problem: 
\begin{align*}
\underset{\left(z_{1t},\ldots,z_{nt}\right)_{t=1}^{3}\in \{0,1\}^{3n}}{\max_{\left(\beta_{1},\beta_{2}\right)\in\MB_{1} \times \MB_{2}}} & \sum_{i=1}^{n}\left[m_{i1}^{S}z_{i1}+m_{i2}^{S}z_{i2}+m_{i3}^{S}z_{i3}\right]\\
\mbox{s.t. } & \frac{(1,H_{it}^{\prime})\beta_{t}}{C_{it}}<z_{it}\leq1+\frac{(1,H_{it}^{\prime})\beta_{t}}{C_{it}}\ \mbox{for }i=1,\ldots,n\mbox{ and }t=1,2,\\
 & z_{i3}=z_{i1}z_{i2} \mbox{ for }i=1,\ldots,n,\\
 & \frac{1}{n}\sum_{t=1}^{2} \sum_{i=1}^{n}K_{tb}z_{it}\leq C_{b}+\alpha_{n}\ \mbox{for }b=1,\ldots,B\mbox{ and }t=1,2,
\end{align*}
where $C_{it}$ are constants that should satisfy $C_{it}>\sup_{\beta_{t}\in \MB_{t}}\left|(1,H_{it}^{\prime})\beta_{t}\right|$.
The linear constraints in the last line correspond to the budget/capacity constraints.


\end{document}